\documentclass[reqno,hidelinks]{amsart}

\usepackage{a4,amsmath,amssymb,amsthm}
\usepackage[top=1in, bottom=1in, left=1.25in, right=1.25in]{geometry}
\usepackage{xfrac}

\usepackage{leftidx}
\usepackage{inputenc}

\usepackage{amsaddr}

\usepackage[T1]{fontenc}

\usepackage{leftidx}
\usepackage{tikz}

\usepackage{hyperref}
\newcounter{mnotecount}[section]

\usepackage{color}

\usepackage{url}

\numberwithin{equation}{section}

\usepackage[mathscr]{eucal}

\usepackage{graphicx}
\graphicspath{pic/}
\def\nablaslash{\mbox{$\nabla \mkern -13mu /$ \!}}

\newcommand{\half}{\tfrac{1}{2}}         %

\newcommand{\veps}{\varepsilon}

\newcommand{\inttau}{\int_{\mathcal{D}(0,\tau)}}
\newcommand{\R}{r^2+a^2}
\newcommand{\PR}{r^3-3Mr^2+a^2r+a^2M}

\newcommand{\KDelta}{\Delta}

\newcommand{\Ls}{\mathbf{L}_s}

\newcommand{\curlV}{\mathcal{V}}
\newcommand{\curlY}{\mathcal{Y}}
\newcommand{\intRinfty}{\int_{\mathcal{D}(0,\tau)\cap [R,\infty)}}

\newcommand{\intcut}{\int_{\mathcal{D}(0,\tau)\cap [R-1,R)}}

\theoremstyle{plain}
\newtheorem{thm}{Theorem}

\newtheorem{lemma}[thm]{Lemma}

\newtheorem{definition}[thm]{Definition}
\newtheorem{prop}[thm]{Proposition}

\newtheorem{remark}[thm]{Remark}

\title{Uniform Energy Bound and Morawetz Estimate for Extreme Components of Spin Fields in the Exterior of a Slowly Rotating Kerr Black Hole I: Maxwell Field}

\author[S. Ma]{Siyuan Ma}
\email{siyuan.ma@sorbonne-universite.fr}
\address{
Laboratoire Jacques-Louis Lions,
Campus Jussieu,
Sorbonne Université,
4 place Jussieu,
 75005 Paris, France}

\begin{document}

\allowdisplaybreaks

\begin{abstract}
This first part of the series treats the Maxwell equations in the exterior of a slowly rotating Kerr black hole. By performing a first-order differential operator on each extreme Newman--Penrose (N--P) scalar in a Kinnersley tetrad, the resulting equation and the Teukolsky master equation for the extreme N--P component
are both in the form of an inhomogeneous \textquotedblleft{spin-weighted Fackerell--Ipser equation\textquotedblright} (SWFIE) and constitute a weakly coupled system.  We first prove energy estimate and  integrated local energy decay (Morawetz) estimate for this type of inhomogeneous SWFIE following the method in \cite{dafermos2010decay}, and then utilize these estimates to achieve both a uniform bound of a positive definite energy and a Morawetz estimate
for the coupled system of each extreme N--P component. The same type of estimates for the regular extreme N--P components defined in the regular Hawking--Hartle tetrad is also proved. The hierarchy here is generalized in our second part \cite{Ma17spin2Kerr} of this series to treat the extreme components of linearized gravity.
\end{abstract}

\maketitle

\section{Introduction}
The Kerr family of spacetimes has metrics solving the vacuum Einstein equations (VEE)
\begin{align}
\mathrm{R}_{\mu\nu}=0,
\label{eq:EinsteinVacuumEq}
\end{align}
with $\mathrm{R}_{\mu\nu}$ being the Ricci curvature tensor of the spacetime. An important step toward establishing
the validity of weak cosmic censorship
conjecture is to prove the Kerr black hole stability conjecture,
i.e., to show that the Kerr family of spacetimes is stable as solutions to VEE against small perturbations of initial data. Before approaching this full nonlinear problem, the scalar wave equation, the Maxwell equations and then some proper linearization of VEE are a sequence of models with increasing accuracy for the nonlinear dynamics.

In this paper, we consider the Maxwell equations for a real two-form $\mathbf{F}_{\alpha\beta}$:
\begin{align}\label{eq:MaxwellEqs}
\nabla^{\alpha}\mathbf{F}_{\alpha\beta}&=0 & \nabla_{[\gamma}\mathbf{F}_{\alpha\beta]}&=0
\end{align}
on a slowly rotating Kerr background,
and prove both a uniform bound of a positive definite energy and a Morawetz estimate
for the extreme Newman--Penrose (N--P) scalars.

\subsection{Kerr background}
 The subextremal Kerr family of spacetimes $(\mathcal{M},g_{M,a})$ ($|a|< M$), in Boyer--Lindquist (B--L) coordinates $(t,r,\theta,\phi)$ \cite{boyer:lindquist:1967}, has the metric
\begin{align}\label{eq:KerrMetricBoyerLindquistCoord}
g_{M,a}= & -\left(1-\tfrac{2Mr}{\Sigma} \right) dt^2 -\tfrac{2Mar \sin^2\theta}{\Sigma}(dt d\phi + d\phi dt) \nonumber\\
 & + \tfrac{\Sigma}{\Delta} dr^2 + \Sigma d\theta^2 +\tfrac{\sin^2\theta}{\Sigma} \left[(r^2+a^2)^2 -a^2\Delta \sin^2\theta\right]d\phi^2,
\end{align}
where
\begin{align}\label{eq:kerrfunctions}
\Delta(r)&= r^2 -2Mr +a^2 \   \   \   \   \   \text{    and    }  & \Sigma(r,\theta) = r^2+a^2 \cos^2\theta .
\end{align}
A Kerr spacetime is parameterized by its mass $M$ and angular momentum per mass $a$ and describes a rotating, stationary (with $\partial_t$ Killing), axisymmetric (with $\partial_{\phi}$ Killing), asymptotically flat vacuum black hole. Setting $a=0$ recovers the spherically symmetric Schwarzschild metric.
The function $\Delta(r)$ has two zeros
\begin{align}
r_+&=M+\sqrt{M^2-a^2} \   \   \   \   \   \text{    and    }  &
r_-=M-\sqrt{M^2-a^2},
\end{align}
which correspond to the locations of event horizon $\mathcal{H}$ and Cauchy horizon, respectively. We will constrain our considerations in the closure of the exterior region, or in another way, in the domain of outer communication (DOC)
\begin{equation}\label{def:DOC}
\mathcal{D}=\overline{\{(t,r,\theta,\phi)\in \mathbb{R}\times (r_+,\infty)\times \mathbb{S}^2\}}.
\end{equation}
As mentioned in Section \ref{sect:MainTheorems}, we will focus only on the future development by symmetry; hence only the future part $\mathcal{H}^+$ of the event horizon will be of interest for us. In what follows, whenever we say \textquotedblleft{a slowly rotating Kerr spacetime,\textquotedblright} it should always be understood as the DOC of a Kerr spacetime $(\mathcal{M},g_{M,a})$ with $|a|/M\ll1$ sufficiently small.

It is useful to introduce the tortoise coordinate system $(t,r^*,\theta,\phi)$\footnote{This is called as Regge--Wheeler \cite{ReggeWheeler1957} coordinate system when on Schwarzschild.}, with $r^*$ defined by:
\begin{align}
\tfrac{dr^*}{dr}=\tfrac{r^2+a^2}{\Delta},\   \   \   \   \   \    r^*(3M)=0.
\end{align}

The B--L coordinate system fails to extend across the future event horizon $\mathcal{H}^+$ due to the singularity in the metric coefficients in \eqref{eq:KerrMetricBoyerLindquistCoord}. Instead,
we define an ingoing Kerr coordinate system $(v,r,\theta,\tilde{\phi})$ \footnote{This is also called as \textquotedblleft{ingoing Eddington--Finkelstein coordinate system.\textquotedblright}  } which is regular on $\mathcal{H}^+$:
\begin{equation}\label{def:IngoingEddiFinkerCoord}
\left\{
  \begin{array}{ll}
    dv=dt +dr^*, \\
    d\tilde{\phi}=d\phi +a(r^2 +a^2)^{-1}dr^*,\\
    r=r,\\
    \theta=\theta.\\
  \end{array}
\right.
\end{equation}

We finally define a global Kerr coordinate system $(t^*,r, \theta,\phi^*)$, via gluing the coordinate system $(\vartheta=v-r, r, \theta,\tilde{\phi})$ near horizon with the B--L coordinate system $(t,r,\theta,\phi)$ away from horizon smoothly, by
\begin{equation}\label{def:globalkerrcoord}
\left\{
  \begin{array}{ll}
    t^*=t+\chi_{1}(r)\left(r^*(r)-r-r^*(r_0)+r_0\right),   \\
    \phi^*=\phi+\chi_{1}(r)\acute{\phi}(r)\ \ \text{mod}\ 2\pi, \ \ d\acute{\phi}/dr=a/\Delta,\\
    r=r,  \\
\theta=\theta.
  \end{array}
\right.
\end{equation}
Here, the smooth cutoff function $\chi_1(r)$,  which equals to $1$ in $[r_+,M+r_0/2]$ and identically vanishes for $r\geq r_0$ with $r_0(M)$ fixed in the red-shift estimate Proposition \ref{prop:RedShiftEstiInhomoSWRWE}, is chosen such that the initial hypersurface
\begin{equation}
\Sigma_{0}=\left\{(t^*,r,\theta,\phi^*)|t^*=0\right\}\cap \mathcal{D}
\end{equation}
is a spacelike hypersurface with
\begin{equation}\label{eq:initialtimehypersurfacegradientmodular}
c(M)\leq-g(\nabla t^*,\nabla t^*)|_{\Sigma_0}\leq C(M),
\end{equation}
$c(M)$ and $C(M)$ being two universal positive constants.
We notice that
\begin{equation}\label{eq:VFieldTandPhi}
\partial_{t^*}=\partial_t\triangleq T\ \ \text{and}\ \  \partial_{\phi^*}=\partial_{\tilde{\phi}}=\partial_{\phi},
\end{equation}
and denote $\varphi_{\tau}$ to be the $1$-parameter family of diffeomorphisms generated by $T$. Define constant time hypersurface
\begin{equation}\label{def:constanttimeHypersurface}
\Sigma_{\tau}=\varphi_{\tau}\left(\Sigma_0\right)=\left\{(t^*,r,\theta,\phi^*)|t^*=\tau\right\}\cap \mathcal{D}.
 \end{equation}
 It is a spacelike hypersurface satisfying \eqref{eq:initialtimehypersurfacegradientmodular}, and in particular, for $r\leq M+r_0/2$, we have
\begin{equation}
-g(\nabla t^*,\nabla t^*)|_{\Sigma_{\tau}}=1+2Mr/\Sigma.
\end{equation}
For any $0\leq \tau_1< \tau_2$, we use the notations
\begin{align*}
\mathcal{D}(\tau_1, \tau_2)=\bigcup_{\tau\in [\tau_1, \tau_2]}\Sigma_{\tau},\quad \text{and} \quad
\mathcal{H}^+(\tau_1, \tau_2)=\partial\mathcal{D}(\tau_1, \tau_2) \cap \mathcal{H}^+.
\end{align*}
The reader may find the Penrose diagram  Figure 1 useful.
\begin{figure}[!h]
\centering
\raisebox{-0.5\height}{
\includegraphics{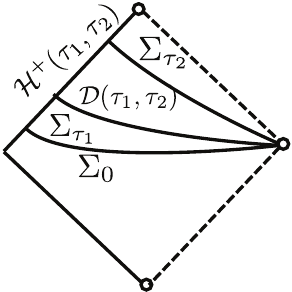}}
\caption{Penrose diagram}
\label{fig:penrosediagram}
\end{figure}

The volume form of the manifold is
\begin{align}\label{eq:VolumeForm}
d{\text{Vol}}_{\mathcal{M}}=
\left\{
  \begin{array}{ll}
    \Sigma dtdr\sin\theta d\theta d\phi & \text{in B--L coordinates,} \\
\Sigma dt^*dr\sin\theta d\theta d\phi^* & \text{in global Kerr coordinates},
  \end{array}
  \right.
\end{align}
and the volume form of the hypersurface $\Sigma_{\tau} (\tau\geq0)$ is
\begin{equation}\label{def:volumeformhypersurface}
d\text{Vol}_{\Sigma_{\tau}}=\Sigma dr \sin \theta d\theta d\phi^* \ \  \text{in global Kerr coordinates}.
\end{equation}
Unless otherwise specified, we will always suppress these volume forms  in this paper when the integral is over a spacetime region or a $3$-dimensional submanifold of $\Sigma_{\tau}$.

\subsection{Maxwell equations in the Newman--Penrose formalism and Teukolsky master equation}

Utilizing the N--P formalim  \cite{newmanpenrose62, newmanpenrose63errata} to project the electromagnetic field $\mathbf{F}_{\alpha\beta}$ onto the Kinnersley null tetrad \cite{Kinnersley1969tetradForTypeD} $(l,n,m,\overline{m})$ written in B--L coordinates:
\begin{align}\label{eq:Kinnersley tetrad}
l^\mu &= \tfrac{1}{\Delta}(r^2+a^2 , \Delta , 0 , a), \notag\\
n^\mu &= \tfrac{1}{2\Sigma} (r^2+a^2 , - \Delta , 0 , a), \notag\\
m^\mu &= \tfrac{1}{\sqrt{2} \bar{\rho}}\left(i a \sin{\theta},0 , 1, \tfrac{i}{\sin{\theta}}\right),
\end{align}
and $\overline{m}^{\mu}$ being the complex conjugates of $m^{\mu}$,
the N--P components for the electromagnetic perturbations are given by
\begin{equation}\label{eq:MaxwellNPcomponentswithnosuperscript}
\Phi_0 = \mathbf{F}_{\mu\nu} l^\mu m^\nu ,\ \Phi_1=\half \mathbf{F}_{\mu\nu}(l^{\mu}n^{\nu}+\overline{m}^{\mu}m^{\nu}),\  \Phi_2 = \mathbf{F}_{\mu\nu} \overline{m}^\mu n^\nu.
\end{equation}
Here, $\bar{\rho}$ is the complex conjugate of $\rho = r- i a \cos{\theta}$.
The extreme components $\Phi_0$ and $\Phi_2$ are the \textquotedblleft{ingoing and outgoing radiative parts,\textquotedblright}
 and are invariant under gauge transformations.
Define the spin $s=\pm1$ components\footnote{These scalars satisfy the Teukolsky master equation and differ with the ones originally used in \cite{Teukolsky1973I} by a rescaling of $\Delta$ or $\Delta^{-1}$ and a constant rescaling.}
\begin{equation}
\begin{split}
\psi_{[+1]}= 2^{-1/2}\Delta \Phi_0  \ \ \text{ and } \ \ & \psi_{[-1]}=\sqrt{2}\Delta^{-1}\rho^2\Phi_2,
\end{split}
\label{eq:spinsfields}
\end{equation}
and a middle component
\begin{align}
\label{def:middlecompMaxwell}
\psi_{[0]}={}&\rho^2\Phi_1.
 \end{align}
Denote the future-directed ingoing and outgoing principal null vector fields in B--L coordinates
\begin{align}\label{def:VectorFieldYandV}
Y&\triangleq \tfrac{2\Sigma}{\Delta}n^{\mu}\partial_{\mu}
=\tfrac{(r^2+a^2)\partial_t +a\partial_{\phi}}{\Delta}-\partial_r, \ &\ V&\triangleq l^{\mu}\partial_{\mu}= \tfrac{(r^2+a^2)\partial_t+
a\partial_{\phi}}{\Delta}+\partial_r.
\end{align}
The full system of Maxwell equations relating these three components is
\begin{subequations}\label{eq:TSIsSpin1Kerr}
\begin{align}
\label{eq:TSIsSpin1Kerr Angular0With1}
\sqrt{2}\bar{\rho}m^{\mu}\partial_{\mu}\psi_{[0]}={}& \rho^2 Y\left({\rho}^{-1}{\psi_{[+1]}}\right),\\
\label{eq:TSIsSpin1Kerr Angular0With-1}
\sqrt{2}\rho \bar{m}^{\mu}\partial_{\mu}\psi_{[0]}={}&{\rho}^2 V\left(\rho^{-1}{\Delta}\psi_{[-1]}\right),\\
\label{eq:TSIsSpin1Kerr Radipal0With1}
\frac{\Delta}{\rho^2}l^{\mu}\partial_{\mu} \psi_{[0]}=
{}&\Big(\partial_{\theta}-\frac{i}{\sin\theta}\partial_{\phi}
-ia\sin\theta\partial_t
+\cot\theta\Big)\left(\rho^{-1}{\psi_{[+1]}}\right),\\
\label{eq:TSIsSpin1Kerr Radial0With-1}
\frac{2\Sigma}{\Delta}n^{\mu}\partial_{\mu}\psi_{[0]}={}&\rho^2\Big(\partial_{\theta}+\frac{i}{\sin\theta}\partial_{\phi}
+ia\sin\theta\partial_t+\cot\theta\Big)
\left(\rho^{-1}{\psi_{[-1]}}\right).
\end{align}
\end{subequations}

In \cite{Teu1972PRLseparability}, the author derived the decoupled equations on Kerr background for the spin $s=\pm1$ components from the system \eqref{eq:TSIsSpin1Kerr}
and found these decoupled equations are separable and governed by a single master equation--\emph{Teukolsky Master Equation} (TME), which in B--L coordinates is given by
\begin{align}\label{eq:TME}
& -\left[\tfrac{(r^2+a^2)^2}{\Delta} -a^2 \sin^2{\theta} \right] \tfrac{\partial^2 \psi_{[s]}}{\partial t^2} - \tfrac{4Mar}{\Delta} \tfrac{\partial^2 \psi_{[s]}}{\partial t \partial \phi}-\left[\tfrac{a^2}{\Delta} -\tfrac{1}{\sin^2{\theta}} \right] \tfrac{\partial^2 \psi_{[s]}}{\partial \phi^2}   \notag\\
&  +\Delta^{s} \tfrac{\partial}{\partial r} \left( \Delta^{-s+1} \tfrac{\partial \psi_{[s]}}{\partial r} \right) + \tfrac{1}{\sin{\theta}} \tfrac{\partial}{\partial \theta} \left( \sin{\theta} \tfrac{\partial \psi_{[s]}}{\partial \theta}\right) +2s \left[ \tfrac{a(r-M)}{\Delta} + \tfrac{i \cos{\theta}}{\sin^2{\theta} } \right] \tfrac{\partial \psi_{[s]}}{\partial \phi} \notag\\
&  +2s\left[ \tfrac{M(r^2-a^2)}{\Delta} -r -ia \cos{\theta} \right] \tfrac{\partial \psi_{[s]}}{\partial t}- (s^2 \cot^2{\theta} +s) \psi_{[s]} = 0 .
\end{align}

However, the Kinnersley tetrad is singular in ingoing Kerr coordinates on $\mathcal{H}^+$, which means the afore-defined N--P components are not all regular there. Therefore, we perform a null rotation to the Kinnersley null tetrad by
\begin{equation}\label{null rotation}
\left\{
  \begin{array}{ll}
    l \rightarrow \tilde{l}=\Delta/(2\Sigma)\cdot l, \\
    n \rightarrow \tilde{n}=(2\Sigma)/\Delta \cdot n,\\
    m \rightarrow m,
  \end{array}
\right.
\end{equation}
and the resulting tetrad $(\tilde{l},\tilde{n},m,\overline{m})$, called as the Hawking--Hartle (H--H) tetrad \cite{HHtetrad72}, if written in the ingoing Kerr coordinates $(v,r,\theta,\tilde{\phi})$, is
\begin{align}
\tilde l^\mu &= \tfrac{1}{2\Sigma}(2(r^2+a^2), \Delta , 0 , 2a), \notag\\
\tilde n^\mu &= (0 , - 1 , 0 , 0), \notag\\
m^\mu &= \tfrac{1}{\sqrt{2} \bar{\rho}}\left(i a \sin{\theta},0 , 1, \tfrac{i}{\sin{\theta}}\right).
\end{align}
It is clear that this H--H tetrad
 is indeed regular up to and including $\mathcal{H}^+$ in ingoing Kerr coordinates, and hence also in global Kerr coordinates.
The regular N--P components $\widetilde{\Phi_i}$ $(i=0,1,2)$ of the Maxwell field tensor $\mathbf{F}$, defined as in \eqref{eq:MaxwellNPcomponentswithnosuperscript} but with respect to H-H tetrad, are then
\begin{align}\label{def:regularNPComps}
\left\{
  \begin{array}{ll}
    \widetilde{\Phi_0}(\mathbf{F})=\mathbf{F}_{\mu\nu} \tilde{l}^\mu m^\nu
    =\tfrac{1}{2\Sigma}\psi_{[+1]}, \\
   \widetilde{\Phi_1}(\mathbf{F})=\half \mathbf{F}_{\mu\nu}(\tilde{l}^{\mu}\tilde{n}^{\nu}+\overline{m}^{\mu}m^{\nu})=\Phi_1=\rho^{-2}\psi_{[0]},\\
    \widetilde{\Phi_2}(\mathbf{F})=\mathbf{F}_{\mu\nu} \overline{m}^\mu \tilde{n}^\nu
    =\tfrac{2\Sigma}{\rho^2}\psi_{[-1]}.\\
  \end{array}
\right.
\end{align}
The complex scalars $\widetilde{\Phi_0}$ and $\widetilde{\Phi_2}$ will be of our main concern in this paper.

\subsection{Coupled systems}

The evolution equations for $\psi_{[+1]}$ and $\psi_{[-1]}$, by TME \eqref{eq:TME}, are
\begin{subequations}\label{eq:TME0orderPosiandNega}
\begin{align}\label{eq:TME0order}
&\Sigma\Box_g \psi_{[+1]}+2i\left(\tfrac{\cos\theta}{\sin^2 \theta}\partial_{\phi}-a\cos \theta \partial_t\right)\psi_{[+1]}-(\cot^2\theta +1)\psi_{[+1]} =-2Z\psi_{[+1]},\\
\label{eq:TME0orderNega}
&\Sigma\Box_g \psi_{[-1]} -2i\left(\tfrac{\cos\theta}{\sin^2 \theta}\partial_{\phi}-a\cos \theta \partial_t\right)\psi_{[-1]}-(\cot^2\theta -1)\psi_{[-1]}=2Z\psi_{[-1]}.
\end{align}
\end{subequations}
Here, $Z=(r-M)Y-2r\partial_t$.
We define two first order differential operators
\begin{align}\label{def:curlYandVops}
\mathcal{Y}(\cdot)={}&\sqrt{r^2+a^2}Y(\sqrt{r^2+a^2}\cdot),
&
\mathcal{V}(\cdot)={}&\sqrt{r^2+a^2}V(\sqrt{r^2+a^2}\cdot),
\end{align}
and construct from $\psi_{[+1]}$ and $\psi_{[-1]}$ the variables
\begin{subequations}\label{eq:DefOf Phi^0^1}
\begin{align}
\label{eq:DefOf Phi^0^1PosiSpin}
\phi^0_{+1}&=\psi_{[+1]}/(r^2+a^2),&
\phi^1_{+1}&=\mathcal{Y}\phi^0_{+1}
\end{align}
and
\begin{align}
\label{eq:DefOf Phi^0^1NegaSpin}
\phi^0_{-1}&=\Delta/(r^2+a^2) \psi_{[-1]},&
\phi^1_{-1}&=\mathcal{V}\phi^0_{-1}.
\end{align}
\end{subequations}
One should notice that though it is $\Delta/(r^2+a^2) V$  but not $V$ which is a regular vector field on $\mathcal{H}^+$, the variable $\phi_{-1}^1$ is indeed smooth on $\mathcal{H}^+$ if the regular N--P scalar $\widetilde{\Phi_2}$ is.

Upon doing the rescalings as in the definitions of $\phi^0_{\pm 1}$, one can gather the resulting $\partial_r$ derivative term and the $Z$-derivative term together and rewrite them as a linear combination of $\phi_{\pm}^1$ and up to first-order derivative terms with $a$-dependent coefficients. The definitions of $\phi_{\pm 1}^1$ are such that their governing equations are \textquotedblleft{spin-weighted Fackerell--Ipser\textquotedblright} equations with inhomogeneous terms weakly coupled to $\phi_{\pm 1}^0$ and $\partial_{\phi}\phi_{\pm 1}^0$. Specifically, from the commutator relations in Appendix \ref{appx:commutatorrelation} and noting $\partial_r\big(\tfrac{\Delta}{(\R)^2}\big)=-\tfrac{2\PR}{(\R)^3}$,
these variables satisfy the following equations
\begin{subequations}\label{eq:ReggeWheeler Phi^01Kerr}
\begin{align}
\label{eq:ReggeWheeler Phi^0Kerr}
\mathbf{L}_{+1}\phi^0_{+1}
=&F^{0}_{+1}
=\tfrac{2(\PR)}{(\R)^2}\phi^1_{+1}
-\tfrac{4ar}{\R}\partial_{\phi}\phi^0_{+1},\\
\label{eq:ReggeWheeler Phi^1Kerr}
\mathbf{L}_{+1}\phi^1_{+1} =&F^1_{+1}=-\tfrac{2a^2(\PR)}{(\R)^2}\phi^0_{+1}
-\tfrac{2a(r^2-a^2)}{\R}\partial_{\phi}\phi^0_{+1},
\end{align}
\end{subequations}
and
\begin{subequations}\label{eq:ReggeWheeler Phi^01KerrNega}
\begin{align}\label{eq:ReggeWheeler Phi^0KerrNega}
\mathbf{L}_{-1}\phi^0_{-1}
=&F^{0}_{-1}=-\tfrac{2(\PR)}{(\R)^2}\phi^1_{-1}
+\tfrac{4ar}{\R}\partial_{\phi}\phi^0_{-1},\\
\label{eq:ReggeWheeler Phi^1KerrNega}
\mathbf{L}_{-1}\phi^1_{-1} =&F^{1}_{-1}=\tfrac{2a^2(\PR)}{(\R)^2}\phi^0_{-1}
-\tfrac{2a(r^2-a^2)}{\R}\partial_{\phi}\phi^0_{-1},
\end{align}
\end{subequations}
respectively.
Here, the subscript $+1$ or $-1$ indicates the spin weight $s=\pm1$, and the operator $\mathbf{L}_s$, defined by
\begin{align}\label{def:SWRWoperator}
\mathbf{L}_s=\Sigma \Box_g+2is\left(\tfrac{\cos\theta}{\sin^2 \theta}\partial_{\phi}-a\cos \theta \partial_t\right)-\big(\cot^2 \theta+\tfrac{r^4-2Mr^3+6a^2Mr-a^4}{(\R)^2}\big),
\end{align}
is called as \emph{\textquotedblleft{spin-weighted Fackerell--Ipser operator\textquotedblright}} in this paper.
In the discussions below, the above governing equations for $\phi^1_{s}$ and $\phi^0_{s}$ will be viewed as a linear wave system. Without confusion, we may always suppress the subscripts and simply write as $\phi^1$ and $\phi^0$.
\begin{remark}
The application of the first-order differential operators $\mathcal{Y}$ and $\mathcal{V}$ to the extreme components is to make the nonzero boost weight vanishing. After making the substitutions $\partial_t\leftrightarrow -i\omega$, $\partial_{\phi}\leftrightarrow im$, and separating the operator $\mathbf{L}_s$, the angular part is the spin-weighted spheroidal harmonic operator of Teukolsky angular equation, while the radial operator is the sum of the radial part of the rescaled scalar wave operator $\Sigma \Box_g$ and $((\R)^2-(r^4-2Mr^3+6a^2Mr-a^4))/(\R)^2$
and reduces to the radial operator for Fackerell--Ipser (F--I) equation \cite{fackerell:ipser:EM} when on Schwarzschild background $(a=0)$. See more details in Section \ref{sect:outlineproof} for Schwarzschild case and Section \ref{sect:SeparateAngAndRadialEqs} for Kerr case.

We note that in the non-static Kerr case $(a\neq 0)$, the classical F--I operator in \cite{fackerell:ipser:EM} has an imaginary zeroth-order term in the potential, thus being quite different from the operator $\mathbf{L}_s$ here in which the imaginary coefficients are accompanied by some first-order derivatives. This is the main difference which enables us in this paper to not introduce fractional derivative operators as in \cite{larsblue15Maxwellkerr} which treats the classical F--I equation.
\end{remark}

\subsection{Notations of energy norms, derivatives and Morawetz densities}\label{sect:EnergynormsDeris}
For any complex-valued smooth function $\psi: \mathcal{M}\rightarrow \mathbb{C}$ with spin weight $s$, define in global Kerr coordinates for any $\tau\geq0$ that
\begin{equation}
 {E}_{\tau}(\psi)= \int_{\Sigma_{\tau}}|\partial\psi|^2,
\end{equation}
and in ingoing Kerr coordinates for any $\tau_2>\tau_1\geq 0$ that
\begin{equation}
 {E}_{\mathcal{H}^+(\tau_1,\tau_2)}(\psi)= \int_{\mathcal{H}^+(\tau_1,\tau_2)}(|\partial_v\psi|^2
 +|\nablaslash\psi|^2)r^2dv\sin\theta d\theta d\tilde{\phi}.
\end{equation}
Note that we adopt the notation
\begin{equation}\label{eq:ModuloSquareofDeris}
|\partial\psi(t^*,r,\theta,\phi^*)|^2=|\partial_{t^*}\psi|^2
+|\partial_r\psi|^2+|\nablaslash\psi|^2,
\end{equation}
and $\nablaslash$ used here are not the standard covariant angular derivatives $\check{\nablaslash}$ on sphere $\mathbb{S}^2(t^*,r)$, but the spin-weighted version of them, i.e. $\nablaslash$ could be any one of $\nablaslash_i$ $(i=1,2,3)$ defined by
\begin{subequations}\label{SpinWeightedAngularDerivaBasisOnSphere}
\begin{align}
r\nablaslash_1&=r\check{\nablaslash}_1-\tfrac{is\cos\phi}{\sin\theta}=-\sin\phi\partial_{\theta}-
\tfrac{\cos\phi}{\sin\theta}\left(\cos\theta\partial_{\phi^*}+is\right),\\
r\nablaslash_2&=r\check{\nablaslash}_2-\tfrac{is\sin\phi}{\sin\theta}=\cos\phi\partial_{\theta}
-\tfrac{\sin\phi}{\sin\theta}\left(\cos\theta\partial_{\phi^*}+is\right),\\
r\nablaslash_3&=r\check{\nablaslash}_3=\partial_{\phi^*}.
\end{align}
\end{subequations}
In global Kerr coordinates, we can express the modulus square of $\nablaslash\psi$ as
\begin{align}\label{def:nablaslashModuleSquare}
|\nablaslash\psi|^2=\sum_{i=1,2,3}\left|\nablaslash \psi\right|^2&=\frac{1}{r^2}\left(|\partial_{\theta}\psi|^2
+\left|\tfrac{\cos\theta\partial_{\phi^*}\psi+
is\psi}{\sin\theta}\right|^2
+|\partial_{\phi^*}\psi|^2\right)\notag\\
&=\frac{1}{r^2}\left(|\partial_{\theta}\psi|^2
+\left|\tfrac{\partial_{\phi^*}\psi+
is\cos\theta\psi}{\sin\theta}\right|^2
+|\psi|^2\right).
\end{align}
The same forms \eqref{SpinWeightedAngularDerivaBasisOnSphere} and \eqref{def:nablaslashModuleSquare} hold true in B--L coordinates and ingoing Kerr coordinates from \eqref{eq:VFieldTandPhi}. Hence, for convenience of calculations, we may always refer to these forms with $\partial_{\phi}$ in place of $\partial_{\phi^*}$ without confusion.

Define a function $\chi_{\text{trap}}(r)=(1-3M/r)^2(1-\eta_{[r^-_{\text{trap}}, r^+_{\text{trap}}]}(r))$ where $\eta_{[r^-_{\text{trap}}, r^+_{\text{trap}}]}(r)$ is the indicator function in the radius region bounded by $r^-_{\text{trap}}=2.9M<3M<r^+_{\text{trap}}=3.1M$.
We also define the weighted Morawetz density $\mathbb{M}_{\text{deg}}(\psi)$ which degenerates in the trapped region defined by
\begin{subequations}\label{def:Mpsi}
\begin{align}\label{def:mathbbMpsi}
\mathbb{M}_{\text{deg}}(\psi)={}\frac{|\partial_{r}\psi|^2}{r^{1+\delta}}
+\chi_{\text{trap}}(r)\left(\frac{|\partial_{t^*} \psi|^2}{r^{1+\delta}}+\frac{|\nablaslash\psi|^2}{r}\right).
\end{align}
See Theorem \ref{thm:MoraEstiAlmostScalarWave} for the definition of $\chi_{\text{trap}}(r)$.
$\mathbb{M}(\psi)$, a weighted Morawetz density non-degenerate in the trapped region, is given by
\begin{align}\label{def:MpsiposiS1}
&\mathbb{M}(\psi)=
\frac{|\partial_{r}\psi|^2}{r^{1+\delta}}+ \frac{|\partial_{t^*} \psi|^2}{r^{1+\delta}}+\frac{|\nablaslash\psi|^2}{r},
\end{align}
and the ones with a hat which do not have the extra $r^{-\delta}$ prefactor for the $r$ and $t^*$ derivatives are
\begin{align}\label{def:MpsinegaS1}
\widehat{\mathbb{M}}_{\text{deg}}(\psi)={}&
\frac{|\partial_{r}\psi|^2}{r}
+\chi_{\text{trap}}(r)\frac{|\partial_{t^*} \psi|^2+|\nablaslash\psi|^2}{r}
,\\
 \widehat{\mathbb{M}}(\psi)={}&\frac{|\partial\psi|^2}{r}.
\end{align}
\end{subequations}
The derivative $\partial^i \psi$, for a multi-index $i=(i_1,i_2,i_3,i_4,i_5)$ with $i_k\geq 0$ $(k=1,2,3,4,5)$ and any smooth function $\psi$ with spin weight $s$, is to be understood as
\begin{equation}\label{def:multiindexderi}
\partial^i \psi=\partial_{t^*}^{i_1}\partial_r^{i_2}\nablaslash_{1}^{i_3}
\nablaslash_{2}^{i_4}\nablaslash_{3}^{i_5}\psi.
\end{equation}

\subsection{Main theorems}\label{sect:MainTheorems}

All the results in this paper will only focus on the future time direction. By performing a symmetry $\psi_{[s]}(t,\phi)\rightarrow \Delta^s \psi_{[-s]}(-t,-\phi)$ which satisfy the same governing equations, one could conclude the corresponding estimates in the past time direction immediately.

\begin{thm}\label{thm:MoraEstiAlmostScalarWave}(Energy and Morawetz estimate for an inhomogeneous spin-weighted Fackerell--Ipser).
In the DOC of a slowly rotating Kerr spacetime $(\mathcal{M},g=g_{M,a})$, assume that a real two-form $\mathbf{F}_{\alpha\beta}$ is a regular\footnote{\textquotedblleft{Regular\textquotedblright} is defined in Definition \ref{def:regularandintegrable}.} solution to the Maxwell equations \eqref{eq:MaxwellEqs}, and $\psi=\phi^i_{s}$ is defined as in \eqref{eq:DefOf Phi^0^1} and satisfies an inhomogeneous \textquotedblleft{spin-weighted Fackerell--Ipser equation\textquotedblright} (SWFIE)
\begin{equation}\label{eq:spinweightedwaveeqAssumption}
\mathbf{L}_s\psi=F,
\end{equation}
which takes the form of any subequation in the linear wave system \eqref{eq:ReggeWheeler Phi^01Kerr} or \eqref{eq:ReggeWheeler Phi^01KerrNega} with $s=\pm1$ being the spin weight, $\mathbf{L}_s$ defined as in \eqref{def:SWRWoperator} and $F=F_s^i$.
Then, given any $0<\delta<1/2$ and $0<\mu_i\ll 1$, there exist universal constants $\veps_0=\veps_0(M)$ and $C=C(M,\delta,\Sigma_0)=C(M,\delta,\Sigma_{\tau})$ and a constant $C(\mu_i^{-1})$ such that for all $|a|/M\leq \veps_0$ and any $\tau\geq0$,
we have
\begin{align}\label{eq:MorawetEnergyEstimateforAlmostScalarWave}
\hspace{4ex}&\hspace{-4ex}
{E}_{\tau}(\psi)+{E}_{\mathcal{H}^+(0,\tau)}(\psi)
+\int_{\mathcal{D}(0,\tau)} \mathbb{M}_{\text{deg}}(\psi)\notag\\
\leq {}&C(\mu_i^{-1})\sum_{j=0,1}E_{0}(\phi_s^j)+C\mu_i^{-1}
\mathcal{E}(F)
\notag\\
&
+C\mu_i\sum_{j=0,1}\bigg(
\int_{\mathcal{D}(0,\tau)}\mathbb{M}_{\text{deg}}(\phi_s^j)
+\mathcal{E}(F_s^j)\bigg).
\end{align}
Here, the error term $\mathcal{E}(F)$ coming from the source $F$ takes the form of
\begin{align}\label{def:ErrorTermFromSourceF}
\mathcal{E}(F)=\bigg|\int_{\mathcal{D}(0,\tau)}
\frac{1}{\Sigma}\Re\left(F \partial_t\bar{\psi}\right)\bigg|
+\int_{\mathcal{D}(0,\tau)}r^{-3+\delta}|F|^2,
\end{align}
with $\Re(\cdot)$ denoting the real part.
\end{thm}

\begin{thm}\label{thm:EneAndMorEstiExtremeCompsNoLossDecayVersion2}(Energy and Morawetz estimates for regular extreme N--P components).
In the DOC of a slowly rotating Kerr spacetime $(\mathcal{M},g=g_{M,a})$, given any $0<\delta<1/2$ and any nonnegative integer $n$, there exist universal constants $\veps_0=\veps_0(M)$ and $C=C(M,\delta,\Sigma_0,n)=C(M,\delta,\Sigma_{\tau},n)$ such that for all $|a|/M\leq \veps_0$ and any regular solution $\mathbf{F}_{\alpha\beta}$ to the Maxwell equations \eqref{eq:MaxwellEqs},
the following estimates hold true for any $\tau\geq 0$:
\begin{itemize}
\item Define
\begin{align}
\label{def:varphi01positive}
&\varphi^0_0=r^{2-\delta}\widetilde{\Phi_0}, & \varphi^1_0&=rY(r\widetilde{\Phi_0}),
\end{align}
then
\begin{align}\label{eq:MoraEstiFinal(2)KerrRegularpsiBothSpinComp}
  \hspace{4ex}&\hspace{-4ex}\sum_{k=0,1}\left({E}_{\tau}(\varphi^k_0)
+{E}_{\mathcal{H}^+(0,\tau)}(\varphi^k_0)\right)
+\int_{\mathcal{D}(0,\tau)} \left(\mathbb{M}_{\text{deg}}(\varphi^1_0)
+\widehat{\mathbb{M}}(\varphi^0_0)\right)\notag\\
\leq{}&
C\sum_{k=0,1}{E}_{0}(\varphi^k_0).
\end{align}
\item
Define
\begin{align}
\label{def:varphi01negative}
&\varphi^0_2=\widetilde{\Phi_2}, & \varphi^1_2&=\frac{r\Delta}{\R}V(r\widetilde{\Phi_2}),
\end{align}
then\footnote{We have dropped the terms $\inttau \left|\nablaslash \varphi^0_2\right|^2$ and $\int_{\Sigma_{\tau}}r\left|\nablaslash \varphi^0_2 \right|^2$ on left hand side (LHS) of \eqref{eq:MoraEstiFinal(2)KerrRegularpsiBothSpinCompV2} and analogous terms on LHS of \eqref{eq:MoraEstiFinal(2)KerrRegularpsiBothSpinCompHighOrder2}. See Section \ref{sect:pfofmainthm}.}
\begin{align}\label{eq:MoraEstiFinal(2)KerrRegularpsiBothSpinCompV2}
  \hspace{4ex}&\hspace{-4ex}\sum_{k=0,1}\left({E}_{\tau}(\varphi^k_2)
+{E}_{\mathcal{H}^+(0,\tau)}(\varphi^k_2)\right)
+\int_{\mathcal{D}(0,\tau)} \left(\mathbb{M}_{\text{deg}}(\varphi^1_2)
+{\mathbb{M}}(\varphi^0_2)\right)\notag\\
\leq{}&
C\sum_{k=0,1}{E}_{0}(\varphi^k_2).
\end{align}
\end{itemize}
Moreover, we have the corresponding high order estimates
\begin{subequations}\label{eq:MoraEstiFinal(2)KerrRegularpsiBothSpinCompHighOrder}
\begin{align}\label{eq:MoraEstiFinal(2)KerrRegularpsiBothSpinCompHighOrder1}
&\sum_{|i|\leq n}\sum_{k=0,1}\left({E}_{\tau}(\partial^i\varphi^k_0)
+{E}_{\mathcal{H}^+(0,\tau)}(\partial^i\varphi^k_0)\right)\notag\\
&\ \ +\sum_{|i|\leq n}\int_{\mathcal{D}(0,\tau)} \left(\mathbb{M}_{\text{deg}}(\partial^i\varphi^1_0)
+\widehat{\mathbb{M}}(\partial^i\varphi^0_0)\right)
\leq
C\sum_{|i|\leq n}\sum_{k=0,1}{E}_{0}(\partial^i\varphi^k_0),\\
\label{eq:MoraEstiFinal(2)KerrRegularpsiBothSpinCompHighOrder2}
&\sum_{|i|\leq n}\sum_{k=0,1}\left({E}_{\tau}(\partial^i\varphi^k_2)
+{E}_{\mathcal{H}^+(0,\tau)}(\partial^i\varphi^k_2)\right)\notag\\
&\ \ +\sum_{|i|\leq n}\int_{\mathcal{D}(0,\tau)} \left(\mathbb{M}_{\text{deg}}(\partial^i\varphi^1_2)
+\mathbb{M}(\partial^i\varphi^0_2)\right)
\leq
C\sum_{|i|\leq n}\sum_{k=0,1}{E}_{0}(\partial^i\varphi^k_2).
\end{align}
\end{subequations}
\end{thm}
\begin{remark}
The trapping degeneracy in the trapped region can be manifestly  removed for the Morawetz densities $\mathbb{M}_{\text{deg}}(\partial^i\varphi^1_0)$ and $\mathbb{M}_{\text{deg}}(\partial^i\varphi^1_2)$ with $|i|\leq n-1$.
The energy and Morawetz estimate \eqref{eq:MoraEstiFinal(2)KerrRegularpsiBothSpinComp} is derived by applying Theorem \ref{thm:MoraEstiAlmostScalarWave} to $(\psi, F)=(\phi^0_s, F_{s}^0)$ and $(\psi, F)=(\phi^1_s, F_{s}^1)$ and is a single estimate at two levels of regularity, since $\phi^1$ is a null derivative of $\phi^0$. Thus, despite the well known trapped phenomenon, it is possible to prove a non-degenerate Morawetz estimate for $\phi^0$. However, the two levels of regularity must be treated simultaneously. In part, it is possible to close the two estimates simultaneously, because the right hand side (RHS) of \eqref{eq:ReggeWheeler Phi^0Kerr} and \eqref{eq:ReggeWheeler Phi^1Kerr} is at the same level of regularity, involving no derivatives of $\phi^1$ and one of $\phi^0$.
\end{remark}
\begin{remark}
An energy and Morawetz estimate for the middle component $\widetilde{\Phi_1}$ can be directly derived using the results in the above theorem and the system of Maxwell equations \eqref{eq:TSIsSpin1Kerr}. In fact, from system \eqref{eq:TSIsSpin1Kerr}, one is able to estimate $\widetilde{\Phi_1}$ up to a stationary part $C\rho^{-2}$, $C$ being a constant independent of the coordinates. Based on the recent result \cite{andersson2019stability} for linearized gravity on Kerr spacetimes, the approach of proving decay estimates for the spin $\pm 2$ components of linearized gravity on Kerr spacetimes therein can be extended here to obtain decay estimates for spin $\pm 1$ components of the Maxwell field which together with system \eqref{eq:TSIsSpin1Kerr} then imply estimates for the entire Maxwell field on Kerr spacetimes. This will be shown elsewhere.
\end{remark}

\subsection{Previous results}
The scalar wave equation in the DOC of vacuum black holes has been studied extensively in the last $15$ years. On a Schwarzschild background, uniform boundedness of scalar wave is first obtained in \cite{kaywald87Schw}, while a Morawetz \cite{morawetz1968time}-type multiplier, which is first introduced to black hole background in \cite{blue2003semilinear}, has been utilized in many works, such as \cite{bluesoffer09phase,dafrod09red}, to achieve Morawetz estimate (or integrated local energy decay estimate). In Kerr spacetime with $a\neq 0$, the fact that the Killing vector field $\partial_t$ fails to be globally timelike as in Schwarzschild case raises a difficulty in constructing a uniformly bounded positive energy. Moreover, the location where the null geodesics can be trapped is enlarged from $r=3M$ in Schwarzschild to a radius region in Kerr case. However, three independent, different approaches \cite{tataru2011localkerr,larsblue15hidden,dafermos2010decay} have been developed on slowly rotating Kerr background to achieve a uniform bound of a positive definite energy and Morawetz estimate.
Different pointwise decay estimates are also proved there. In particular, the separability of the wave equation or the complete integrability of the geodesic flow \cite{Carter1968Separability} is a point of crucial importance in these works. More recently, energy and decay estimates for the scalar field on the full subextremal Kerr $(|a|<M)$ backgrounds have been achieved in \cite{dafermos2016decay}.

Decay behaviors for Maxwell field have been proved in \cite{blue08decayMaxSchw} outside a Schwarzschild black hole, and on some spherically symmetric backgrounds or non-stationary asymptotically flat backgrounds in
\cite{metcalfe2014PWdecayMaxBH,sterbenz2015decayMaxSphSym}. All the works above focus on estimating the middle component which satisfies a decoupled, separable Fackerell--Ipser equation \cite{fackerell:ipser:EM} in a form similar to scalar wave equation, and then make use of these estimates to achieve a Morawetz estimate and decay estimates for the extreme components. In contrast with these works, the author in \cite{Fede2016MaxwellSchw} treats the extreme components satisfying the TME by applying some first-order differential operators to the extreme components, \footnote{These operators are the ones we used in this work if restricted to Schwarzschild case.} which then also satisfy the Fackerell--Ipser equation, while a superenergy tensor is constructed in \cite{andersson16decayMaxSchw} to yield a conserved energy current when contracted with $\partial_t$. In particular, the constructed superenergy tensor vanishes for the non-radiating Coulomb field. A decay estimate is also obtained in \cite{ghanem2014decayMaxSchw} under the assumption of a Morawetz estimate. We refer to the recent paper \cite{andersson16decayMaxSchw} for a more complete description of the literature in Maxwell equations on Schwarzschild background.

The method of linearizing VEE subject to metric perturbations was carried out for the Schwarzschild metric in \cite{ReggeWheeler1957,Vishveshwara1970stability,Zerilli1970evenparity}.
In these papers, the time and angular dependence can be easily separated out from the equations due to the metric being static and spherically symmetric. The resulting radial equations can be reduced to Regge--Wheeler equation governing the odd-parity perturbations and Zerilli equation governing the even-parity perturbations. In particular, these equations were derived later in \cite{moncrief74gravitational} without assuming any gauge conditions. The linear stability of Schwarzschild metric has been resolved recently in  \cite{DRG16linearstabSchw,Hung2017linearstabSchw}, with the former one starting from a Regge--Wheeler type equation satisfied by some scalar constructed from \emph{Chandrasekhar
transformation \cite{chandrasekhar1975linearstabSchw}} by applying some second order differential operator to the extreme component and the later one treating with Regge--Wheeler--Zerilli--Moncrief system. The energy and Morawetz estimates, as well as decay estimates, for this system are also obtained in \cite{Jinhua17LinGraSchw}.

For nonzero integer spin fields in the
DOC of a Kerr spacetime, only a few results about stability issue can be found in the literature. The only result for Maxwell equations we are aware of is given in \cite{larsblue15Maxwellkerr} on slowly rotating Kerr background, in which energy and Morawetz estimates for both the entire Maxwell system and the Fackerel--Ipser equation for the middle component are proved by introducing fractional derivative operators to treat the presence of an imaginary potential term in Fackerel--Ipser equation. The estimates therein enable the authors to prove a uniform bound of a positive definite energy and the convergence property of the Maxwell field to a stationary Coulomb field.
Turning to the extreme components, as mentioned already, they satisfy decoupled, separable TME \eqref{eq:TME}. Differential relations between the radial parts of the modes with extreme spin weights, as well as similar relations between the angular parts, are derived in \cite{starobinsky1973amplification,TeuPress1974III}, known as \textquotedblleft{\emph{Teukolsky--Starobinsky Identities}.\textquotedblright} In \cite{whiting1989mode}, it is shown that the TME admits no modes with frequency having positive imaginary part, or in another way, no exponentially growing mode solutions exist, by assuming some proper boundary conditions. This mode stability result is recently generalized to the case of real frequencies in \cite{shlapentokh-rothman2015} for scalar wave where the spin weight is $0$ and \cite{andersson2016mode} for any half-integer spin fields. We mention here the papers \cite{finster2016linear} which
discuss the stability problem for each azimuthal mode solution to TME, and \cite{kla2015globalstabwavemapKerr} concerning a toy model problem arising from the nonlinear stability of the Kerr
solution under small polarized axisymmetric perturbations.

\begin{flushleft}
\textit{References added:}
\end{flushleft}During the submission of this work, there has been great advance toward the nonlinear stability conjecture of the Kerr spacetimes. The linear stability of Schwarzschild has been reproved using instead (generalized) harmonic gauge in \cite{Hung2018linearwavegaugeodd,Johnson2018linear,Hung2019linearwavegaugeeven}.
A nontrivial nonlinear stability result of Schwarzschild spacetime under polarized axisymmetry is shown in \cite{klainermanszeftel2017Schw}.  See also \cite{Giorgi2019linearRNsmallcharge,Giorgi2019linearRNfullcharge}
for the linear stability results of the Reissner--Nordstr\"{o}m spacetimes. Energy and decay estimates for the spin $\pm 2$ TME on slowly rotating Kerr backgrounds are obtained in \cite{Ma17spin2Kerr,Dafermos2019TMEKerr}, and linear stability of slowly rotating Kerr spacetimes is shown in \cite{andersson2019stability,hafner2019linear}.

\subsection{Outline of the proof}\label{sect:outlineproof}
We start with the Schwarzschild case, which will outline the idea utilized for the slowly rotating Kerr case, and for simplicity, only the positive-spin component is considered here, the proof for the negative-spin component being the same.

In Schwarzschild spacetime ($a=0$), equation \eqref{eq:spinweightedwaveeqAssumption} reduces to \begin{equation}\label{eq:spinweightedwaveeqAssumptionSchw}
\Sigma \Box_g \psi+\tfrac{2is\cos\theta}{\sin^2 \theta}\partial_{\phi}\psi-\left(\cot^2 \theta+\tfrac{\Delta}{r^2}\right)\psi=F.
\end{equation}
Decompose the solution $\psi$ and the inhomogeneous term $F$ into
\begin{align}\label{eq:decompSWSH Schw}
\psi&=\sum_{m,\ell}\psi_{m\ell}(t,r)Y_{m\ell}^{s}(\cos\theta)e^{im\phi},
m\in \mathbb{Z},\\
\label{eq:decompSWSH Schw SourceTermF}
F&=\sum_{m,\ell}F_{m\ell}(t,r)Y_{m\ell}^{s}(\cos\theta)e^{im\phi},
m\in \mathbb{Z}.
\end{align}
Here, for each $m$, $\left\{Y_{m\ell}^{s}(\cos\theta)\right\}_{\ell}$ with $\min{\{\ell\}}=\max{(|m|,|s|)}$ are the eigenfunctions of the self-adjoint operator
\begin{align}\label{eq:SWSHOpe Schw}
\textbf{S}_m=\tfrac{1}{\sin\theta}\partial_{\theta}\sin\theta\partial_{\theta}
-\tfrac{m^2+2ms\cos\theta+s^2}{\sin^2\theta}
\end{align}
on $L^2(\sin\theta d\theta)$. These eigenfunctions, called as \textquotedblleft{\emph{spin-weighted spherical harmonics},\textquotedblright}  form a complete orthonormal basis on $L^2(\sin\theta d\theta)$ and have eigenvalues $-\Lambda_{m\ell}=-\ell(\ell+1)$ defined by
\begin{equation}
\label{eq:eigenvalueSWSHO}
\textbf{S}_mY_{m\ell}^{s}(\cos\theta)=-\Lambda_{m\ell}
Y_{m\ell}^{s}(\cos\theta).
\end{equation}
An integration by parts, together with a usage of Plancherel lemma and the orthonormality property of the basis $\left\{Y_{m\ell}^{s}(\cos\theta)e^{im\phi}\right\}_{m\ell}$, gives
\begin{align}\label{eq:IdenOfEigenvaluesAndAnguDeriSchw}
\sum_{m,\ell}\ell(\ell+1)|\psi_{m\ell}(t,r)|^2&
=\int_{0}^{\pi}\int_{0}^{2\pi}\left|\nablaslash \psi(t,r)\right|^2r^2\sin\theta d\phi d\theta.
\end{align}
Hence, the equation for $\psi_{m\ell}$ is now
\begin{align}\label{eq:SWRWReducedSchw}
r^4\Delta^{-1}
\partial_{tt}^2\psi_{m\ell}-\partial_r(\Delta\partial_r)\psi_{m\ell}
+(\ell(\ell+1)-2M/r)\psi_{m\ell}+F_{m\ell}=0.
\end{align}
In the case that the source term $F=0$, this is exactly the equation one gets after decomposing into spherical harmonics the solution to the classical Fackerell--Ipser equation \cite{fackerell:ipser:EM} on Schwarzschild:
\begin{equation}
\Sigma \Box_g \psi + 2M/r\psi=0.
\end{equation}

A proof for a nonnegative energy estimate is straightforward via performing integration by parts from multiplying \eqref{eq:SWRWReducedSchw} by $\partial_{t}\overline{\psi_{m\ell}}$, summing over $m$ and $\ell$, and  integrating in $(t^*,r)$-plane.
To obtain the Morawetz estimate, we choose a radial multiplier $(\hat{f}(r)\partial_r+\hat{q}(r))\overline{\psi_{m\ell}}$ with the choices of $\hat{f}$ and $\hat{g}$ taken from \cite{larsblue15hidden,Jinhua17LinGraSchw}. We postpone the detailed proof for this energy estimate and Morawetz estimate to Appendix \ref{appe:MoraFS1FIeqSchw} but emphasis here the fact that $\ell\geq |s|=1$ is essential in both proofs of Morawetz estimate and the nonnegativity of the energy. Together with a Morawetz estimate in large $r$ region in Section \ref{sect:MorawetzLarger} and red-shift estimate near horizon in Section \ref{sect:Redshift},
Theorem \ref{thm:MoraEstiAlmostScalarWave} is proved in the Schwarzschild case. Note that the energy and Morawetz estimates have been obtained in \cite{Fede2016MaxwellSchw} for the case of Schwarzschild background but under a double null foliation.

Now we apply Theorem \ref{thm:MoraEstiAlmostScalarWave} to equations \eqref{eq:ReggeWheeler Phi^01Kerr} of $\phi^0$ and $\phi^1$, which reduce in Schwarzschild spacetime to
\begin{subequations}\label{eq:ReggeWheeler Phi^01Schw}
\begin{align}
\label{eq:ReggeWheeler Phi^0}
\mathbf{L}_{+1}\phi^0 =&F^{0}_{+1}=\tfrac{2(r-3M)}{r^2}\phi^1,\\
\label{eq:ReggeWheeler Phi^1}
\mathbf{L}_{+1}\phi^1 =&0.
\end{align}
\end{subequations}
\begin{subequations}
Equation \eqref{eq:ReggeWheeler Phi^1} satisfied by $\phi^1$ is a decoupled equation of form \eqref{eq:spinweightedwaveeqAssumption} with $F=0$; hence we have from \eqref{eq:MorawetEnergyEstimateforAlmostScalarWave} that
\begin{align}\label{eq:Mora1order}
{E}_{\tau}(\phi^1)+{E}_{\mathcal{H}^+(0,\tau)}(\phi^1)
+\int_{\mathcal{D}(0,\tau)} \mathbb{M}_{\text{deg}}(\phi^1)\leq C{E}_{0}(\phi^1).
\end{align}
An application of estimate \eqref{eq:MorawetEnergyEstimateforAlmostScalarWave} to  \eqref{eq:ReggeWheeler Phi^0} implies
\begin{align}\label{eq:Mora0order}
{E}_{\tau}(\phi^0)+{E}_{\mathcal{H}^+(0,\tau)}(\phi^0)
+\int_{\mathcal{D}(0,\tau)} \mathbb{M}_{\text{deg}}(\phi^0)
\leq  C {E}_{0}(\phi^0)+C\mathcal{E}(F^{0}_{+1}).
\end{align}
\end{subequations}
Adding an $A$ multiple of estimate \eqref{eq:Mora1order} to \eqref{eq:Mora0order} with $A$ being a large constant to be chosen, we note that the left hand side (LHS) of the resulting inequality will bound over the integral of
\begin{equation}
|\partial_{r}\phi^0|^2/r^{1+\delta}
+|\phi^0|^2/r^{3}+|Y\phi^0|^2/r
\end{equation}
in the trapped region, hence over $\int_{\mathcal{D}(0,\tau)} c|\partial\phi^0|^2/r^{1+\delta}$ which does not degenerate in the trapped region through elliptic estimates,
$c=c(M)$ being a universal positive constant. Moreover, the error term $C\mathcal{E}(F^{0}_{+1})$ can be absorbed by the LHS through Cauchy--Schwarz inequality for $A$ large enough. This leads to an estimate
\begin{align}\label{eq:Mora01:Schw}
\sum_{i=0,1}\bigg({E}_{\tau}(\phi^i)
+{E}_{\mathcal{H}^+(0,\tau)}(\phi^i)
+\int_{\mathcal{D}(0,\tau)} \mathbb{M}_{\text{deg}}(\phi^i)\bigg)\lesssim {} \sum_{i=0,1}{E}_{0}(\phi^i).
\end{align}
Therefore, together with the Morawetz estimate \eqref{eq:MoraInftyr2minusdeltaphi0},  estimate \eqref{eq:MoraEstiFinal(2)KerrRegularpsiBothSpinComp} is proved for Schwarzschild case.

Let us turn to slowly rotating Kerr case. After applying a cutoff in time both to the future and to the past to the solution, performing Fourier transform in time and decomposing into spin-weighted spheroidal harmonics, we follow \cite{dafermos2010decay} to introduce the microlocal currents for the resulting radial equation and prove Theorem \ref{thm:MoraEstiAlmostScalarWave} by making use of the energy estimate and choosing frequency-dependent currents in different separated frequency regimes to achieve a Morawetz estimate.

To treat the system \eqref{eq:ReggeWheeler Phi^01Kerr}, we find that equation \eqref{eq:ReggeWheeler Phi^1Kerr} for $\phi^1$ is no longer decoupled due to its dependence on $\phi^0$.
When $|a|/M\ll 1$ is sufficiently small, however, the coupling effect with $\phi^0$ in \eqref{eq:ReggeWheeler Phi^1Kerr} is weak. This is the main observation for being able to go through the approach from Schwarzschild case to slowly rotating Kerr case.

There are some difficulties which arise from estimating the error terms $\mathcal{E}(F_{+1}^i)$ $(i=0,1)$ to overcome when applying Theorem \ref{thm:MoraEstiAlmostScalarWave} to system \eqref{eq:ReggeWheeler Phi^01Kerr} based on the above observation. One is the \textquotedblleft{trapping degeneracy\textquotedblright} in  estimates \eqref{eq:MorawetEnergyEstimateforAlmostScalarWave} for $\phi^0$ and $\phi^1$ in bounded trapped $r$ region. It turns out that this degeneracy in $\mathbb{M}_{\text{deg}}(\phi^0)$ can be removed as in Schwarzschild case after adding the separate estimates for $\phi^0$ and $\phi^1$ together,
and this enables us to estimate the error term $\mathcal{E}(F_{+1}^0)$ straightforwardly. However, the same argument does not hold for $\phi^1$, and the trapping degeneracy in $\mathbb{M}_{\text{deg}}(\phi^1)$ in the Morawetz estimate makes a direct estimate for $\mathcal{E}(F_{+1}^1)$ impossible. Estimates \eqref{eq:errortermnew}--\eqref{eq:ControlInTrappingRegionPosiSpin1} are utilized to control and absorb the error terms in bounded radius region by the left hand side (LHS) of a sum of estimates \eqref{eq:MorawetEnergyEstimateforAlmostScalarWave} for $\phi^0$ and $\phi^1$.

The other difficulty is in the large radius region, where we are met with a problem of an $r^{\delta}$ weight loss in $\mathbb{M}_{\text{deg}}(\phi^i)$ to control the error terms $\mathcal{E}(F_{+1}^1)$ directly.
Consider the error term $\int_{\mathcal{D}(0,t)}r^{-3+\delta}|F_{+1}^1|^2$ for example, which cannot be bounded by the LHS of the sum of estimates \eqref{eq:MorawetEnergyEstimateforAlmostScalarWave} for $\phi^0$ and $\phi^1$ because of the presence of $\partial_{\phi}\phi^0$ term in \eqref{eq:ReggeWheeler Phi^1Kerr}. The additional power $\delta$ in this error term is related to the degeneracy of Morawetz estimate in large $r$ region which, in our situation, is mainly exhibited by the $r^{\delta}$-weight loss in both $r$- and $t^*$-derivatives in \eqref{eq:MorawetEnergyEstimateforAlmostScalarWave}. It turns out that there is a damping effect on the RHS of \eqref{eq:TME0order} in the sense that the prefactor of $\partial_{t}$ derivative has the positive sign for $r$ large, and based on this fact, we improve the Morawetz estimate for $\phi^0_{+1}$ in large radius region such that the error term is manifestly controlled. The analysis for the spin $-1$ component is more subtle, in which
estimate \eqref{eq:estiphi0byphi1negafinal1} giving a better control over $|\partial_{\phi} \phi^0_{-1}|$ is derived by viewing definition \eqref{eq:DefOf Phi^0^1NegaSpin} of $\phi^{1}_{-1}$ as a transport equation of $\phi^0_{-1}$.

Here is an overview of this paper. In Section \ref{sect:PrelimandNotations}, we collect a (well-known) well-posedness theorem for a general linear wave system, give some definitions and introduce further notations.
A proof of Theorem \ref{thm:MoraEstiAlmostScalarWave} is given in Section \ref{Sect:ProofofTheoremOnSWRWKerr}.  Afterward, in Section \ref{sect:spin1case}, we make use of Theorem \ref{thm:MoraEstiAlmostScalarWave} and complete the proof of Theorem \ref{thm:EneAndMorEstiExtremeCompsNoLossDecayVersion2} for two regular extreme N--P components of Maxwell field.

\section{Preliminaries and Further Notations}\label{sect:PrelimandNotations}

\subsection{Well-posedness theorem}\label{sect:LWPandGlobalExistenceLinearWaveSystem}
We state here a well-posedness (WP) theorem for a general system of linear wave equations, cf. \cite[Chapter 3.2]{bar2007wave}. Due to the fact that smooth initial data which vanish near spatial infinity can be approached by smooth, compactly supported data, we restrict our considerations on initial data which are smooth and of compact support on initial hypersurface $\Sigma_0$.
\begin{prop}\label{prop:LWP LinearWaveEqSystem}
For any $1\leq n\in \mathbb{N}^+$ and $0\leq |a| <M$, let $\Sigma_0$ be an initial spacelike hypersurface defined as in \eqref{def:constanttimeHypersurface} in the DOC of a Kerr spacetime $(\mathcal{M},g_{M,a})$, and let $\varphi_0^i$, $\varphi_1^i$ be compactly supported smooth sections in the vector bundle $\mathbf{E}$ over the manifold $\mathcal{D}$, $i=1,2,\ldots,n$. Then  there exists a unique $\varphi=\left(\varphi^i\right)_{i=1,2,\ldots, n}$, with $\varphi^i\in C^{\infty}(D^+(\Sigma_0)\cap \mathcal{D},\mathbf{E})$, to the system of linear wave equations
\begin{equation}
\left\{
\begin{array}{ll}
\mathbf{L}_{\varphi}\varphi=0\\
\varphi^i|_{\Sigma_0}=\varphi_0^i,\ \ n_{\Sigma_0}^{\mu}\partial_{\mu}\varphi^i|_{\Sigma_0}=\varphi_1^i, \ \ \forall i=1,2,\ldots,n.\\
\end{array}
\right.
\end{equation}
Here, $D^+(\Sigma_0)$ is the future domain of dependence of $\Sigma_0$, $\mathbf{L}_{\varphi}$ is a linear wave operator for $\varphi$ and $n_{\Sigma_0}=n_{\Sigma_0}^{\mu}\partial_{\mu}$ is the future-directed unit normal vector field of initial hypersurface $\Sigma_0$.
Moreover, $\varphi$ is continuously dependent on the initial data $(\varphi_0, \varphi_1)$ and $C^{\infty}$-dependent on the parameter $a$, i.e., the map
\begin{equation}
(\varphi_0, \varphi_1)\times a\mapsto \varphi
\end{equation}
is a $C^0\times C^{\infty}$ map. By finite speed of propagation, the solution $\varphi$ will be smooth and compactly supported on each $\Sigma_{\tau}$ for $\tau\geq 0$.
\end{prop}
We apply this WP theorem to the linear wave systems \eqref{eq:ReggeWheeler Phi^01Kerr} and \eqref{eq:ReggeWheeler Phi^01KerrNega} of $\varphi=(\phi^0,\phi^1)$ and ensure the existence and uniqueness of the solution for any given compactly supported smooth initial data.

\subsection{Regular and integrable}
\begin{definition}\label{def:regularandintegrable}
\begin{itemize}
\item
A two-form $\mathbf{F}_{\alpha\beta}$ to the Maxwell equations \eqref{eq:MaxwellEqs} is called a \textbf{regular} solution if all components in global Kerr coordinates are smooth in the exterior of the Kerr spacetime and admit a smooth extension to the closure of the exterior region in the maximal analytic extension and the field vanishes near spatial infinity up to a charged stationary Coulomb solution $\mathbf{F}^{\text{charge}}_{\alpha\beta}$ \footnote{The electric charge $\frac{1}{4\pi}\int_{\mathbb{S}^2(t,r)} \mathbf{F}$ and the magnetic charge $\frac{1}{4\pi}\int_{\mathbb{S}^2(t,r)}\leftidx^{\star} \mathbf{F}$ are conserved constants since both $\mathbf{F}$ and its Hodge dual $\leftidx^{\star}\mathbf{F}$ are closed two forms.
Cf. \cite[Sect.3]{larsblue15Maxwellkerr}.}
to the Maxwell equations with $\widetilde{\Phi_0}(\mathbf{F}^{\text{charge}})
=\widetilde{\Phi_2}(\mathbf{F}^{\text{charge}})=0$ and $$\widetilde{\Phi_1}(\mathbf{F}^{\text{charge}})=\frac{1}{2\sqrt{2}\rho^2}
\bigg(\frac{1}{4\pi}\int_{\mathbb{S}^2(t,r)}\leftidx^{\star} \mathbf{F}+i\frac{1}{4\pi}\int_{\mathbb{S}^2(t,r)} \mathbf{F}\bigg).$$
\item
  A solution $\psi=\phi^0_s$ or $\phi^1_s$, which is defined as in \eqref{eq:DefOf Phi^0^1} for a regular Maxwell field, is called an \textbf{integrable} solution to the inhomogeneous SWFIE \eqref{eq:spinweightedwaveeqAssumption} if for every integer $n\geq 0$, every multi-index $0\leq |i|\leq n$ and any $\check{r}_0>r_+$, we have
    \begin{align}\label{eq:integrabledefinition}
    \sum_{0\leq |i|\leq n}\int_{\mathcal{D}(-\infty,\infty)\cap \{r=\check{r}_0\}}\left(\left|\partial^i \psi\right|^2+\left|\partial^i F\right|^2\right)<\infty.
    \end{align}
    Here, we recall in \eqref{def:multiindexderi} the definition of $\partial^i \psi$.
\end{itemize}
\end{definition}

\subsection{Generic constants and general rules}

Universal constants $C$ and $c$, depending only on $\veps_0$, $M$ and $\delta$, are always understood as large and small positive constants, respectively, and may change from line to line throughout this paper based on the algebraic rules: $C+C=C$, $CC=C$, $Cc=C$, etc. When there is no confusion, the dependence on $M$, $\veps_0$ and $\delta$ may always be suppressed. We may fix the value of $\delta\in (0,1/2)$ once for all, and once the constant $\veps_0$ in Theorems \ref{thm:MoraEstiAlmostScalarWave} and \ref{thm:EneAndMorEstiExtremeCompsNoLossDecayVersion2} is chosen, these universal constants can be made to be only dependent on $M$.

$G_1\lesssim G_2$ for two functions $G_1$ and $G_2$ indicates that there exists a universal constant $C$ such that $G_1\leq CG_2$  holds true everywhere. If $G_1\lesssim G_2$ and $G_2\lesssim G_1$, we denote $G_1\sim G_2$ and say that \emph{$G_1$ is equivalent to $G_2$}.

The dependence of constants on other additional objects will be specified for example as $c(\omega_1)$. $\varepsilon$ and $\epsilon$ are potentially small parameters.

The standard Laplacian on unit $2$-sphere $\mathbb{S}^2$ is denoted as $\triangle_{\mathbb{S}^2}$, and the volume form $d\sigma_{\mathbb{S}^2}$ of unit $2$-sphere is $\sin\theta d\theta d\phi^*$ or $\sin\theta d\theta d\phi$ depending on which coordinate system is used.

Some cutoff functions will be used in this paper. Denote $\chi_R(r)$ to be a smooth cutoff function which is $1$ for $r\geq R$ and vanishes identically for $r\leq R-1$, and $\chi_0(r)$ a smooth cutoff function which equals to $1$ for $r\leq r_0$ and is identically zero for $r\geq r_1$. See Section \ref{sect:Redshift} for the choices of $r_0$ and $r_1$. The function $\chi$ is a smooth cutoff both to the future time and to the past time, which will be applied to the solution in Section \ref{sect:MoraEstiCurrentsKerr}.

An overline or a bar will always denote the complex conjugate; $\Re(\cdot)$ denotes the real part, and for brevity, we refer to \textquotedblleft{left-hand side(s)\textquotedblright} and \textquotedblleft{right-hand side(s)\textquotedblright} as \textquotedblleft{LHS\textquotedblright} and \textquotedblleft{RHS,\textquotedblright} respectively.

Throughout this paper, whenever we talk about choosing some multiplier for some equation, it should always be understood as multiplying the equation by the multiplier, performing integration by parts, taking the real part and finally integrating in some spacetime region in global Kerr coordinate system with respect to the measure $\Sigma dt^*drd\theta d\phi^*$.

\section{Proof of Theorem \ref{thm:MoraEstiAlmostScalarWave}}\label{Sect:ProofofTheoremOnSWRWKerr}

In this section, we give the proof of Theorem \ref{thm:MoraEstiAlmostScalarWave}.

\subsection{Morawetz estimate for large $r$}\label{sect:MorawetzLarger}

All the $R_0$ in Propositions \ref{prop:ImprovedMoraEstiLargerSWRWE} and \ref{prop:ImprEnergyMoraForr2minusdeltaphi0} are \textit{a priori} different, and we take the maximal $R_0$ among them such that the estimates hold true uniformly and still denote as $R_0$.

We rewrite the inhomogeneous SWFIE \eqref{eq:spinweightedwaveeqAssumption} into the following form
\begin{align}\label{eq:RewrittenFormofSWRWE}
\Sigma\widetilde{\Box}_g\psi
\triangleq{}&\left\{\partial_r(\Delta\partial_r)
-\tfrac{\left((r^2+a^2)\partial_t+a\partial_{\phi}\right)^2}{\Delta}\right.\notag\\
&\left.\quad+
\tfrac{1}{\sin \theta}\tfrac{d}{d\theta}\left(\sin \theta \tfrac{d}{d\theta}\right)+\left(\tfrac{\partial_{\phi}+is\cos \theta }{\sin \theta}+a\sin \theta\partial_{t}\right)^2\right\}\psi\notag\\
={}&
\left(4ias\cos\theta\partial_t+\tfrac{r^4-2Mr^3+6a^2Mr-a^4}{(\R)^2}\right)\psi+F,
\end{align}
such that $\Sigma\widetilde{\Box}_g$ is the same as the rescaled scalar wave operator $\Sigma \Box_g$, except for $(\tfrac{\partial_{\phi}}{\sin \theta}+is\cot \theta +a\sin \theta\partial_{t})^2$ in place of the operator $(\tfrac{\partial_{\phi}}{\sin \theta}+a\sin \theta\partial_{t})^2$ in the expansion of $\Sigma \Box_g$.

Recall (see, e.g., \cite{dafermos2010decay})  that for each $0<\delta<1/2$ there exist constants $R_0\gg 4M$ and $c=c(\delta)$ such that for all $R\geq R_0$, one can choose a multiplier
\begin{equation}
X_w\bar{\psi}=-\tfrac{1}{\Sigma}\left(f(r)\partial_{r^*}
+\tfrac{1}{4}w(r)\right)\bar{\psi}
\end{equation}
for the rescaled inhomogeneous scalar wave equation
\begin{equation}\label{eq:InhomoRescaledScalarWaveEq}
\Sigma \Box_g\psi=G
\end{equation}
on a subextremal Kerr background
and obtain the following Morawetz estimate in large $r$ region for any $\tau_2>\tau_1\geq 0$:
\begin{align}\label{eq:MoraLargerScalarWaveKerr}
\hspace{4ex}&\hspace{-4ex}c\int_{\mathcal{D}(\tau_1,\tau_2)\cap\{r\geq R\}}
\bigg\{\frac{|\partial_{r^*}\psi|^2}{r^{1+\delta}}+\frac{|\partial_t \psi|^2}{r^{1+\delta}}+\frac{|\check{\nablaslash}\psi|^2}{r}+
\frac{|\psi|^2}{r^{3+\delta}}\bigg\}\notag\\
\leq{} &\check{E}^{R-M}_{\tau_1}(\psi)+\check{E}^{R-M}_{\tau_2}(\psi)
+\int_{\mathcal{D}(\tau_1,\tau_2)\cap\{R-M\leq r\leq R\}}\left(|\check{\partial} \psi|^2+|\psi|^2\right)\notag\\
&+\int_{\mathcal{D}(\tau_1,\tau_2)}\Re\left(G\cdot X_w\bar{\psi}\right).
\end{align}
Here, $f=\chi_R(r)\cdot(1-M^{\delta}r^{-\delta})$,  $w=2\partial_{r^*}f+4\tfrac{1-2M/r}{r}f-2\delta M^{\delta}\tfrac{1-2M/r}{r^{1+\delta}}f$,
$\check{\nablaslash}$ are the standard covariant angular derivatives on sphere $\mathbb{S}^2(t,r)$ as in \eqref{SpinWeightedAngularDerivaBasisOnSphere}, and
\begin{equation*}
\check{E}^{R-M}_{\tau}(\psi)= \int_{\Sigma_{\tau}\cap\{r\geq R-M\}}\left|\check{\partial} \psi\right|^2= \int_{\Sigma_{\tau}\cap\{r\geq R-M\}}( \left|\partial_{t^*}\psi\right|^2
+\left|\partial_r\psi\right|^2
+|\check{\nablaslash}\psi|^2).
\end{equation*}

Since the difference between the operator $(\tfrac{\partial_{\phi}}{\sin \theta}+is\cot \theta +a\sin \theta\partial_{t})^2$ \footnote{Note that the lower order term $\cot \theta$ is not a singularity. In fact, $(\frac{\partial_{\phi}}{\sin\theta}+is\cot \theta)^2$ should be viewed as part of a spin-weighted spherical Laplacian. This issue occurs at other places as well, and the reason is the same as here and hence omitted.} in a spin-weighted wave operator $\Sigma\widetilde{\Box}_g$ and $(\tfrac{\partial_{\phi}}{\sin \theta}+a\sin \theta\partial_{t})^2$ in the expansion of $\Sigma \Box_g$ has terms with coefficients independent of $r$, we will achieve the same type of Morawetz estimate in large $r$ region by utilizing the same multiplier $X_w\bar{\psi}$, with $|\nablaslash\psi|^2-|\psi|^2/r^2$ and $|\partial \psi|^2-|\psi|^2/r^2$ in place of $|\check{\nablaslash}\psi|^2$ and $|\check{\partial} \psi|^2$, $E(\psi)$ replacing $\check{E}(\psi)$, and a substitution of
\begin{equation}\label{eq:sourcetermMoraLarger}
G=\left(4ias\cos\theta\partial_t
+\tfrac{r^4-2Mr^3+6a^2Mr-a^4}{(\R)^2}\right)\psi+F
\end{equation}
in \eqref{eq:MoraLargerScalarWaveKerr}.
The bulk term coming from the source term \eqref{eq:sourcetermMoraLarger} in \eqref{eq:MoraLargerScalarWaveKerr} is then
\begin{align}
\hspace{4ex}&\hspace{-4ex}
\int_{\mathcal{D}(\tau_1,\tau_2)}\frac{-1}{\Sigma}\Re\Big(\left(\left(
4ias\cos\theta\partial_t+\tfrac{r^4-2Mr^3+6a^2Mr-a^4}{(\R)^2}\right)\psi+F\right)
{(f\partial_{r^*}+\tfrac{1}{4}w)\bar{\psi}}\Big)\notag\\
\leq {}&\int_{\mathcal{D}(\tau_1,\tau_2)\cap[R-M,\infty)}\bigg(\frac{C|\partial\psi|^2}{r^2}
+C\Re\left(FX_w\bar{\psi}\right)-\frac{c|\psi|^2}{r^3}\bigg)\notag\\
&+\int_{\mathcal{D}(\tau_1,\tau_2)}
-\Re\Big(\tfrac{r^4-2Mr^3+6a^2Mr-a^4}{\Sigma(\R)^2}\psi
f\partial_{r^*}\bar{\psi}\Big).
\end{align}
Taking into account the factor $\Sigma$ in the volume form, an integration by parts then shows that the last integral term is non-positive for $R$ large enough.
Therefore, we conclude with the following Morawetz estimate in large $r$ region for the inhomogeneous SWFIE \eqref{eq:spinweightedwaveeqAssumption}.
\begin{prop}\label{prop:ImprovedMoraEstiLargerSWRWE}
In a subextremal Kerr spacetime $(\mathcal{M},g_{M,a})$, for any fixed $0<\delta<\half$, there exist constants $R_0$ and $C=C(\delta)$ such that for all $R\geq R_0$, the following estimate holds for any solution $\psi$ to the inhomogeneous SWFIE \eqref{eq:spinweightedwaveeqAssumption} for any $\tau_2>\tau_1\geq 0$:
\begin{align}\label{eq:ImprovedMoraEstiLargerSWRWE}
\int_{\mathcal{D}(\tau_1,\tau_2)\cap\{r\geq R\}}
\mathbb{M}(\psi)
\leq{} &
C({E}^{R-M}_{\tau_1}(\psi)
+{E}^{R-M}_{\tau_2}(\psi))\notag\\
&
+C\int_{\mathcal{D}(\tau_1,\tau_2)\cap\{R-M\leq r\leq R\}}|\partial \psi|^2\notag\\
&+C\bigg|\int_{\mathcal{D}(\tau_1,\tau_2)\cap\{r\geq R-M\}}\Re\left(F X_w\bar{\psi}\right)\bigg|,
\end{align}
where
\begin{align}
{E}^{R-M}_{\tau}(\psi)= \int_{\Sigma_{\tau}\cap\{r\geq R-M\}}\left|{\partial} \psi\right|^2.
\end{align}
\end{prop}
We derive an improved Morawetz estimate near infinity for $\phi^0_{+1}$ as follows.
\begin{prop}\label{prop:ImprEnergyMoraForr2minusdeltaphi0}
In a subextremal Kerr spacetime $(\mathcal{M},g_{M,a})$, for any fixed $0<\delta<\half$, there exist universal constants $R_0$ and $C=C(\delta)$ such that for all $R\geq R_0$, the following estimate holds for any $\tau_2>\tau_1\geq 0$:
\begin{align}\label{eq:MoraInftyr2minusdeltaphi0}
\hspace{4ex}&\hspace{-4ex}\int_{\Sigma_{\tau_2}\cap [R,+\infty)}\left|\partial(r^{2-\delta}{\phi}^{0}_{+1})\right|^2
+\int_{\mathcal{D}(\tau_1,\tau_2)\cap [R,\infty)}r^{-1}\left|\partial (r^{2-\delta}{\phi}^{0}_{+1})\right|^2\notag\\
\leq{}&
C\int_{\Sigma_{\tau_2}\cap [R-M,R)}\left|\partial (r^{2-\delta}{\phi}^{0}_{+1})\right|^2
+C\int_{\Sigma_{\tau_1}\cap [R-M,+\infty)}\left|\partial (r^{2-\delta}{\phi}^{0}_{+1})\right|^2\notag\\
&+C\int_{\mathcal{D}(\tau_1,\tau_2)\cap [R-M,R)}r^{-1}\left|\partial (r^{2-\delta}{\phi}^{0}_{+1})\right|^2.
\end{align}
\end{prop}
\begin{proof}
The equation for $\grave{\phi}^{0}_{+1}=\big(\tfrac{\R}{\sqrt{\Delta}}\big)^{2-\delta}\phi_{+1}^0$  is
\begin{align}
\label{eq:ImprEqForr2minusdeltaphi0}
\hspace{4ex}&\hspace{-4ex}\left(\Sigma \Box_g+\tfrac{2i\cos\theta}{\sin^2 \theta}\partial_{\phi}-\cot^2 \theta+(1+{\delta^2}-3\delta
)\right)\grave{\phi}^{0}_{+1}\notag\\
={}&\tfrac{2(r^3-3Mr^2+a^2r+a^2M)}{\R}
\Big(\tfrac{(1-\delta)V(\sqrt{\R}\grave{\phi}^{0}_{+1})}{\sqrt{\R}}
+\delta\left(\tfrac{\R}{\Delta}\partial_t +\tfrac{a}{\Delta}\partial_{\phi}\right)\grave{\phi}^{0}_{+1}\Big)
\notag\\
&
+\tfrac{\underline{P}_5(r)}{\Delta(\R)^2}\grave{\phi}^{0}_{+1}
+\left(2ia\cos \theta \partial_t-\tfrac{4ar}{\R}\partial_{\phi}\right)\grave{\phi}^{0}_{+1}.
\end{align}
We note that $\underline{P}_5(r)$ here is a polynomial in $r$ with powers no larger than $5$ and coefficients depending only on $a, M$ and $\delta$, and the explicit expression is
\begin{align*}
\underline{P}_5(r)={}&(10-14\delta+4\delta^2)Mr^5
-(20-30\delta+9\delta^2)M^2r^4
+(2-6\delta+2\delta^2)a^2r^4\notag\\
&-(8-8\delta)a^2Mr^3
+(28-32\delta+6\delta^2)a^2M^2r^2
+(4-8\delta+2\delta^2)a^4r^2\notag\\
&-(18-22\delta+4\delta^2)a^4Mr+(1+\delta-\delta^2)a^4M^2+(1-3\delta+\delta^2)a^6.
\end{align*}
For any smooth complex scalar $\psi$ of spin weight $s$, we expand
\begin{align}\label{eq:BoxInTermsOfYV}
\hspace{4ex}&\hspace{-4ex}\left(\Sigma \Box_g+\tfrac{2is\cos\theta}{\sin^2 \theta}\partial_{\phi}-s^2\cot^2 \theta+|s|+{\delta^2}-3\delta)\right)\psi\notag\\
={}&\left(\tfrac{1}{\sin{\theta}} \partial_{\theta}(\sin \theta \partial_{\theta})+\tfrac{1}{\sin^2\theta}\partial_{\phi\phi}^2+\tfrac{2is\cos\theta}{\sin^2 \theta}\partial_{\phi}-s^2\cot^2 \theta+|s|+{\delta^2}-3\delta\right)\psi\notag\\
&-\sqrt{\R}Y\left(\tfrac{\Delta}{\R}V\left(\sqrt{\R}\psi\right)\right)
+\tfrac{2ar}{\R}\partial_{\phi}\psi\notag\\
&+\left(2a\partial_{t\phi}^2+a^2 \sin^2 \theta\partial_{tt}^2\right)\psi-\tfrac{2Mr^3+a^2r^2-4a^2Mr+a^4}{(\R)^2}\psi.
\end{align}
From \eqref{eq:eigenvalueSWSHO} and the fact that $\Lambda_{m\ell}\geq \ell(\ell+1)\geq 2$,  the eigenvalue of the operator in the first line on RHS of \eqref{eq:BoxInTermsOfYV} is not larger than ${\delta^2}-3\delta$ which is negative. Hence, choosing the multiplier
\begin{align}
\hspace{4ex}&\hspace{-4ex}-\frac{1}{\Sigma}\chi_R X_0\overline{\grave{\phi}^{0}_{+1}}\notag\\
\triangleq{}&-\frac{1}{\Sigma}\chi_R \tfrac{\Delta}{\R}\left(\tfrac{(2-\delta)
V(\sqrt{\R}\overline{\grave{\phi}^{0}_{+1}})}{\sqrt{\R}}
+\delta\left(\tfrac{\R}{\Delta}\partial_t +\tfrac{a}{\Delta}\partial_{\phi}\right)\overline{\grave{\phi}^{0}_{+1}}\right)
\end{align}
for equation \eqref{eq:ImprEqForr2minusdeltaphi0} gives
\begin{align}\label{eq:MultplyEqForr2minusdeltaphi0ByTimelikeVF}
\hspace{4ex}&\hspace{-4ex}\int_{\Sigma_{\tau_2}\cap [R,+\infty)}|\partial\grave{\phi}^{0}_{+1}|^2
+\int_{\mathcal{D}(\tau_1,\tau_2)\cap [R,\infty)}r^{-1}\left(| X_0 \grave{\phi}^{0}_{+1}|^2+|\nablaslash \grave{\phi}^{0}_{+1}|^2\right)\notag\\
\lesssim {}&
\Big(\int_{\Sigma_{\tau_2}\cap [R-M,R)}+\int_{\Sigma_{\tau_1}\cap [R-M,+\infty)}\Big)|\partial \grave{\phi}^{0}_{+1}|^2
+\int_{\mathcal{D}(\tau_1,\tau_2)\cap [R-M,\infty)}\tfrac{|\partial\grave{\phi}^{0}_{+1}|^2}{r^{2}}.
\end{align}
Moreover, we can choose the multiplier $-\chi_R r^{-3}(1-2M/r)\overline{\grave{\phi}^{0}_{+1}}$ for \eqref{eq:ImprEqForr2minusdeltaphi0} and arrive at
\begin{align}\label{eq:MultplyEqForr2minusdeltaphi0Byself}
\hspace{4ex}&\hspace{-4ex}\int_{\mathcal{D}(\tau_1,\tau_2)\cap [R,\infty)}r^{-1}\left(|\partial_{r}\grave{\phi}^{0}_{+1}|^2
+|\nablaslash \grave{\phi}^{0}_{+1}|^2\right)\notag\\
\lesssim {}&
\int_{\Sigma_{\tau_2}\cap [R-M,+\infty)}|\partial\grave{\phi}^{0}_{+1}|^2
+\int_{\Sigma_{\tau_1}\cap[R-M,+\infty)}|\partial\grave{\phi}^{0}_{+1}|^2
\notag\\
&
+\int_{\mathcal{D}(\tau_1,\tau_2)\cap [R-M,\infty)}\Big(r^{-1}|\partial_{t^*} \grave{\phi}^{0}_{+1}|^2+r^{-2}|\partial \grave{\phi}^{0}_{+1}|^2\Big).
\end{align}
Estimate \eqref{eq:MoraInftyr2minusdeltaphi0} follows by adding a sufficiently large multiple of \eqref{eq:MultplyEqForr2minusdeltaphi0ByTimelikeVF} to \eqref{eq:MultplyEqForr2minusdeltaphi0Byself} and taking $R_0$ sufficiently large.
\end{proof}

\subsection{Red-shift estimate near $\mathcal{H}^+$}\label{sect:Redshift}
The following red-shift estimate near $\mathcal{H}^+$ for rescaled inhomogeneous scalar wave equation \eqref{eq:InhomoRescaledScalarWaveEq} is taken from \cite[Sect.5.2]{dafermos2011bdedness}.
\begin{lemma}\label{lem:RedshiftInhomoScalarWaveKerr}
In a slowly rotating Kerr spacetime $(\mathcal{M},g_{M,a})$ $(|a|/M\ll 1)$, there exist constants $\veps_0$, $r_+\leq 2M<r_0(M)<r_1(M)<(1+\sqrt{2})M$ and $c$, two smooth real functions $y_1(r)$ and $y_2(r)$ on $[r_+,\infty)$ with $y_1(r)\to 1$ and $y_2(r)\to 0$ as $r\to r_+$, and a $\varphi_{\tau}$-invariant timelike vector field
\begin{equation}\label{def:NVectorField}
N=T+\chi_0(r)\left(y_1(r)Y+y_2(r)T\right)
\end{equation}
such that for all $|a|/M \leq \veps_0$, by choosing a multiplier
\begin{equation}\label{def:NchiVF}
-N_{\chi_0}\bar{\psi}=-\chi_0(r)\Sigma^{-1} N\bar{\psi},
\end{equation}
the following estimate holds for any solution $\psi$ to the rescaled inhomogeneous scalar wave equation \eqref{eq:InhomoRescaledScalarWaveEq} for  any $\tau_2>\tau_1\geq 0$:
\begin{align}\label{eq:RedshiftInhomoScalarWaveKerr}
\hspace{4ex}&\hspace{-4ex}c\bigg(\int_{\Sigma_{\tau_2}\cap\{r\leq r_0\}}\left|\check{\partial}
\psi\right|^2
+\check{E}_{\mathcal{H}^{+}(\tau_1,\tau_2)}(\psi)
+\int_{\mathcal{D}(\tau_1,\tau_2)\cap\{r\leq r_0\}}\left|\check{\partial}\psi\right|^2\bigg)
\notag\\
\leq{} &\int_{\Sigma_{\tau_1}\cap\{r\leq r_1\}}\left|\check{\partial}\psi\right|^2+
\int_{\mathcal{D}(\tau_1,\tau_2)\cap\{r_0\leq r\leq r_1\}}\left|\check{\partial}\psi\right|^2\notag\\
&
+\int_{\mathcal{D}(\tau_1,\tau_2)\cap\{r\leq r_1\}}\Re\left(-G \cdot N_{\chi_0}\bar{\psi}\right).
\end{align}
Here, the horizon flux in the ingoing Kerr coordinates is
\begin{equation}
\check{E}_{\mathcal{H}^{+}(\tau_1,\tau_2)}(\psi) \sim \int_{\mathcal{H}^+(\tau_1,\tau_2)}(|\partial_v\psi|^2
+|\check{\nablaslash}\psi|^2)r^2dv\sin\theta d\theta d\tilde{\phi},
\end{equation}
and we fixed the parameters $r_0$ and $r_1$ so that the universal constant $c$ has no dependence on them.
\end{lemma}
As in the last section, we refer to the rewritten form \eqref{eq:RewrittenFormofSWRWE} of the inhomogeneous SWFIE \eqref{eq:spinweightedwaveeqAssumption}.
The difference between the operator $(\tfrac{\partial_{\phi}}{\sin \theta}+is\cot \theta +a\sin \theta\partial_{t})^2$ in $\Sigma \widetilde{\Box}_g$ and $(\tfrac{\partial_{\phi}}{\sin \theta}+a\sin \theta\partial_{t})^2$ in the expansion of $\Sigma \Box_g$ has terms with coefficients independent of $t$, $\phi$ and $r$; therefore we could use the same multiplier $-N_{\chi_0}$
 to achieve the same estimate for sufficient small $|a|/M$ with the same replacements as in the last section.
On RHS, we are left with
\begin{align*}
\int_{\mathcal{D}(\tau_1,\tau_2)\cap\{r\leq r_1\}}
-\Re\left(\left(\left(4ias\cos\theta\partial_t
+\tfrac{r^4-2Mr^3+6a^2Mr-a^4}{(\R)^2}\right)\psi+F\right)
N_{\chi_0}\bar{\psi}\right),
\end{align*}
which in turn, using both an integration by parts and Cauchy--Schwarz inequality, is bounded by
\begin{align*}
&\int_{\mathcal{D}(\tau_1,\tau_2)\cap\{r\leq r_0\}}\left(-N\left(\tfrac{r^4-2Mr^3+6a^2Mr-a^4}{2(\R)^2}|\psi|^2\right)
-c|\psi|^2\right)\notag\\
&+C|a|\int_{\mathcal{D}(\tau_1,\tau_2)\cap\{r\leq r_1\}}|\partial \psi|^2
+C\int_{\mathcal{D}(\tau_1,\tau_2)\cap[r_+,r_1]}
(\veps^{-1}|F|^2+\veps |\partial\psi|^2).
\end{align*}
The integrals of $\veps |\partial\psi|^2$ and $|a||\partial\psi|^2$ over $\mathcal{D}(\tau_1,\tau_2)\cap[r_+,r_0]$ can be absorbed by the LHS by choosing $\veps$ and $|a|$ small enough.
In conclusion, we have the following red-shift estimate for the inhomogeneous SWFIE \eqref{eq:spinweightedwaveeqAssumption}.
\begin{prop}\label{prop:RedShiftEstiInhomoSWRWE}
In a slowly rotating Kerr spacetime $(\mathcal{M},g_{M,a})$, there exist constants $\veps_0$, $r_+\leq 2M<r_0(M)<r_1(M)<(1+\sqrt{2})M$ and a universal constant $C$, and a $\varphi_{\tau}$-invariant vector field $N$ defined as in \eqref{def:NchiVF} such that for all $|a|/M \leq \veps_0$, the following estimate holds for any solution $\psi$ to the inhomogeneous SWFIE \eqref{eq:spinweightedwaveeqAssumption} for any $\tau_2>\tau_1\geq 0$:
\begin{align}\label{eq:RedShiftEstiInhomoSWRWE}
\hspace{4ex}&\hspace{-4ex}E_{\mathcal{H}^{+}(\tau_1,\tau_2)}(\psi)
+\int_{\Sigma_{\tau_2}\cap\{r\leq r_0\}}\left|\partial\psi\right|^2 +\int_{\mathcal{D}(\tau_1,\tau_2)\cap\{r\leq r_0\}}\left|\partial\psi\right|^2
\notag\\
\lesssim{} &
\int_{\Sigma_{\tau_1}\cap\{r\leq r_1\}}\left|\partial\psi\right|^2
+\int_{\mathcal{D}(\tau_1,\tau_2)\cap\{r_0\leq r\leq r_1\}}\left|\partial\psi\right|^2
+\int_{\mathcal{D}(\tau_1,\tau_2)\cap[r_+,r_1]}|F|^2.
\end{align}
\end{prop}

\subsection{Energy estimate}\label{sect:EnerEstiPartialtKerr}
Choosing a multiplier $-2\Sigma^{-1}\partial_t \bar{\psi}$ for \eqref{eq:RewrittenFormofSWRWE} gives a conservation law that for any $\tau_2>\tau_1\geq 0$,
\begin{equation}\label{eq:SWRWEnergyIdentity}
\int_{\Sigma_{\tau_2}}e(\psi)
+\int_{\mathcal{H}^+(\tau_1,\tau_2)}e_{\mathcal{H}}(\psi)
=\int_{\Sigma_{\tau_1}}e(\psi)
-\int_{\mathcal{D}(\tau_1,\tau_2)}2\Sigma^{-1}\Re\left(
F\partial_t \bar{\psi}\right).
\end{equation}
Here,
\begin{align*}
e(\psi)={}&dt^*(t) e_t(\psi) +dt^* (r^*)e_{r^*}(\psi), & e_{\mathcal{H}}(\psi)={}&dr(t) e_t(\psi) +dr(r^*)e_{r^*}(\psi),
\end{align*}
and
\begin{subequations}
\begin{align}\label{eq:SWRWEnergyDensity}
e_t(\psi)=&
\tfrac{1}{\Sigma}\Big(|\partial_{\theta}\psi|^2
+\left|\tfrac{\partial_{\phi}\psi+
is\cos\theta\psi}{\sin\theta}\right|^2
-\tfrac{a^2}{\Delta}|\partial_{\phi}\psi|^2
+\tfrac{r^4-2Mr^3+6a^2Mr-a^4}{(\R)^2} |\psi|^2\Big)\notag\\
&+
\tfrac{(r^2+a^2)^2-a^2\sin^2\theta\Delta}{\Delta\Sigma}
|\partial_t\psi|^2+\tfrac{(r^2+a^2)^2}{\Delta\Sigma}|\partial_{r^*}\psi|^2,\\
\label{eq:SWRWHorizonDensity}
e_{r^*}(\psi)={}&
-\tfrac{2(r^2+a^2)^2}{\Sigma\Delta}\Re\left(\partial_t \psi \partial_{r^*}\bar{\psi}\right).
\end{align}
\end{subequations}
The energy densities $e(\psi)$ and $e_{\mathcal{H}}(\psi)$ are nonnegative in Schwarzschild case $(a=0)$. For sufficiently small $|a|/M \ll 1$, it holds true in $[r_+, r_0]$ that
\begin{align}\label{eq:SWRWEnergyIdentityNearHorizonControl}
-e(\psi)\leq Ca^2/M^2 |\partial \psi|^2.
\end{align}
We have from \eqref{eq:IdenOfEigenvaluesAndAnguDeriSchw} that
\begin{align}\label{eq:AnEstiForEigenvalueKerr}
\hspace{4ex}&\hspace{-4ex}
\int_{\mathbb{S}^2}\bigg(|\partial_{\theta}\psi|^2
+\left|\tfrac{\partial_{\phi}\psi+
is\cos\theta\psi}{\sin\theta}\right|^2
+|\psi|^2\bigg)d\sigma_{\mathbb{S}^2}\notag\\
\geq {}&\int_0^{\pi}\sum_{m\in \mathbb{Z}}\left(\max\{s^2+|s|,m^2+|m|\}|\psi_{m}|^2\right)
\sin\theta d\theta
\end{align}
with
\begin{equation}
\psi_m(t,r,\theta)=\frac{1}{\sqrt{2\pi}}\int_0^{2\pi}e^{-im\phi}
\psi(t,r,\theta,\phi)d\phi,
\end{equation}
then it follows that
\begin{align}\label{eq:SWRGEnergyNonnegative Kerr}
\hspace{4ex}&\hspace{-4ex}\int_{\mathbb{S}^2}\left(|\partial_{\theta}\psi|^2+\left|
\tfrac{\partial_{\phi}\psi+is\cos\theta\psi}{\sin\theta}
\right|^2-\tfrac{a^2}{\Delta}|\partial_{\phi}\psi|^2
+\tfrac{r^4-2Mr^3+6a^2Mr-a^4}{(\R)^2}|\psi|^2\right)d\sigma_{\mathbb{S}^2}\notag\\
\geq {}&\int_0^{\pi}\sum_{m\in \mathbb{Z}} \left( m^2-\tfrac{a^2m^2}{\Delta}+\tfrac{r^4-2Mr^3+6a^2Mr-a^4}{(\R)^2}\right)
|\psi_{m}|^2\sin\theta d\theta,
\end{align}
which is obviously positive when $r>r_0> 2M$. Therefore, the energy density $e(\psi)\geq c|\partial \psi|^2$ for $r\geq r_0$.
This implies the following energy estimate together with \eqref{eq:SWRWEnergyIdentity} and \eqref{eq:SWRWEnergyIdentityNearHorizonControl}:
\begin{align}\label{eq:SWRWEEnerEsti}
\hspace{4ex}&\hspace{-4ex}\int_{\Sigma_{\tau_2}\cap [r_0,\infty)}|\partial \psi|^2
+\int_{\mathcal{H}^+(\tau_1,\tau_2)}|\partial_t \psi +\partial_{r^*}\psi|^2\notag\\
\lesssim {}&
\int_{\Sigma_{\tau_1}}e(\psi)
+\frac{a^2}{M^2}\bigg(\int_{\Sigma_{\tau_2}
\cap [r_+,r_0]}|\partial  \psi|^2+E_{\mathcal{H}^+(\tau_1,\tau_2)}(\psi)\bigg)\notag\\
&+\bigg|\int_{\mathcal{D}({\tau_1},
\tau_2)}\Sigma^{-1}\Re\left({F} T\bar{\psi}\right)\bigg|.
\end{align}

Clearly, there exist a $\veps_0\geq 0$ and a nonnegative differential function $e_0(\veps_0)\sim \veps_0^2$ with $e_0(0)=0$ such that for all $|a|/M \leq \veps_0$ and any $\tilde{e}>e_0$, by adding $\tilde{e}$ times the red-shift estimate \eqref{eq:RedShiftEstiInhomoSWRWE} to the energy estimate \eqref{eq:SWRWEEnerEsti}, we obtain the following non-degenerate energy estimate analogous to \cite[Prop.5.3.1]{dafermos2010decay}:
\begin{prop}\label{prop:energyestimateforT+eN}
\begin{align}\label{eq:SWRWEEnerEstiII}
\hspace{4ex}&\hspace{-4ex}
\int_{\Sigma_{\tau_2}}\left|e(\psi)\right|
+\int_{\mathcal{H}^+(\tau_1,\tau_2)}|\partial_t \psi +\partial_{r^*}\psi|^2
+\int_{\mathcal{D}(\tau_1,\tau_2)\cap\{r\leq r_0\}}\left|\partial\psi\right|^2\notag\\
\hspace{4ex}&\hspace{-4ex}+\tilde{e}E_{\tau_2}(\psi)
+\tilde{e}E_{\mathcal{H}^+(\tau_1,\tau_2)}(\psi)\notag\\
\lesssim &
\int_{\Sigma_{\tau_1}}\left|e(\psi)\right|
+\tilde{e}E_{\tau_1}(\psi)
+\tilde{e}\int_{\mathcal{D}(\tau_1,\tau_2)\cap[r_0,r_1]}
\left|\partial\psi\right|^2+\mathcal{E}_{F,\tilde{e},\tau_1,\tau_2}.
\end{align}
where $\mathcal{E}_{F,\tilde{e},\tau_1,\tau_2}$ is defined as
\begin{align}\label{eq:ErrorFtildeetau1tau2}
\mathcal{E}_{F,\tilde{e},\tau_1,\tau_2}(\psi)={}\tilde{e}
\int_{\mathcal{D}({\tau_1},\tau_2)\cap[r_+,r_1]}
\left|F\right|^2+\bigg|\int_{\mathcal{D}({\tau_1},
\tau_2)}\Re\left({\Sigma}^{-1}{F} T\bar{\psi}\right)\bigg|.
\end{align}
\end{prop}

We here state a finite in time energy estimate for the inhomogeneous SWFIE \eqref{eq:spinweightedwaveeqAssumption} based on the above discussions, which is an analog of \cite[Prop.5.3.2]{dafermos2010decay}.
\begin{prop}\label{prop:FiniteInTimeEnergyEstimateInhomoSWRWE}
\textbf{(Finite in time energy estimate).}
Given an arbitrary $\epsilon>0$, there exists an $\veps_0>0$ depending on $\epsilon$ and a universal constant $C$ such that for $|a|/M\leq \veps_0$, $1\geq \tilde{e}\geq e_0(\veps_0)\sim \veps_0^2$ and for any $\tau_0\geq 0$ and all $0\leq \tau\leq \epsilon^{-1}$, we have
\begin{align}\label{eq:FiniteInTimeEnergyEstimateInhomoSWRWETo T+eN}
\hspace{4ex}&\hspace{-4ex}
\int_{\Sigma_{\tau_0+\tau}}\left|e(\psi)\right|
+\tilde{e}E_{\tau_0+\tau}(\psi)
+\tilde{e}E_{\mathcal{H}^+(\tau_0,\tau_0+\tau)}(\psi)\notag\\
\leq  {}& (1+C\tilde{e})\bigg(\int_{\Sigma_{\tau_0}}
\left|e(\psi)\right|
+\tilde{e}\sum_{i=0,1}{E}_{\tau_0}(\phi^i)\bigg)
+C\mathcal{E}_{F,\tilde{e},\tau_0,\tau_0+\tau}
\end{align}
and
\begin{align}\label{eq:FiniteInTimeEnergyEstimateInhomoSWRWEII}
\int_{\mathcal{D}({\tau_0},{\tau_0}+\tau)\cap[r_0,r_1]}|\partial \psi|^2\leq C\sum_{i=0,1}{E}_{\tau_0}(\phi^i).
\end{align}
\end{prop}
\begin{proof}
The previous proposition together with the second estimate implies the first estimate, while the second estimate follows from the fact that \eqref{eq:FiniteInTimeEnergyEstimateInhomoSWRWEII} holds for Schwarzschild case for all $\epsilon$ from \eqref{eq:Mora01:Schw} and the well-posedness property in Proposition \ref{prop:LWP LinearWaveEqSystem} to the linear wave system \eqref{eq:ReggeWheeler Phi^01Kerr} or \eqref{eq:ReggeWheeler Phi^01KerrNega}.
\end{proof}

\subsection{Separated angular and radial equations}\label{sect:SeparateAngAndRadialEqs}

In the exterior of a subextremal Kerr black hole, if the solution $\psi$ to the equation \eqref{eq:spinweightedwaveeqAssumption} is integrable (as in Definition \ref{def:regularandintegrable}), it then holds in $L^2(dt)$ that
\begin{align}
\psi=\frac{1}{\sqrt{2\pi}}\int_{-\infty}^{\infty}e^{-i\omega t}\psi_{\omega}(r,\theta,\phi)d\omega,
\end{align}
with $\psi_{\omega}$ defined as the Fourier transform of $\psi$:
\begin{equation}
\psi_{\omega}=\frac{1}{\sqrt{2\pi}}\int_{-\infty}^{\infty}e^{i\omega t}\psi(t,r,\theta,\phi)dt.
\end{equation}
Further decompose $\psi_{\omega}$ in $L^2(\sin\theta d\theta d\phi)$ into
\begin{equation}
\psi_{\omega}=\sum_{m,\ell}\psi_{m\ell}^{(a\omega)} (r)Y_{ m\ell}^{s}(a\omega, \cos\theta)e^{im\phi}, \ \ m\in \mathbb{Z}.
\end{equation}
For each $m$, $\left\{Y_{ m\ell}^{s}(a\omega, \cos\theta)\right\}_{\ell}$, with $\min{\{\ell\}}=\max\left\{|m|,|s|\right\}$, are the eigenfunctions of the self-adjoint operator
\begin{align}\label{eq:SWSHOpe Kerr}
\textbf{S}_m=\tfrac{1}{\sin\theta}\partial_{\theta}\sin\theta\partial_{\theta}
-\tfrac{m^2+2ms\cos \theta+s^2}{\sin^2 \theta}+a^2\omega^2 \cos^2 \theta-2a\omega s\cos \theta
\end{align}
on $L^2(\sin\theta d\theta)$. These eigenfunctions, called as \textquotedblleft{\emph{spin-weighted spheroidal harmonics},\textquotedblright} form a complete orthonormal basis on $L^2(\sin\theta d\theta)$ and have eigenvalues $-\Lambda_{m\ell}^{(a\omega)}$ defined by
\begin{equation}\label{eq:SWSHOpeEq Kerr}
\textbf{S}_mY_{ m\ell}^{s}(a\omega, \cos\theta)=-\Lambda_{m\ell}^{(a\omega)}
Y_{ m\ell}^{s}(a\omega, \cos\theta).
\end{equation}
Similarly, we define $F_{\omega}$ and $F_{m\ell}^{(a\omega)}$.

The radial equation for $\psi_{m\ell}^{(a\omega)}$ is then
\begin{align}\label{eq:SWRWRadialEqKerr}
\left\{\partial_r(\Delta\partial_r)
+(\underline{V})_{m\ell}^{(a\omega)}(r)\right\}\psi_{m\ell}^{(a\omega)}
=F_{m\ell}^{(a\omega)},
\end{align}
where the potential
\begin{equation}
(\underline{V})_{m\ell}^{(a\omega)}(r)
=\tfrac{(r^2+a^2)^2\omega^2+a^2m^2
-4aMrm\omega}{\Delta}-\left(\lambda_{m\ell}^{(a\omega)}(r)
+a^2\omega^2\right).
\end{equation}
We here utilized a substitution of
\begin{equation}\label{eq:defOflambda}
\lambda_{m\ell}^{(a\omega)}(r)=\Lambda_{m\ell}^{(a\omega)}-1+\tfrac{r^4-2Mr^3+6a^2Mr-a^4}{(\R)^2},
\end{equation}
by which the above radial equation \eqref{eq:SWRWRadialEqKerr} is in the same form as the radial equation \cite[Eq.(33)]{dafermos2010decay}\footnote{There is one term $-4aMrm\omega$  missed in \cite[Eq.(33)]{dafermos2010decay}, but what is used thereafter is the Schr\"{o}dinger equation $(34)$ in Section $9$ which is correct. Therefore, the validity of the proof will not be influenced by the missing term.} for the scalar field.

An integration by parts and a usage of the Plancherel lemma and the orthonormality property of the basis $\left\{Y_{ m\ell}^{s}(a\omega, \cos\theta)e^{im\phi}\right\}_{m\ell}$ imply
\begin{align*}
&\int_{-\infty}^{+\infty}\sum_{m,\ell}\Lambda_{m\ell}^{(a\omega)}
|\psi_{m\ell}^{(a\omega)}|^2d\omega\notag\\
=& \int_{-\infty}^{\infty}\int_{\mathbb{S}^2}
\left\{|\partial_{\theta}\psi|^2+\left|\tfrac{\partial_{\phi}\psi+
is\cos\theta\psi}{\sin\theta}\right|^2
-|a\cos\theta\partial_t\psi+is\psi|^2
+2s^2|\psi|^2\right\}
d\sigma_{\mathbb{S}^2}dt
.
\end{align*}
We have from basic properties of Fourier transform that for any $r>r_+$,
\begin{align*}
&\int_{-\infty}^{\infty}\int_{0}^{2\pi}\int_{0}^{\pi}|\psi(t,r,
\theta,\phi)|^2\sin\theta d\theta d\phi dt =\int_{-\infty}^{\infty}\sum_{m,\ell}\left|\psi_{m\ell}^{(a\omega)}(r)
\right|^2d\omega,\notag\\
&\int_{-\infty}^{\infty}\int_{0}^{2\pi}\int_{0}^{\pi}\left|
\partial_r \psi(t,r,
\theta,\phi)\right|^2\sin\theta d\theta d\phi dt =\int_{-\infty}^{\infty}\sum_{m,\ell}\left|\partial_r\psi_{m\ell}^{(a\omega)}(r)
\right|^2d\omega,\notag\\
&\int_{-\infty}^{\infty}\int_{0}^{2\pi}\int_{0}^{\pi}\left|
\partial_t\psi(t,r,\theta,\phi)\right|^2\sin\theta d\theta d\phi dt =\int_{-\infty}^{\infty}\sum_{m,\ell}\omega^2\left|\psi_{m\ell}^{(a\omega)}(r)
\right|^2d\omega.
\end{align*}

\subsection{Frequency localised multiplier estimates}\label{sect:MoraEstiCurrentsKerr}

To justify the separation procedures in Section \ref{sect:SeparateAngAndRadialEqs}, the assumption that the solution $\psi$ is integrable (recall Definition \ref{def:regularandintegrable}) is sufficient, but this is \textit{a priori} unknown. Therefore, we apply cutoffs in time to the solution both to the future and to the past, and then do separation for the wave equation which the gained function satisfies.

Let $\chi_1(x)$ be a smooth cutoff function which equals to $0$ for $x\leq 0$ and is identically $1$ when $x\geq 1$. For any $\varepsilon>0$ and $\tau\geq 2\varepsilon^{-1}$, we define
\begin{equation}
\chi=\chi_{\tau,\varepsilon}(t^*)=\chi_1(\varepsilon t^*)\chi_1(\varepsilon(\tau-t^*))
\end{equation}
and
\begin{equation}
\psi_{\chi}=\chi\psi
\end{equation}
in global Kerr coordinates. The cutoff function $\psi_{\chi}$ is now a smooth function supported in $0\leq t^*\leq \tau$, and $\psi_{\chi}=\psi$ in $\varepsilon^{-1}\leq t^*\leq \tau-\varepsilon^{-1}$. Moreover, it satisfies the following inhomogeneous wave equation
\begin{align}\label{eq:SWRWE CutoffVersion}
\mathbf{L}_s\psi_{\chi}&=\chi F+\Sigma\left(2\nabla^{\mu}\chi\nabla_{\mu}\psi
+\left(\Box_g\chi\right)\psi\right)-2isa\cos\theta\partial_t\chi\psi \notag\\
&\triangleq F_{\chi} .
\end{align}
The fact that the afore-defined $\chi$ is $\phi$-independent is used here in deriving this equation.

Note the fact that the functions $\psi_{\chi}$ and $F_{\chi}$ are compactly supported in $0\leq t^*\leq \tau$ at each fixed $r> r_+$, and the assumption that $\psi$ is a compactly supported smooth section solving one subequation of a linear wave system; hence $\psi_{\chi}$ is an integrable solution to \eqref{eq:SWRWE CutoffVersion} from Proposition \ref{prop:LWP LinearWaveEqSystem}. In the following discussions, we apply the mode decompositions as in Section \ref{sect:SeparateAngAndRadialEqs} to $\psi_{\chi}$ and $F_{\chi}$ and separate the wave equation \eqref{eq:SWRWE CutoffVersion} into the angular equation \eqref{eq:SWSHOpeEq Kerr} and radial equation \eqref{eq:SWRWRadialEqKerr}, but with $R_{m\ell}^{(a\omega)}=(\psi_{\chi})_{m\ell}^{(a\omega)}$ and $\left(F_{\chi}\right)_{m\ell}^{(a\omega)}$ in place of $\psi_{m\ell}^{(a\omega)}$ and $F_{m\ell}^{(a\omega)}$, respectively.

Before introducing the microlocal currents, we give some estimates for the inhomogeneous term $F_{\chi}$ here. Due to the fact that $\nabla \chi$ and $\Box_g \chi$ are supported in
\begin{equation}
\left\{0\leq t^* \leq \varepsilon^{-1}\right\}\cup
\left\{\tau -\varepsilon^{-1}\leq t^*\leq \tau\right\},
\end{equation}
it holds in the coordinate system $\left(t^*, r,\theta,\phi^*\right)$ that
\begin{subequations}\label{eq:PropertyofCutoffchi}
\begin{align}
|\partial_{t^*}\chi|\leq C\varepsilon, \quad &\left|\Box_g\chi\right|\leq C\varepsilon^2,\\
\left|\nabla^{\mu}\chi \nabla_{\mu}\psi\right|^2
+\left|\tfrac{ias\cos\theta\partial_t\chi\psi}{\Sigma}\right|^2&\leq C\varepsilon^2\left(|\partial\psi|^2
+a^2 M^{-2}\left|r^{-1}{\psi}\right|^2\right).
\end{align}
\end{subequations}
\subsubsection{Currents in phase space}
In what follows, we will suppress the dependence on $a$, $\omega$, $m$ and $\ell$ of the functions $R_{m\ell}^{(a\omega)}(r)$, $F_{m\ell}^{(a\omega)}(r)$, $\Lambda_{m\ell}^{(a\omega)}$, $\lambda_{m\ell}^{(a\omega)}(r)$, $(\underline{V})_{m\ell}^{(a\omega)}(r)$ and other functions defined by them, and when there is no confusion, the dependence on $r$ may always be implicit (except for the radial part $R(r)$ to avoid misunderstanding with the radius parameter $R$). Thus, we will write them as $R(r)$, $F$, $\Lambda$, $\lambda$ and $\underline{V}$.

We transform the radial equation \eqref{eq:SWRWRadialEqKerr} into a Schr\"{o}dinger form, which will be of great use to define the microlocal currents below, by setting
\begin{equation}
\label{def:u(r)andH(r)}
u(r)=\sqrt{r^2+a^2}R(r), \ \ \ \
H(r)=\tfrac{\Delta F_{\chi}(r)}{\left(r^2+a^2\right)^{3/2}}.
\end{equation}
The Schr\"{o}dinger equation for $u(r)$ reads after some calculations
\begin{equation}\label{eq:eqofu}
u(r)''+\left(\omega^2-V(r)\right)u(r)=H(r),
\end{equation}
where
\begin{align}\label{eq:SepaRadiSchroFormPoential}
V= &\omega^2-\tfrac{\KDelta}{(r^2+a^2)^2} \underline{V} + \tfrac{1}{(r^2+a^2)} \tfrac{d^2}{dr^{*2}} (r^2+a^2)^{1/2}\notag\\
=&\tfrac{4Mram\omega-a^2m^2+\Delta\left(\lambda+a^2\omega^2
\right)}{(r^2+a^2)^2}+\tfrac{\Delta}{(r^2+a^2)^4} \left( a^2\Delta + 2Mr(r^2-a^2) \right ),
\end{align}
and a prime $'$ denotes a partial derivative with respect to $r^*$ in tortoise coordinates.

Given any real, smooth functions $y,h$ and $f$, define the microlocal currents
\begin{subequations}\label{eq:currents}
\begin{align}
\label{eq:currenty(3)}
Q^y&=y ( |u'|^2+ (\omega^2-V)|u|^2),\\
\label{eq:currenth(3)modified}
  Q^h&=h\Re (u'\overline{u})-\tfrac{1}{2}h'|u|^2,\\
Q^f&=Q^{h=f'}+Q^{y=f}=f'\Re (u'\overline{u})
- (\tfrac{1}{2}f''-f (\omega^2-V))|u|^2+f |u'|^2.
\label{eq:currentf(3)}
\end{align}
\end{subequations}
The currents $Q^y$ and $Q^h$ are constructed via multiplying equation \eqref{eq:eqofu} by $2y\bar{u}'$ and $h\bar{u}$, respectively.
We calculate the derivatives of the above currents as follows
\begin{subequations}\label{eq:currentsderi}
\begin{align}
(Q^y)'&=y'(|u'|^2+
(\omega^2-V)|u|^2)-yV'|u|^2
+2y\Re (u'\overline{H} ),
\label{eq:Qyderi}\\
 (Q^h )'&=h ( |u' |^2+
 (V-\omega^2 )|u|^2 )-\tfrac{1}{2}h''|u|^2
+h\Re (u\overline{H} ),
\label{eq:Qhderi}\\
 (Q^f )'&=2f' |u' |^2-fV'|u|^2+
\Re (2fu'\overline{H}+f'u\overline{H} )
-\tfrac{1}{2}f'''|u|^2.
\label{eq:Qfderi}
\end{align}
\end{subequations}

\subsubsection{Frequency regimes}

We now define the separated frequency regimes.  Let $\omega_1$ and $\lambda_1$ be (potentially large) parameters and $\lambda_2$ be a (potentially small) parameter, all of which to be determined in the proof below. The frequency space is divided into
\begin{itemize}
\item $\mathcal{F}_{T}=\{(\omega, m, \ell): |\omega|\geq \omega_1, \lambda<\lambda_2\omega^2\}$;
\item $\mathcal{F}_{Tr}=\{(\omega, m, \ell): |\omega|\geq \omega_1, \lambda\geq\lambda_2\omega^2\}$;
\item $\mathcal{F}_{A}=\{(\omega, m, \ell): |\omega|\leq \omega_1, \Lambda>\lambda_1\}$;
\item $\mathcal{F}_{B}=\{(\omega, m, \ell): |\omega|\leq \omega_1, \Lambda\leq \lambda_1\}$.
\end{itemize}
We fix $R_0$ as in Section \ref{sect:MorawetzLarger} and  $r_0$ as in Section \ref{sect:Redshift}.
\begin{lemma}
\label{rem:PotentialDereasingForRadiusLarge}
For all $|a|<M$ and all frequency triplets $(\omega, m, \ell)$,
\begin{align}\label{eq:Lambdaesti:rough}
 \Lambda+3a^2\omega^2 \geq {}  \frac{3}{4}\max\{m^2+|m|, s^2+|s|\},
 \end{align}
 and there exists a sufficiently large $R_1\geq R_0+M$ and two constants $c=c(M)>0$ and $C=C(M)>0$ such that for all $r\geq R_1$,
\begin{align}
\label{eq:estimatederiVnearinf}
V'(r)<-cr^{-3}(\Lambda+4a^2\omega^2+m^2+1)+Cr^{-3}a^2\omega^2.
\end{align}
\end{lemma}

\begin{proof}
Let
\begin{subequations}\label{eq:decompofpotentialV}
\begin{align}
V={}&V_m+V_e,\\
V_m={}&\tfrac{\Delta(\Lambda+a^2\omega^2)}{(r^2+a^2)^2}.
\end{align}
\end{subequations}
Then
\begin{align}
\label{eq:decompofpotentialV1e}
V_e={}&\tfrac{4Mram\omega-a^2m^2}{(r^2+a^2)^2}
-\tfrac{a^2\Delta^2}{(r^2+a^2)^4}.
\end{align}
One finds for all $r\in [r_+,\infty)$,
\begin{align}
\label{eq:derivativeofVm}
V_m'={}&-2\Delta(r^2+a^2)^{-4}(\Lambda+a^2\omega^2)
(r^3-3Mr^2+a^2r+a^2M),\\
\label{eq:generalformofderiofVe}
V_e'\leq{}&Ca^2\Delta^2r^{-9}+C\Delta{r^{-2}}\left(r^{-4} |am\omega|+r^{-5}a^2m^2\right).
\end{align}
One adds $-3a^2\omega^2Y_{ m\ell}^{s}$ to both sides of \eqref{eq:SWSHOpeEq Kerr} and obtains
\begin{align}
\hspace{4ex}&\hspace{-4ex}(\tfrac{1}{\sin\theta}\partial_{\theta}\sin\theta\partial_{\theta}
-\tfrac{m^2+2ms\cos \theta+s^2}{\sin^2 \theta}+a^2\omega^2 \cos^2 \theta-2a\omega s\cos \theta -3a^2\omega^2)Y_{ m\ell}^{s}\notag\\
={}&-(\Lambda+3a^2\omega^2)
Y_{ m\ell}^{s}.
\end{align}
Rewrite the LHS as
\begin{align}
&\left(\tfrac{1}{\sin\theta}\partial_{\theta}\sin\theta\partial_{\theta}
-\tfrac{m^2+2ms\cos \theta+s^2}{\sin^2 \theta}\right)Y_{ m\ell}^{s}\notag\\
&-\left(a^2\omega^2 \sin^2 \theta+\tfrac{1}{2}(2a\omega  +s\cos \theta)^2 -\tfrac{1}{2}s^2\cos^2\theta\right)Y_{ m\ell}^{s}.
\end{align}
From the discussions around \eqref{eq:eigenvalueSWSHO}, the operator at the first line has eigenvalues not greater than $-\max\{m^2+|m|, s^2+|s|\}$, and the eigenvalues of the operator at the second line are less than or equal to $\frac{1}{2}s^2$, hence
\begin{align*}
\Lambda + 3a^2\omega^2 \geq \max\{m^2+|m|, s^2+|s|\}-\frac{1}{2}s^2 \geq \frac{3}{4}\max\{m^2+|m|, s^2+|s|\}.
\end{align*}
Hence from \eqref{eq:derivativeofVm}, $V_m'\leq -\frac{3}{2}r^{-3}(\Lambda+3a^2\omega^2)+Cr^{-3}a^2\omega^2$ for $r\gg 4M$ large enough.
Together with \eqref{eq:generalformofderiofVe}, this implies
\begin{align}
V'\leq {}& (-\tfrac{3}{2}(\Lambda+3a^2\omega^2)+Ca^2r^{-2})r^{-3}
+Cr^{-3}(a^2\omega^2+r^{-2}a^2m^2).
\end{align}
For $R_1$ large enough and $r\geq R_1$, there exists a constant $c_0\in (0,1)$ such that the coefficient of the first term on the RHS is negative and bounded above by $-c_0r^{-3}(\Lambda +3a^2\omega^2+m^2+s^2)$, and the second term of the RHS is bounded by $Cr^{-3}a^2\omega^2+\frac{c_0}{2}r^{-3}a^2m^2$; therefore, this completes the proof.
\end{proof}

We obtain a phase-space version of Morawetz estimate by choosing different functions $y$, $h$ and $f$ in each frequency regime in Sections \ref{sect:timedominatedfreq}--\ref{sect:boundfreq}. The proofs in Sections \ref{sect:timedominatedfreq}, \ref{sect:trapfreq} and \ref{sect:angulardominatedfreq} follow from the discussions in \cite[Sect.9.4--9.6]{dafermos2010decay}; nevertheless, we present here the entire proof for completeness.

\subsubsection{$\mathcal{F}_{T}$ regime (time-dominated frequency regime)}\label{sect:timedominatedfreq}

For $|a|/M\leq \veps_0\ll 1$, by choosing small enough $\lambda_2$ and large enough $\omega_1$, there exists a constant $c_0<1/2$ such that for frequency triplet in $\mathcal{F}_{T}$,
\begin{equation}
\omega^2-V\geq ({1-c_0})\omega^2\ \text{for}\ r\in[r_+,\infty).
\end{equation}
As to the potential $V$, apart from the fact in Lemma \ref{rem:PotentialDereasingForRadiusLarge}, it satisfies that for all $r^*$,
\begin{equation}
|V'|\leq C\Delta/r^{5}((\Lambda+3a^2\omega^2)+1).
\end{equation}
We choose a function $y$ to satisfy the following properties:
\begin{enumerate}
  \item $y\geq 0$, $y'\geq c\Delta/r^4$ in $(r_+, R_1-M]$,
  \item $y\geq 0$, $y'\geq 0$ in $[R_1-M, R_1]$,
  \item $y=1$ in $[R_1,\infty)$.
\end{enumerate}
In view of the above properties,
$
(Q^y)'\geq 2y\Re(u'\overline{H})-Ca^2r^{-3}\omega^2
$
for $r\geq R_1$, $
(Q^y)'\geq 2y\Re(u'\overline{H})
$
for $r\leq r_0$,
and in the region $r\in (r_0, R_1-M)$,  we have
\begin{align}
y'(\omega^2-V)-yV'\geq{}& (1-c_0)y'\omega^2 -Cy\Delta/r^5((\lambda+a^2\omega^2)+1)\notag\\
\geq{}&(c(1-c_0) - Cr^{-1}(\lambda_2+a^2+\omega_1^{-2}))\Delta/r^4\omega^2.
\end{align}
By taking both $\lambda_2$ and $\veps_0$ sufficiently small and $\omega_1$ sufficiently large, the RHS is larger than $c\Delta/r^4\omega^2$.
Hence, integrating \eqref{eq:Qyderi} over $[r_{-\infty}^*, r_{\infty}^*]$ gives
\begin{lemma}\label{lem:TimeDominatedFreqEstimate}
Let $r_{\infty}^*>R_1^*> (R_0+M)^*$ and $r_{-\infty}^*<r_0^*$ be arbitrary. There exists a small enough $\lambda_2$, a large enough $\omega_1$ and a sufficiently small $\veps_0$ such that we have in $\mathcal{F}_T$ frequency regime the following estimate for all $|a|/M\leq \veps_0$:
\begin{align}\label{eq:TimeDominatedFreqEstimate}
\hspace{4ex}&\hspace{-4ex}c\int_{r_0^*}^{R_0^*}\frac{\Delta}{r^{4}}
\left(\left|u'\right|^2
+\left(\omega^2+(\Lambda+4a^2\omega^2)
+1\right)|u|^2\right)\notag\\
\leq {}& \int_{r_{-\infty}^*}^{r_{\infty}^*}2y\Re\left(u'\overline{H}\right)
+Q^{y}\left(r_{\infty}^*\right)-Q^y\left(r_{-\infty}^*\right)
+C\int_{R_1^*}^{r_{\infty}^*}a^2r^{-3}\omega^2|u|^2.
\end{align}
\end{lemma}

\subsubsection{$\mathcal{F}_{Tr}$ regime (trapped frequency regime)}\label{sect:trapfreq}

The small parameter $\lambda_2$ has been fixed in Section \ref{sect:timedominatedfreq}, and we will fix $\omega_1$ here.
This is the only frequency regime where trapping phenomenon could happen.

Recall from \eqref{eq:derivativeofVm} that
\begin{equation}\label{eq:Vmprime:trap}
V_m'=\tfrac{-2\Delta}{(r^2+a^2)^4}(\Lambda+a^2\omega^2)
(r^3-3Mr^2+a^2r+a^2M),
\end{equation}
and the polynomial $P_3(r)=r^3-3Mr^2+a^2r+a^2M$ has a unique zero point $r_{a,M}$ in $r\in [r_+,\infty)$ which satisfies $|r_{a,M}-3M|\leq Ca^2$. Meanwhile, estimate \eqref{eq:Lambdaesti:rough} and the requirements of $|\omega|\geq \omega_1$ and $\lambda\geq\lambda_2\omega^2$ in this frequency regime imply that for large enough $\omega_1$ and sufficiently small $|a|/M\leq \veps_0$,
\begin{align}
\label{eq:Lambdaesti:rough2}
\Lambda+a^2\omega^2\geq c(\omega^2+\max\{m^2+|m|, s^2+|s|\}).
\end{align}
On the other hand,
\begin{equation}\label{eq:Veprime:trap}
r^3|V_e'|\leq C\Delta/r^2(a^2 m^2+|am\omega|+1).
\end{equation}
Therefore, given any small $\veps_1>0$, we have for sufficiently large $\omega_1$ (depending on the $\lambda_2$) and sufficiently small $\veps_0$ that $V'$ has no zeros outside the region $[r_{a,M}-\veps_1, r_{a,M}+\veps_1]$. In fact, $V'$ has a unique simple zero in this neighborhood. This can be seen as follows.  Note from \eqref{eq:decompofpotentialV1e} and \eqref{eq:derivativeofVm} that for sufficiently small $\veps_0$,
\begin{subequations}
\begin{align}
\left(\tfrac{(r^2+a^2)^4}{\Delta r^2}V_m'\right)'(r_{a,M})={}&-\tfrac{2\Delta(r_{a,M})(\Lambda+a^2\omega^2)}
{r_{a,M}^2+a^2}
\Big(1-\tfrac{a^2}{r_{a,M}^2}-\tfrac{2a^2M}{r_{a,M}^3}\Big)\notag\\
\leq{}& -c(\Lambda+a^2\omega^2),\\
\label{eq:Vesecondderi}
\Big|\left(\tfrac{(r^2+a^2)^4}{\Delta r^2}V_e'\right)'\Big|\lesssim{}&\Delta r^{-2}(a^2m^2+|am\omega|+1).
\end{align}
\end{subequations}
For $\omega_1$ sufficiently large and $\veps_0$ sufficiently small, $\Lambda+a^2\omega^2$ will be much bigger than the RHS of \eqref{eq:Vesecondderi}. This implies that in the small region $[r_{a,M}-\veps_1, r_{a,M}+\veps_1]$, we have
\begin{align}
\left(\tfrac{(r^2+a^2)^4}{\Delta r^2}V'\right)'
\leq{}& -c(\Lambda+a^2\omega^2).
\end{align}
Therefore, this proves that for sufficiently large $\omega_1$ and sufficiently small $\veps_0$, there is a unique simple zero, which we denote by $r_{m\ell}^{(a\omega)}$, in $[r_+,\infty)$ for any $(\omega, m, \ell)\in \mathcal{F}_{Tr}$, and $r_{m\ell}^{(a\omega)}\in [r_{a,M}-\veps_1, r_{a,M}+\veps_1]$. Moreover, for $r\leq r_{m\ell}^{(a\omega)}$,
\begin{align}
V'\geq -c(r-r_{m\ell}^{(a\omega)})\Delta/r^6(\Lambda +a^2\omega^2),
\end{align}
and for $r\geq r_{m\ell}^{(a\omega)}$,
\begin{align}
V'\leq -c(r-r_{m\ell}^{(a\omega)})\Delta/r^6(\Lambda +a^2\omega^2).
\end{align}

Choose a function $f$ in the $Q^f$ current to satisfy the following properties:
\begin{enumerate}
  \item $f'\geq 0$ for all $r^*$, and $f'\geq c\Delta/r^{4}$ for $r_0 \leq r\leq R_0$,
  \item $f$ changes sign from negative to positive at $r=r_{m\ell}^{(a\omega)}$, $\lim\limits_{r^*\to -\infty}f=-1$, and $f=1$ for some large $R_2\geq R_0$,
  \item $f'''\leq -c$ near $r=r_{m\ell}^{(a\omega)}$,
  \item $-fV'-\tfrac{1}{2}f'''\geq c(\omega_1)((\Lambda+a^2\omega^2)+\omega^2)
      (r-r_{m\ell}^{(a\omega)})^2\Delta/r^7$ for all $r^*$.
\end{enumerate}
By integrating \eqref{eq:Qfderi} over $[r_{-\infty}^*, r_{\infty}^*]$, we arrive at the following conclusion.
\begin{lemma}\label{lem:TrappedDominatedFreqEstimate}
Let $r_{\infty}^*>R_2^*\geq R_0^*$ and $r_{-\infty}^*<r_0^*$ be arbitrary. There exists a large $\omega_1$ and a sufficiently small $\veps_0$ such that in $\mathcal{F}_{Tr}$ frequency regime, it holds for all $|a|/M\leq \veps_0$ that
\begin{align}\label{eq:TrappedFreqEstimate}
\hspace{4ex}&\hspace{-4ex}\int_{r_0^*}^{R_0^*}\frac{\Delta}{r^4}\left(c\left|u'\right|^2
+r^{-1}\left[c+c(\omega_1)(1-r^{-1}r_{m\ell}^{(a\omega)})^2
\left(\omega^2+(\Lambda+4a^2\omega^2)
\right)\right]|u|^2\right)\notag\\
\leq & \int_{r_{-\infty}^*}^{r_{\infty}^*}\left(2f\Re\left(u'\overline{H}\right)
+f'\Re\left(u\overline{H}\right)\right)
+Q^{f}\left(r_{\infty}^*\right)-Q^f\left(r_{-\infty}^*\right).
\end{align}
\end{lemma}

\subsubsection{$\mathcal{F}_{A}$ regime (angular-dominated frequency regime)}\label{sect:angulardominatedfreq}

Here, we fix $\omega_1$ and a small parameter $\veps_1$ as in Section \ref{sect:trapfreq} and will choose $\lambda_1$ to be sufficiently large. The general idea is as follows. In this regime, for sufficiently small $|a|/M$, the zero points of $V'(r)$ in $[r_+,\infty)$ are located in a small neighborhood of $r=3M$.
The $Q^f$ current is utilized to achieve the positivity of the bulk term outside this small neighborhood, while in this small neighborhood, $hV|u|^2$ in $(Q^h)'$ with $h(r)$ being a positive constant is used to dominate the potentially negative bulk terms in $(Q^f)'$.

Recall from Section \ref{sect:trapfreq} properties \eqref{eq:Vmprime:trap} and \eqref{eq:Veprime:trap} of $V_m'$ and $V_e'$ and the polynomial $P_3(r)=r^3-3Mr^2+a^2r+a^2M$ has a unique zero point $r_{a,M}$ in $r\in [r_+,\infty)$ satisfying $|r_{a,M}-3M|\leq Ca^2$. In this frequency regime, it is clear that for $\lambda_1$ sufficiently large,  property \eqref{eq:Lambdaesti:rough2} is valid.
Therefore, given any small $\veps_1>0$ which is fixed as in Section \ref{sect:trapfreq}, we have for sufficiently small $\veps_0$ and sufficiently large $\lambda_1$ (depending on the $\omega_1$) that $V'$ has no zeros outside the region $[r_{a,M}-\veps_1, r_{a,M}+\veps_1]$.

Given $\veps_0\ll 1$, choose constants
$$r_+<r_0<r_{left1}<r_{left2}<r_{a,M}-\veps_1 < r_{a,M}+\veps_1 <r_{right1} <r_{right2} <R_0 <\infty.$$
For sufficiently small $\veps_1$ and $\veps_0$, there exist constants $c_1, c_2>0$ and functions $f$ and $h$ satisfying the following conditions:
\begin{enumerate}
\item $f'\geq 0$ for $r>r_+$, $f'\geq c_1$ for $r_0\leq r\leq R_0$,
\item $f\sim -1 +(r-r_+)$ near the horizon, $f\leq -\frac{1}{2}$ for $r\leq r_{left2}$, $f\geq \frac{1}{2}$ for $r\geq r_{right1}$ and $f=1$ for $r\geq R_0+1$ such that $fV'<0$ in $(r_+, r_{a,M}-\veps_1)\cup (r_{a,M}+\veps_1, \infty)$ and $fV'\leq c_2 V$ in $[r_{a,M}-\veps_1, r_{a,M}+\veps_1]$,
\item $h=0$ in $[r_+, r_{left1}]\cup [r_{right2}, \infty)$ and $h=2c_2$ in $[r_{left2},r_{right1}]$.
\end{enumerate}
We then take $\lambda_1$ sufficiently large such that for $r\in [r_+, r_{a,M}-\veps_1]\cup [r_{a,M}+\veps_1, \infty)$,
\begin{align*}
-\frac{1}{2}fV'-\frac{1}{2}f'''-\frac{1}{2}h''\geq 0,
\end{align*}
and for $r\in[r_{a,M}-\veps_1,r_{a,M}+\veps_1]$,
\begin{align*}
h(V-\omega^2)-\frac{1}{2}h''-fV'-\frac{1}{2}f'''
\geq{}&c_2V-2c_2\omega^2-\frac{1}{2}f'''
\geq{}cc_2\Lambda.
\end{align*}
Therefore, by integrating over $r_{-\infty}^*\leq r^*\leq r_{\infty}^*$, we conclude
\begin{lemma}\label{lem:AngularDominatedFreqEstimate}
Let $r_{\infty}^*>R_0^*$ and $r_{-\infty}^*<r_0^*$ be arbitrary. Fix $\omega_1$ as in Section \ref{sect:trapfreq}. There exists a large $\lambda_1$ and a sufficiently small $\veps_0$ such that in $\mathcal{F}_A$ frequency regime, we have for all $|a|/M\leq \veps_0$ the following estimate
\begin{align}\label{eq:AngularDominatedFreqEstimate}
\hspace{4ex}&\hspace{-4ex}c\int_{r_0^*}^{R_0^*}\left(\left|u'\right|^2+\Delta/ r^{5}\left(\omega^2+(\Lambda+4a^2\omega^2)
+1\right)|u|^2\right)\notag\\
\leq & \int_{r_{-\infty}^*}^{r_{\infty}^*}\left(2f\Re\left(u'\overline{H}\right)
+\left(f'+h\right)\Re\left(u\overline{H}\right)\right)
+Q^{f}\left(r_{\infty}^*\right)-Q^f\left(r_{-\infty}^*\right).
\end{align}
\end{lemma}

\subsubsection{$\mathcal{F}_{B}$ regime (bounded frequency regime)}\label{sect:boundfreq}
Fix $\omega_1$ and $\lambda_1$ as above. This is a compact frequency regime. Since the Morawetz estimates in the Schwarzschild case are proved in Appendix \ref{appe:MoraFS1FIeqSchw}, the Morawetz estimates in this frequency regime for sufficiently small $|a|/M\leq \veps_0$ can be obtained by a stability argument. We shall show this below by choosing an appropriate function $f$ in the current $Q^f$.

In this regime, the minimum value of eigenvalues $\Lambda$ for the separated angular equation \eqref{eq:SWSHOpeEq Kerr} is close to $\max\{s^2+|s|,m^2+|m|\}$ due to smallness of $|a\omega|$. More explicitly, from \eqref{eq:SWSHOpeEq Kerr},
\begin{align}
\Lambda +a^2\omega^2\geq{}& \max\{s^2+|s|,m^2+|m|\}-2|a\omega s|
\notag\\
\geq {}&\max\{2,m^2+|m|\}-\tilde\varepsilon-\tilde\varepsilon^{-1}a^2\omega_1^2.
\end{align}
 By choosing $\tilde\varepsilon=1/2$ and $\veps_0$ sufficiently small, this gives directly that
\begin{align}
m^2\leq C(\lambda_1+a^2\omega_1^2).
\end{align}

We take
\begin{align}\label{eq:boundfreq:fchoice}
f=\frac{r^3-3Mr^2+a^2r+a^2M}{r^3}
 \end{align}
 in the current $Q^f$ and find for sufficiently small $|a|/M$,
\begin{equation}
f'=\frac{\Delta}{r^4(r^2+a^2)}(3Mr^2-2a^2r-3a^2M)\geq\frac{cM\Delta}{r^4}.
\end{equation}
Moreover, $f$ vanishes at a unique point $r_{a,M}$, which is already defined in Section \ref{sect:angulardominatedfreq}, with $|r_{a,M}-3M|\leq Ca^2$ for small enough $|a|/M$.
To obtain a Morawetz estimate in the phase space, we need to verify $-fV'-\frac{1}{2}f'''\geq 0$. Note from \eqref{eq:decompofpotentialV} that $V_e=O(r^{-3})Mam\omega+O(r^{-4})a^2m^2+O(r^{-4})a^2$, hence
\begin{align*}
&V_m'={}-2(r^2+a^2)^{-4}\Delta(r^3-3Mr^2+a^2r+a^2M)
(\Lambda+a^2\omega^2),\\
&|V_e'|\leq {}C(r^2+a^2)^{-1}\Delta(Mr^{-4}|am\omega|+r^{-5}a^2m^2+r^{-5}a^2),\\
&\left|-\frac{1}{2}f'''+\frac{3M\Delta}{(r^2+a^2)^4}(3r^2-20Mr+30M^2)\right|
\leq Ca^2\Delta r^{-7}.
\end{align*}
This implies
\begin{align}
\label{eq:boundedfreq:posibulk11}
-fV'-\frac{1}{2}f'''
\geq {}&-fV_m'-\frac{3M\Delta}{(r^2+a^2)^4}(3r^2-20Mr+30M^2)\notag\\
&
-\frac{C\Delta}{r^2+a^2}(Mr^{-4}|am\omega|+a^2r^{-5}(m^2+1)).
\end{align}
For the right-hand side of the first line, it equals
\begin{align}
\label{eq:boundedfreq:posibulk12}
\hspace{4ex}&\hspace{-4ex}\frac{2\Delta(r^3-3Mr^2+a^2r+a^2M)^2(\Lambda+a^2\omega^2)}{r^3(r^2+a^2)^4}
-\frac{3M\Delta(3r^2-20Mr+30M^2)}{(r^2+a^2)^4}\notag\\
\geq {}&\frac{\Delta}{r^3(r^2+a^2)^4}
\left[2(r-3M)^2r^4(2-\varepsilon)-3Mr^3(3r^2-20Mr+30M^2)\right]
\notag\\
&-\frac{C\Delta}{r^2+a^2}r^{-3}\varepsilon^{-1}a^2\omega_1^2
-Ca^2r^{-5}.
\end{align}
From the above two estimates, to show that there exists a constant $c>0$ such that $-fV'-\frac{1}{2}f'''\geq \frac{c\Delta r^3}{(r^2+a^2)^4}$, it is enough to prove there exists a positive constant $c_3>0$ such that
\begin{align}
\label{eq:boundfreq:poly}
4r(r-3M)^2 -3M(3r^2-20Mr+30M^2) \geq c_3r^3,
\end{align}
since then the second line of \eqref{eq:boundedfreq:posibulk12} will dominate over $\frac{c_3}{2}\frac{\Delta r^3}{(r^2+a^2)^4}$ by taking $\varepsilon=c_3/2$, which will in turn dominate over the sum of the second line of \eqref{eq:boundedfreq:posibulk11} and the last line of \eqref{eq:boundedfreq:posibulk12} by taking $|a|/M\leq \veps_0$ sufficiently small. The LHS of \eqref{eq:boundfreq:poly} is $4r^3-33Mr^2+96M^2r-90M^3$, and it is routine to verify that there is no zero point in $[r_+,\infty)$; hence the relation \eqref{eq:boundfreq:poly} follows.

In total, a choice of function $f$ as in \eqref{eq:boundfreq:fchoice} gives
\begin{align}
(Q^f)'\geq{}&{cM\Delta}{r^{-4}}|u'|^2+{c\Delta}{r^{-5}}|u|^2
+
\Re\left(2fu'\overline{H}+f'u\overline{H}\right).
\end{align}
By integrating over $[r_{-\infty}^*,r_{\infty}^*]$ and taking into account of the boundedness of $|\omega|$ and $\Lambda$, we arrive at
\begin{lemma}\label{lem:BoundFreqEstimate}
Let $r_{\infty}^*>R_0^*$ and $r_{-\infty}^*<r_0^*$ be arbitrary. Fix $\omega_1$ as in Section \ref{sect:trapfreq} and $\lambda_1$ as in Section \ref{sect:angulardominatedfreq}.  In $\mathcal{F}_B$ frequency regime,  there exists a $\veps_0>0$ such that the following estimate holds true for all $|a|/M\leq \veps_0$
\begin{align}\label{eq:BoundFreqEstimate}
\hspace{4ex}&\hspace{-4ex}c\int_{r_0^*}^{R_0^*}\left(\Delta/r^4 \left|u'\right|^2+\Delta/ r^{5}\left(\omega^2+(\Lambda+4a^2\omega^2)
+1\right)|u|^2\right)\notag\\
\leq & \int_{r_{-\infty}^*}^{r_{\infty}^*}\left(2f\Re\left(u'\overline{H}\right)
+f'\Re\left(u\overline{H}\right)\right)
+Q^{f}\left(r_{\infty}^*\right)-Q^f\left(r_{-\infty}^*\right).
\end{align}
\end{lemma}

\subsection{Summing up}\label{sect:summingandIEDcurrentesti}

Fix $\lambda_2,\omega_1,\veps_1,\lambda_1$ such that all the estimates in Sections \ref{sect:timedominatedfreq}--\ref{sect:boundfreq} and the relation $[r_{a,M}-\veps_1, r_{a,M}+\veps_1]\subseteq [r_{\text{trap}}^-,r_{\text{trap}}^+]=[2.9M,3.1M]$ hold true. This ensures that  all the $r_{m\ell}^{(a\omega)}$ in $\mathcal{F}_{Tr}$ frequency regime are lying inside $[r_{\text{trap}}^-,r_{\text{trap}}^+]$.
Summing over $m,\ell$ and integrating over $\omega$ in the estimates proved in the previous subsection, we obtain an inequality of the form
\begin{align}
\label{eq:sumup:mainidentity}
\mathbf{M}\leq{}\mathbf{E} + \mathbf{F}_{r_{\infty}^*}
+\mathbf{F}_{r_{-\infty}^*}+\mathbf{F}_{r_{-\infty}^*}
+\mathbf{R}.
\end{align}
Here, $\mathbf{M}$, $\mathbf{E}$, $\mathbf{F}_{r_{\infty}^*}$, $\mathbf{F}_{r_{-\infty}^*}$ and $\mathbf{R}$ are Morawetz integral terms, error terms from the source $H$, flux term at $r_{\infty}^*$, flux term at $r_{-\infty}^*$ and a remainder term coming from the last integral in \eqref{eq:TimeDominatedFreqEstimate}, respectively. We shall estimate these terms one by one below.

\subsubsection{Morawetz terms $\mathbf{M}$}
From the properties in Section \ref{sect:SeparateAngAndRadialEqs}, we have
\begin{align*}
\mathbf{M}\geq{}c\int_{\mathcal{D}(0,\tau)\cap \{r_0\leq r\leq R_0\}} \left(\frac{|\partial_{r^*}\psi_{\chi}|^2+|\psi_{\chi}|^2}{r^3}
+\chi_{\text{trap}}(r)\frac{|\partial_{t^*} \psi_{\chi}|^2+|\nablaslash\psi_{\chi}|^2}{r^3}\right).
\end{align*}
Since $\chi=1$ in $[\varepsilon^{-1}, \tau-\varepsilon^{-1}]$, we use the estimate \eqref{eq:FiniteInTimeEnergyEstimateInhomoSWRWETo T+eN} and obtain
\begin{align}
\label{eq:sumup:Moraterm}
\hspace{4ex}&\hspace{-4ex}c\int_{\mathcal{D}(0,\tau)\cap \{r_0\leq r\leq R_0\}} \left(\frac{|\partial_{r^*}\psi|^2+|\psi|^2}{r^3}
+\chi_{\text{trap}}(r)\frac{|\partial_{t^*} \psi|^2+|\nablaslash\psi|^2}{r^3}\right)
\notag\\
\leq{}&\mathbf{M}
+C(r_0,R_0)\varepsilon^{-1}\bigg(
\int_{\Sigma_{0}}\left|e(\psi)\right|
+e_0\sum_{j=0,1}E_{0}(\phi^j)+\mathcal{E}_{F,e_0,0,\tau}
\bigg)\notag\\
&
+C(r_0,R_0)\varepsilon^{-1}\bigg(
\int_{\Sigma_{\tau-\varepsilon^{-1}}}\left|e(\psi)\right|
+e_0\sum_{j=0,1}E_{\tau-\varepsilon^{-1}}(\phi^j)
\bigg).
\end{align}

\subsubsection{Error terms $\mathbf{E}$}\label{sect:EstiErrorTermCurrents}
The error terms $\mathbf{E}$ are of this form
\begin{align*}
\int_{r_{-\infty}^*}^{r_{\infty}^*}\int_{-\infty}^{\infty}\sum_{m,\ell}
\Re\left(
\tfrac{\Delta }{\left(r^2+a^2\right)^{3/2}}
\left(F_{\chi}\right)_{m\ell}^{(a\omega)}\left(c(r)\bar{u}_{m\ell}^{(a\omega)}
+d(r)\partial_{r^*}\bar{u}_{m\ell}^{(a\omega)}\right)\right)d\omega dr^*.
\end{align*}
Let
\begin{subequations}
\begin{align}
F_{\chi}={}&F_{\chi,c}+F_{\chi,s},\\
F_{\chi,c}={}&\Sigma\left(2\nabla^{\mu}\chi\nabla_{\mu}\psi
+\left(\Box_g\chi\right)\psi\right)-2isa\cos\theta\partial_t\chi\psi,\\
F_{\chi,s}={}&\chi F.
\end{align}
\end{subequations}
We decompose
\begin{align}
\mathbf{E}={}\mathbf{E}_{c,[r_{-\infty}^*, R_3^*]}
+\mathbf{E}_{c,[R_3^*,r_{\infty}^*]}
+\mathbf{E}_{s,[r_{-\infty}^*, R_3^*]}
+\mathbf{E}_{s,[R_3^*,r_{\infty}^*]},
\end{align}
where $R_3$ is a constant to be fixed and lies in $(\max\{R_2, R_1\}, \max\{R_2, R_1\}+M)$ such that $y=f=1$ and $h=0$ for $r\geq R_3$, the intervals $[r_{-\infty}^*, R_3^*]$ and $[R_3^*,r_{\infty}^*]$ are the $r^*$ region to integrate, and the subscripts $c$ and $s$ denote the error terms coming from the cutoff part $F_{\chi,c}$ and the source $F_{\chi,s}$ respectively.

By applying the Cauchy--Schwarz inequality, we have for any $\veps_2>0$ that
\begin{align}
\mathbf{E}_{c,[r_{-\infty}^*, R_3^*]}
\leq{}& \int_{\mathcal{D}(0,\tau)\cap[r_{-\infty},R_3]}
\frac{C\veps_2 }{ r^{3}}\left(|\psi_{\chi}|^2+|\partial_{r^*} \psi_{\chi}|^2\right)
+\frac{C}{\veps_2}F_{\chi,c}.
\end{align}
Note that we have incorporated powers of $R_3$ into the constant $C$.
The first term is bounded by $ \int_{\mathcal{D}(0,\tau)\cap[r_{-\infty},R_3]}
\frac{C\veps_2 }{r^{3}}\left(|\psi|^2+|\partial_{r^*} \psi|^2\right)$,
and, using estimates \eqref{eq:PropertyofCutoffchi} and in view of the support of the derivatives of $\chi$, we bound the second term by
\begin{align}
\hspace{4ex}&\hspace{-4ex}
C\veps_2^{-1}\varepsilon^2
\bigg(\int_{\mathcal{D}(0,\varepsilon^{-1})\cap[r_{-\infty},R_3]}
+\int_{\mathcal{D}(\tau-\varepsilon^{-1},\tau)
\cap[r_{-\infty},R_3]}\bigg)
|\partial\psi|^2\notag\\
\leq{}&C\veps_2^{-1}\varepsilon \bigg(\sum_{j=0,1}(E_{\tau-\varepsilon^{-1}}(\phi^j)
+E_{0}(\phi^j))
+\mathcal{E}_{F,1,0,\tau}\bigg),
\end{align}
where estimate \eqref{eq:FiniteInTimeEnergyEstimateInhomoSWRWETo T+eN} with $\tilde{e}=1$ is used in this inequality. Therefore, we arrive at
\begin{align}
\label{eq:sumup:errorterm:cnear}
\mathbf{E}_{c,[r_{-\infty}^*, R_3^*]}
\lesssim {} &
\frac{\varepsilon}{\veps_2} \sum_{j=0,1}
(E_{0}(\phi^j)+E_{\tau-\varepsilon^{-1}}(\phi^j))
+\frac{\varepsilon}{\veps_2}\mathcal{E}_{F,1,0,\tau}\notag\\
&
+\veps_2 \int_{\mathcal{D}(0,\tau)\cap [r_0,R_0]}
\frac{|\psi|^2+|\partial_{r^*} \psi|^2}{r^3}\notag\\
&+\veps_2 \bigg(\int_{\mathcal{D}(0,\tau)\cap [r_+,r_0]}
+\int_{\mathcal{D}(0,\tau)\cap [R_0,\infty)}\bigg)
\frac{|\psi|^2+|\partial_{r^*} \psi|^2}{r^3} .
\end{align}
Here, we have separated the integrals over different radius regions, the reason of which will be clear in Section \ref{sect:SumAndFinishProof}.

The same argument applies to $\mathbf{E}_{s,[r_{-\infty}^*, R_3^*]}$, yielding
\begin{align}
\label{eq:sumup:errorterm:snear}
\mathbf{E}_{s,[r_{-\infty}^*, R_3^*]}
\lesssim {} &
\veps_2 \int_{\mathcal{D}(0,\tau)\cap [r_0,R_0]}
\frac{|\psi|^2+|\partial_{r^*} \psi|^2}{r^3}
+\int_{\mathcal{D}(0,\tau)}{\veps_2^{-1}}
\frac{|F|^2}{r^3}\notag\\
&+\veps_2 \bigg(\int_{\mathcal{D}(0,\tau)\cap [r_+,r_0]}
+\int_{\mathcal{D}(0,\tau)\cap [R_0,\infty)}\bigg)
\frac{|\psi|^2+|\partial_{r^*} \psi|^2}{r^3}.
\end{align}

Consider the term $\mathbf{E}_{c,[R_3^*,r_{\infty}^*]}$. Since $f(r)=y(r)=1$, $h(r)=0$ and $\partial_{r^*}\chi=0$, and from the expression of $F_{\chi,s}$, it equals
\begin{align}\label{exp:sumup:errorterm:caway}
&\int_0^{\tau}\int_{\mathbb{S}^2}
\int^{r_{\infty}^*}_{R_3^*}\Re\left(\chi\partial_{r^*}
(\sqrt{r^2+a^2}\bar{\psi})
\frac{2\Sigma\Delta }{(r^2+a^2)^{3/2}}\nabla^{\mu}\chi\nabla_{\mu}\psi\right)
dr^*d\sigma_{\mathbb{S}^2}dt^*\notag\\
&+\int_0^{\tau}\int_{\mathbb{S}^2}
\int^{r_{\infty}^*}_{R_3^*}\Re\left(\partial_{r^*}
(\sqrt{r^2+a^2}\chi\bar{\psi})
\frac{\Sigma\Delta }{(r^2+a^2)^{3/2}}\left(\Box_g\chi\right)\psi\right)
dr^*d\sigma_{\mathbb{S}^2}dt^*\notag\\
&-\int_0^{\tau}\int_{\mathbb{S}^2}
\int^{r_{\infty}^*}_{R_3^*}\Re\left(\partial_{r^*}
(\sqrt{r^2+a^2}\chi\bar{\psi})
\frac{2ias\cos\theta \Delta }{(r^2+a^2)^{3/2}}\partial_t\chi\psi\right)
dr^*d\sigma_{\mathbb{S}^2}dt^*.
\end{align}
The derivatives of $\chi$ are supported in $[0,\varepsilon^{-1}]\cup[\tau-\varepsilon^{-1},\tau]$; hence, using estimates \eqref{eq:PropertyofCutoffchi}, the first term of \eqref{exp:sumup:errorterm:caway} is bounded by
 \begin{align}
 \label{eq:sumup:errorterm:cawaybulk}
 C\varepsilon\bigg(\int_{\mathcal{D}(0,\varepsilon^{-1})
\cap[R_3,r_{\infty}]}
+\int_{\mathcal{D}(\tau-\varepsilon^{-1},\tau)
\cap[R_3,r_{\infty}]}\bigg)
|\partial\psi|^2.
\end{align}
Estimate \eqref{eq:FiniteInTimeEnergyEstimateInhomoSWRWETo T+eN} yields that this is further controlled by
\begin{align}
\label{eq:sumup:errorterm:caway}
&C\int_{\Sigma_{0}}
\left|e(\psi)\right|
+Ce_0\sum_{j=0,1}{E}_{0}(\phi^j)
+C\int_{\Sigma_{\tau-\varepsilon^{-1}}}
\left|e(\psi)\right|
+Ce_0\sum_{j=0,1}{E}_{\tau-\varepsilon^{-1}}(\phi^j)
\notag\\
&
+C\mathcal{E}_{F,e_0,0,\varepsilon^{-1}}
+C\mathcal{E}_{F,e_0,\tau-\varepsilon^{-1},\tau}.
\end{align}
The second term can be rewritten as
\begin{align*}
\int_0^{\tau}\int_{\mathbb{S}^2}
\int^{r_{\infty}^*}_{R_3^*}\left(\partial_{r^*}
\left(\left|\sqrt{r^2+a^2}\chi\bar{\psi}\right|^2\right)
\frac{\Sigma\Delta }{(r^2+a^2)^{2}}\chi^{-1}\Box_g\chi\right)
dr^*d\sigma_{\mathbb{S}^2}dt^*.
\end{align*}
Applying an integration by parts in $r^*$ and noting the bounds for $r\geq R_3$
\begin{align}
|\partial_{r^*}\chi \Box_g\chi|
+|\chi \partial_{r^*}\Box_g\chi|+|\partial_r(\Sigma\Delta (r^2+a^2)^{-2})\chi\Box_g\chi|\lesssim \varepsilon r^{-2},
\end{align}
this is bounded by
\begin{align}
&C\bigg|\int_0^{\tau}\int_{\mathbb{S}^2}
\left(
\left|\sqrt{r^2+a^2}\chi\bar{\psi}\right|^2
\frac{\Sigma\Delta }{(r^2+a^2)^{2}}\chi^{-1}\Box_g\chi\right)_{r=R_3}
dr^*d\sigma_{\mathbb{S}^2}dt^*\bigg|\notag\\
&
+C\varepsilon\bigg(\int_{\mathcal{D}(0,\varepsilon^{-1})
\cap[R_3,r_{\infty}]}
+\int_{\mathcal{D}(\tau-\varepsilon^{-1},\tau)
\cap[R_3,r_{\infty}]}\bigg)
|\partial\psi|^2.
\end{align}
Here the boundary term at $r=r_{\infty}$ vanishes for sufficiently large $r_{\infty}$ from the reduction in Section \ref{sect:LWPandGlobalExistenceLinearWaveSystem}. A simple application of the mean-value principle in $r$ allows us to fix $R_3$ by requiring that the boundary term at $r=R_3$ is bounded by $C$ times an integral over $[\max\{R_2, R_1\}, \max\{R_2, R_1\}+M]$, which is in turn bounded by
$$
C\varepsilon\bigg(\int_{\mathcal{D}(0,\varepsilon^{-1})
\cap[R_3-M,r_{\infty}]}
+\int_{\mathcal{D}(\tau-\varepsilon^{-1},\tau)
\cap[R_3-M,r_{\infty}]}\bigg)
|\partial\psi|^2.
$$
The same argument as in estimating \eqref{eq:sumup:errorterm:cawaybulk} above then applies, suggesting the second term of \eqref{exp:sumup:errorterm:caway} is bounded by \eqref{eq:sumup:errorterm:caway}.
The Cauchy--Schwarz inequality applied to the last term of \eqref{exp:sumup:errorterm:caway} gives an upper bound
\begin{align}
C|a|\varepsilon\bigg(\int_{\mathcal{D}(0,\varepsilon^{-1})
\cap[R_3,r_{\infty}]}
+\int_{\mathcal{D}(\tau-\varepsilon^{-1},\tau)
\cap[R_3,r_{\infty}]}\bigg)
|\partial\psi|^2,
\end{align}
which can again be bounded by \eqref{eq:sumup:errorterm:caway}. These together imply that $\mathbf{E}_{c,[R_3^*,r_{\infty}^*]}$ is bounded by  \eqref{eq:sumup:errorterm:caway}.

Consider in the end the term $\mathbf{E}_{s,[R_3^*,r_{\infty}^*]}$. This equals
\begin{align*}
\int_0^{\tau}\int_{\mathbb{S}^2}
\int^{r_{\infty}^*}_{R_3^*}\Re\left(\chi\partial_{r^*}
(\sqrt{r^2+a^2}\bar{\psi})
\frac{\Delta }{(r^2+a^2)^{3/2}}\chi F\right)
dr^*d\sigma_{\mathbb{S}^2}dt^*.
\end{align*}
Since $R_3\geq R_0$, we utilize the Cauchy--Schwarz inequality and find
\begin{align}\label{eq:sumup:errorterm:saway}
\mathbf{E}_{s,[R_3^*,r_{\infty}^*]}
\lesssim{}
\veps_3\int_{\mathcal{D}(0,\tau)
\cap[R_0,r_{\infty}]}\frac{|\partial\psi|^2}{r^{1+\delta}}
+\veps_3^{-1}\int_{\mathcal{D}(0,\tau)
\cap[R_0,r_{\infty}]}\frac{|F|^2}{r^{3-\delta}}.
\end{align}

\subsubsection{Flux terms}
First note that the flux term $\mathbf{F}_{r_{\infty}^*}$ vanishes by choosing $ r_{\infty}^*$ sufficiently large in view of the reduction in Section \ref{sect:LWPandGlobalExistenceLinearWaveSystem}.
We have for the flux term $\mathbf{F}_{r_{-\infty}^*}$ that
\begin{align*}
\mathbf{F}_{r_{-\infty}^*}
\lesssim{}&
\int_{-\infty}^{\infty}\int_{\mathbb{S}^2} \left(|\partial_t \psi_{\chi}|^2 +|\partial_{r^*}\psi_{\chi}|^2
+a^2|\partial_{\phi}\psi_{\chi}|^2 \right)_{r^*=r_{-\infty}^*}d\sigma_{\mathbb{S}^2}dt\notag\\
&+\int_{-\infty}^{\infty}\int_{\mathbb{S}^2} \left(\Delta\left(|\nablaslash \psi_{\chi}|^2 +|\psi_{\chi}|^2 \right)\right)_{r^*=r_{-\infty}^*}d\sigma_{\mathbb{S}^2}dt.
\end{align*}
By taking the limit $r_{-\infty}^*\to -\infty$, the second line tends to $0$, and the term on the RHS of the first line has a limit, arriving at
\begin{align}
\label{eq:sumup:fluxterm}
\hspace{4ex}&\hspace{-4ex}\limsup_{r_{-\infty}^*\to -\infty}\mathbf{F}_{r_{-\infty}^*}\notag\\
\lesssim{}&\int_{-\infty}^{\infty}\int_{\mathbb{S}^2} \left(|\partial_t \psi_{\chi}|^2 +|\partial_{r^*}\psi_{\chi}|^2
+a^2|\partial_{\phi}\psi_{\chi}|^2 \right)_{r=r_+}d\sigma_{\mathbb{S}^2}dt\notag\\
\lesssim{}& \int_{\mathcal{H}^+(0,\tau)} \left(|\partial_t\psi+\partial_{r^*}\psi|^2
+a^2|\partial_{\phi}\psi|^2 \right)
+\varepsilon^2
\int_{\mathcal{H}^+(0,\tau)}
|\psi|^2\notag\\
\lesssim{}&E_{0}(\psi)
+(\veps_0^2+\varepsilon^2){E}_{\mathcal{H}^+(0,\tau)}(\psi)
+\veps_0^2\int_{\mathcal{D}(0,\tau)\cap[r_0,r_1]}
\left|\partial\psi\right|^2+\mathcal{E}_{F,e_0,0,\tau}.
\end{align}
Here, we used the fact that $\partial_t=\partial_{r^*}-\frac{a}{2Mr_+}\partial_{\phi}$ on $\mathcal{H}^+$ in the second step  and inequality \eqref{eq:SWRWEEnerEstiII} with $\tilde{e}=e_0(\veps_0)\sim \veps_0^2$ in the third step.

\subsubsection{Remainder term $\mathbf{R}$}

Recall that this remainder term is from the last integral in \eqref{eq:TimeDominatedFreqEstimate}, hence from \eqref{eq:PropertyofCutoffchi} and the fact that $R_1\geq R_0$,
\begin{align}
\label{eq:sumup:remainderR}
\mathbf{R}\lesssim{}&\int_{-\infty}^{\infty}\int_{R_1^*}^{\infty}\int_{\mathbb{S}^2}a^2r^{-3}|\partial_t (\sqrt{r^2+a^2}\psi_{\chi})|^2 d\sigma_{\mathbb{S}^2}dr^*dt \notag\\
\lesssim{}&\int_{\mathcal{D}(0,\tau)\cap [R_0,\infty)}a^2r^{-3}(|\partial_t \psi|^2 +\veps^2 |\psi|^2).
\end{align}

\subsubsection{Closing the proof}\label{sect:SumAndFinishProof}
The separation in variables in Section \ref{sect:SeparateAngAndRadialEqs} requires $\tau> 2\veps^{-1}$; hence
we separate the cases $\tau> 2\veps^{-1}$ and $0\leq \tau\leq 2\veps^{-1}$. For each $\psi=\phi_s^i$, the choice of $\veps$ is made to be different. The rest of the proof will be for each single $\psi=\phi_s^i$ and $F=F_s^i$.

In the first case where $\tau> 2\veps^{-1}$, we have obtained an identity \eqref{eq:sumup:mainidentity}. Add to the equation \eqref{eq:sumup:mainidentity} $C_1$ times the non-degenerate energy estimate \eqref{eq:SWRWEEnerEstiII} with $\tilde{e}=c_3$
and $c_3$ times the
Morawetz estimate in large $r$ region in Proposition \ref{prop:ImprovedMoraEstiLargerSWRWE} with $(\tau_1,\tau_2,R)=(0,\tau,R_0)$. Then for sufficiently small $\veps_0$ and $\varepsilon$, choosing $C_1$, $\veps_2$, $\veps_3$ and $c_3$ to satisfy $C_1\gg1$, $\veps_0^2+\varepsilon^2+e_0 + \veps_2+ \veps_3\ll c_3\ll 1$ and $C_1c_3\ll 1$ allows us to absorb the bulk terms of $\psi$ over $[r_0,r_1]$
and $[R_0-M,R_0]$ on the RHS of these two estimates, the bulk integrals over $[r_0,R_0]$ in both \eqref{eq:sumup:errorterm:cnear} and \eqref{eq:sumup:errorterm:snear} and the bulk integral over $[r_0,r_1]$ in \eqref{eq:sumup:fluxterm} by the LHS of \eqref{eq:sumup:Moraterm}, and absorb the last line of \eqref{eq:sumup:errorterm:cnear}, the last line of \eqref{eq:sumup:errorterm:snear}, the bulk integral of $\psi$ in \eqref{eq:sumup:errorterm:saway} and the horizon flux term in \eqref{eq:sumup:fluxterm} by the LHS of the just added two estimates. The parameters $C_1$, $\veps_2$, $\veps_3$ and $c_3$ are now fixed, and the LHS of the obtained estimate thus dominates over
\begin{align}
\label{eq:sumup:closeproof:LHS}
c\bigg({E}_{\tau}(\psi)+{E}_{\mathcal{H}^+(0,\tau)}(\psi)
+\int_{\mathcal{D}(0,\tau)} \mathbb{M}_{\text{deg}}(\psi)\bigg).
\end{align}
The remaining bulk integrals on the RHS of the added non-degenerate energy estimate and the Morawetz estimate in large radius region are bounded by
\begin{align}\label{eq:sumup:errorMoraredandnearinf}
C\bigg(\mathcal{E}_{F,1,0,\tau}(\psi)
+\int_{\mathcal{D}(0,\tau)}r^{-3+\delta}|F|^2\bigg).
\end{align} Note here that we actually used the following estimate from the Cauchy--Schwarz inequality for the last term in \eqref{eq:ImprovedMoraEstiLargerSWRWE}:
\begin{align*}
 \bigg|\int_{\mathcal{D}(0,\tau)\cap\{r\geq R_0-M\}}\Re\left(FX_w\overline{\psi}\right)\bigg|
\lesssim{}&\veps_4 \int_{\mathcal{D}(0,\tau)} \mathbb{M}_{\text{deg}}(\psi)
+\frac{1}{\veps_4}\int_{\mathcal{D}(0,\tau)}r^{-3+\delta}|F|^2,
\end{align*}
and chose $\veps_4$ small enough such that the term $\veps_4 \int_{\mathcal{D}(0,\tau)} \mathbb{M}_{\text{deg}}(\psi)$ is absorbed by \eqref{eq:sumup:closeproof:LHS}.

We now obtain an estimate in which the LHS dominates over the LHS of \eqref{eq:MorawetEnergyEstimateforAlmostScalarWave} and all the terms on the RHS are bounded by
\begin{align}\label{eq:sumup:mainerror:close}
C\veps^{-1}\sum\limits_{j=0,1}{E}_{0}(\phi_s^j)
+C\mathcal{E}_{F,1,0,\tau}(\psi)
+C\int_{\mathcal{D}(0,\tau)}r^{-3+\delta}|F|^2
\end{align}
except for terms in the following two categories:
\begin{enumerate}
\item\label{pt:1}
$\frac{1}{\varepsilon}\big(
\int_{\Sigma_{\tau-\varepsilon^{-1}}}\left|e(\psi)\right|
+e_0\sum\limits_{j=0,1}E_{\tau-\varepsilon^{-1}}(\phi^j)
\big)$ in \eqref{eq:sumup:Moraterm}, $\frac{\varepsilon}{\veps_2} \sum\limits_{j=0,1}{E}_{\tau-\varepsilon^{-1}}(\phi^j)$ in \eqref{eq:sumup:errorterm:cnear}, and
$\int_{\Sigma_{\tau-\varepsilon^{-1}}}
\left|e(\psi)\right|
+e_0\sum\limits_{j=0,1}{E}_{\tau-\varepsilon^{-1}}(\phi^j)$ in \eqref{eq:sumup:errorterm:caway};
\item\label{pt:2}
$\veps^{-1}\mathcal{E}_{F,e_0,0,\tau}(\psi)$ in \eqref{eq:sumup:Moraterm}
and $\veps_2^{-1}{\varepsilon}\mathcal{E}_{F,1,0,\tau}(\psi)$ in \eqref{eq:sumup:errorterm:cnear}.
\end{enumerate}
Consider first the three terms in Category \ref{pt:1}.
One takes $\psi=\phi_s^i$, $\tilde{e}=1$, $\tau_1=0$ and $\tau_2=\tau-{\varepsilon}^{-1}$ in \eqref{eq:SWRWEEnerEstiII} and sums over $i\in \{0,1\}$ gives
\begin{align*} \sum\limits_{i=0,1}E_{\tau-\varepsilon^{-1}}(\phi_s^i)
\lesssim{}& \sum_{i=0,1}\bigg(E_{0}(\phi_s^i)+\int_{\mathcal{D}(0,\tau)\cap [r_+,r_1]}|\partial\phi_s^i|^2+\mathcal{E}_{F,1,0,\tau}(\phi_s^i)\bigg).
\end{align*}
For the term $\int_{\Sigma_{\tau-\varepsilon^{-1}}}\left|e(\psi)\right|$, one obtains from estimate \eqref{eq:SWRWEEnerEstiII} with $\tilde{e}=e_0\sim \veps_0^2$, $\tau_2=\tau-\varepsilon^{-1}$ and $\tau_1=0$ that
\begin{align*}
\int_{\Sigma_{\tau-\varepsilon^{-1}}}\left|e(\psi)\right|
\lesssim {}&E_0(\psi)
+\veps_0^2\int_{\mathcal{D}(0,\tau)\cap [r_+,r_1]}|\partial\psi|^2
+\mathcal{E}_{F,1,0,\tau}(\psi).
\end{align*}
We combine the above two estimates and choose $\veps=\mu_i$, $\mu_i>0$ being arbitrary (but potentially small), then for sufficiently small $\veps_0\ll\mu_i$, all the terms in Category \ref{pt:1} are bounded by
\begin{align}
&C\mu_i^{-1}\left(E_0(\psi)
+\mathcal{E}_{F,1,0,\tau}(\psi)\right)\notag\\
&
+C\mu_i\sum_{j=0,1}\bigg(E_{0}(\phi_s^j)
+\int_{\mathcal{D}(0,\tau)}\mathbb{M}_{\text{deg}}(\phi_s^j)
+\mathcal{E}_{F,1,0,\tau}(\phi_s^j)\bigg).
\end{align}
Given the choice that $\veps= \mu_i$, the terms in Category \ref{pt:2} are clearly bounded by $C\mu_i^{-1}\mathcal{E}_{F,1,0,\tau}(\psi)$. In total, all the terms on the RHS have an upper bound
\begin{align}
&C\mu_i^{-1}\bigg(\sum_{j=0,1}E_{0}(\phi_s^j)
+\mathcal{E}_{F,1,0,\tau}(\psi)\bigg)
+C\int_{\mathcal{D}(0,\tau)}r^{-3+\delta}|F|^2\notag\\
&
+C\mu_i\sum_{j=0,1}\bigg(
\int_{\mathcal{D}(0,\tau)}\mathbb{M}_{\text{deg}}(\phi_s^j)
+\mathcal{E}_{F,1,0,\tau}(\phi_s^j)\bigg).
\end{align}
This closes the proof of the energy and Morawetz estimate \eqref{eq:MorawetEnergyEstimateforAlmostScalarWave} in the case that $\tau\geq 2\veps^{-1}\sim \mu_i^{-1}$.

For each $\psi=\phi_s^i$, the other case $\tau\leq 2\veps^{-1}= 2\mu_i^{-1}$ follows from a standard well-posedness argument of a general linear wave system.

\section{Proof of Theorem \ref{thm:EneAndMorEstiExtremeCompsNoLossDecayVersion2}}\label{sect:spin1case}

We derive energy and Morawetz estimates for $(\phi^0_s,\phi^1_s)$ in Section \ref{sect:estiphi01}, leaving the proof of Theorem \ref{thm:EneAndMorEstiExtremeCompsNoLossDecayVersion2} in Section \ref{sect:pfofmainthm}.

\subsection{Estimate for $(\phi^0,\phi^1)$}\label{sect:estiphi01}
For a large, positive constant $A$ to be fixed in the proof of Theorem \ref{thm:EneAndMorEstiExtremeCompsNoLossDecay}, we denote for the spin $s=+1$ component that
\begin{subequations}
\begin{align}
\widetilde{\mathbb{M}}(\phi^0_{+1},\phi^1_{+1})
&
=\widehat{\mathbb{M}}(r^{2-\delta}\phi^0_{+1})+
A\mathbb{M}_{\text{deg}}(\phi^1_{+1}),\\
{E}_{\tau}(\phi^0_{+1},\phi^1_{+1})&= {E}_{\tau}(r^{2-\delta}\phi^0_{+1})+A{E}_{\tau}(\phi^1_{+1}),\\
{E}_{\mathcal{H}^+(0,\tau)}(\phi^0_{+1},\phi^1_{+1})&=
{E}_{\mathcal{H}^+(0,\tau)}(r^{2-\delta}\phi^0_{+1})+
A{E}_{\mathcal{H}^+(0,\tau)}(\phi^1_{+1}),
\end{align}
and for spin $s=-1$ component that
\begin{align}
\widetilde{\mathbb{M}}(\phi^0_{-1},\phi^1_{-1})&=\mathbb{M}(\phi^0_{-1})+
A\mathbb{M}_{\text{deg}}(\phi^1_{-1})+\left|\nablaslash \phi^0_{-1}\right|^2,\\
{E}_{\tau}(\phi^0_{-1},\phi^1_{-1})&= {E}_{\tau}(\phi^0_{-1})+\int_{\Sigma_{\tau}}r\left|\nablaslash \phi_{-1}^0 \right|^2+A{E}_{\tau}(\phi^1_{-1}),\\
{E}_{\mathcal{H}^+(0,\tau)}(\phi^0_{-1},\phi^1_{-1})&=
{E}_{\mathcal{H}^+(0,\tau)}(\phi^0_{-1})+
A{E}_{\mathcal{H}^+(0,\tau)}(\phi^1_{-1}).
\end{align}
\end{subequations}
The main result for
$(\phi^0_{s},\phi^1_{s})$ is as follows.
\begin{thm}\label{thm:EneAndMorEstiExtremeCompsNoLossDecay}
In the DOC of a slowly rotating Kerr spacetime $(\mathcal{M},g=g_{M,a})$, assume $\phi^1_{s}$ and $\phi^0_{s}$ are
 defined as in \eqref{eq:DefOf Phi^0^1PosiSpin} for positive-spin component and in \eqref{eq:DefOf Phi^0^1NegaSpin} for negative-spin component of a regular\footnote{Recall here Definition \ref{def:regularandintegrable} of \textquotedblleft{regular.\textquotedblright}} solution $\mathbf{F}_{\alpha\beta}$ to the Maxwell equations \eqref{eq:MaxwellEqs}.
Then, for any $0<\delta<1/2$, there exist universal constants $\veps_0=\veps_0(M)$, $A(M,\delta)$ and $C_0=C(M,\delta,A)$ such that for all $|a|/M\leq \veps_0$,
the following estimate holds true for all $\tau\geq 0$:
\begin{align}\label{eq:MoraEstiFinal(2)Kerru^0andphi1BothSpinCompNoLossDecay}
 {E}_{\tau}(\phi^0_{s}, \phi^1_{s})
+{E}_{\mathcal{H}^+(0,\tau)}(\phi^0_{s},\phi^1_{s})
+\int_{\mathcal{D}(0,\tau)} \widetilde{\mathbb{M}}(\phi^0_{s},\phi^1_{s})
\leq
C_0{E}_{0}(\phi^0_{s},\phi^1_{s}).
\end{align}
\end{thm}
We will always suppress the subscript $s$ when it is clear which spin component we are treating.
Notice that the assumption \textquotedblleft{regular\textquotedblright} for the solution to the Maxwell equations implies that $(\phi^0,\phi^1)$ is a smooth solution to the linear wave system \eqref{eq:ReggeWheeler Phi^01Kerr} or \eqref{eq:ReggeWheeler Phi^01KerrNega} and vanishes near spatial infinity.
Hence, we can perform the proof in Section \ref{Sect:ProofofTheoremOnSWRWKerr} and apply the result in Theorem \ref{thm:MoraEstiAlmostScalarWave} to each separate equation of these linear wave systems.

We shall now prove this theorem for spin $+1$ and $-1$ components in Section \ref{sec:positivespin1} and Section \ref{sec:NegativeSpin1},
respectively. Note the constants $\mu_0\ll 1$, $\mu_1\ll 1$ and $A\gg 1$ are to be fixed. For convenience, we use the notation $G_1\lesssim_{\veps_0} G_2$ for two functions in the region $\mathcal{D}(0,\tau)$ if there exists a universal constant $C$ and a constant $\tilde{C}=C(\mu_0^{-1},\mu_1^{-1},A)$ such that
\begin{align}\label{def:AlesssimaB}
G_1 \leq {}& CG_2+\tilde{C}{E}_{0}(\phi^0_s,\phi^1_s)\notag\\
&+\tilde{C}\veps_0\bigg({E}_{\tau}(\phi^0_s,\phi^1_s)
+{E}_{\mathcal{H}^+(0,\tau)}(\phi^0_s,\phi^1_s)
+\int_{\mathcal{D}(0,\tau)}
\widetilde{\mathbb{M}}(\phi^0_s,\phi^1_s)\bigg).
\end{align}

\subsubsection{Spin $+1$ component}\label{sec:positivespin1}

It is manifest from the Cauchy--Schwarz inequality that
\begin{align}\label{eq:ErrorTermsControlphi0Posispin}
\mathcal{E}(F_{+1}^0)\lesssim_{\veps_0} \sqrt{A}\int_{\mathcal{D}(0,\tau)}\mathbb{M}_{\text{deg}}(\phi^1) +\frac{1}{\sqrt{A}}\int_{\mathcal{D}(0,\tau)}
\widehat{\mathbb{M}}(r^{2-\delta}\phi^0),
\end{align}
and all terms in $\mathcal{E}(F_{+1}^1)$ are bounded by $\tfrac{Ca^2}{M^2}\int_{\mathcal{D}(0,\tau)}
\widetilde{\mathbb{M}}(\phi^0,\phi^1)$ except for the term
\begin{align}\label{eq:errortermnew}
\hspace{4ex}&\hspace{-4ex}\bigg|\int_{\mathcal{D}(0,\tau)}
\Re\left(\Sigma^{-1}F_{+1}^1\partial_{t^*}
\overline{\phi^1}\right)\bigg|\notag\\
\lesssim{}&\bigg|\int_{\mathcal{D}(0,\tau)}
\Re\left(\frac{a^2(\PR)}{(\R)^2\Sigma}\phi^0\partial_{t^*}
\overline{\phi^1}\right)\bigg|\notag\\
& + \bigg|\int_{\mathcal{D}(0,\tau)}\frac{a(r^2-a^2)}{\Sigma(\R)}
\Re\left(\partial_{\phi}\phi^0\partial_{t}
\overline{\phi^1}\right)\bigg|.
 \end{align}
We use an integration by parts in $t^*$ for the first term on the RHS of \eqref{eq:errortermnew} and find it is bounded by
\begin{align*}
\frac{C|a|}{M}\bigg({E}_{0}(\phi^0,\phi^1)
+{E}_{\tau}(\phi^0,\phi^1)
+{E}_{\mathcal{H}^+(0,\tau)}(\phi^0,\phi^1)
+\int_{\mathcal{D}(0,\tau)}\widetilde{\mathbb{M}}(\phi^0,\phi^1)\bigg).
\end{align*}
As to the other term on the RHS of \eqref{eq:errortermnew}, we split it into three sub-integrals with $r_+ < \check{r}_2 < r_{\text{trap}}^- \leq r_{\text{trap}}^+ < R_4<\infty$ to be chosen:
$$
\bigg|\bigg(\int_{\mathcal{D}(0,\tau)\cap[r_+, \check{r}_2]}+\int_{\mathcal{D}(0,\tau)\cap[R_4, \infty)}
+\int_{\mathcal{D}(0,\tau)\cap[\check{r}_2,R_4]}\bigg)
\frac{a(r^2-a^2)}{\Sigma(\R)}
\Re(\partial_{\phi}\phi^0\partial_{t}
\overline{\phi^1})\bigg|.
 $$
The first two sub-integrals are controlled by $\tfrac{Ca^2}{M^2}\int_{\mathcal{D}(0,\tau)}
\widetilde{\mathbb{M}}(\phi^0,\phi^1)$ directly.
We substitute
\begin{equation*}
\partial_t \phi^1=
(\R)^{-1}\left(\Delta Y\phi^1-a\partial_{\phi}\phi^1+(\R)\partial_{r^*}\phi^1\right)
\end{equation*}
into the third sub-integral and find it is bounded by
\begin{align}\label{eq:ControlInTrappingRegionPosiSpin1}
&\bigg|\int_{\mathcal{D}(0,\tau)\cap[\check{r}_2,R_4]}
\frac{a\Delta(r^2-a^2)}{\Sigma(\R)^{5/2}}\Re\left(\partial_{\phi}
({\sqrt{\R}\phi^0})Y\overline{\phi^1}\right)\bigg|\notag\\
&+\bigg|\int_{\mathcal{D}(0,\tau)\cap[\check{r}_2,R_4]}
\frac{a^2(r^2-a^2)}{\Sigma(\R)^2}Y\Big(\big|
\partial_{\phi}(\sqrt{\R}\phi^0)\big|^2\Big)\bigg|\notag\\
& +\bigg|\int_{\mathcal{D}(0,\tau)\cap[\check{r}_2,R_4]}
\frac{a(r^2-a^2)}{\Sigma(\R)}
\Re\left(\partial_{\phi}
\overline{\phi^0}\partial_{r^*}\phi^1\right)\bigg|
\lesssim_{\veps_0} {}
0.
\end{align}
In proving the above inequality, we applied integration by parts to the first two lines and controlled the boundary terms at $R_4$ and $\check{r}_2$ by appropriately choosing these two radius parameters such that these boundary terms are bounded via an average of integration by $\tfrac{C|a|}{M} \int_{\mathcal{D}(0,\tau)}\widetilde{\mathbb{M}}(\phi^0,\phi^1)$.
In conclusion,
\begin{align}\label{eq:ErrorTermsControlphi1Posispin}
\mathcal{E}(F_{+1}^1)\lesssim_{\veps_0} 0.
\end{align}

Applying estimate \eqref{eq:MorawetEnergyEstimateforAlmostScalarWave} to $\psi=\phi^i_{+1}$ and $F=F_{+1}^i$ $(i=0,1)$ and combined with estimates \eqref{eq:MoraInftyr2minusdeltaphi0}, \eqref{eq:ErrorTermsControlphi0Posispin} and \eqref{eq:ErrorTermsControlphi1Posispin}, we obtain
\begin{subequations}
\begin{align}\label{eq:EnerMoraphi0V1}
\hspace{6ex}&\hspace{-6ex} {E}_{\tau}(r^{2-\delta}\phi^0)
+{E}_{\mathcal{H}^+(0,\tau)}(r^{2-\delta}\phi^0)
+\int_{\mathcal{D}(0,\tau)} \widehat{\mathbb{M}}_{\text{deg}}(r^{2-\delta}\phi^0)\notag\\
\lesssim {}&
C(\mu_0^{-1})\left({E}_{0}(r^{2-\delta}\phi_{+1}^0)+
E_{0}(\phi_{+1}^1)\right)+\mu_0^{-1}
\mathcal{E}(F_{+1}^0)
\notag\\
&
+\mu_0\sum_{j=0,1}\bigg(
\int_{\mathcal{D}(0,\tau)}\mathbb{M}_{\text{deg}}(\phi_{+1}^j)
+\mathcal{E}(F_{+1}^j)\bigg)\notag\\
\lesssim_{\veps_0}{}&\mu_0^{-1}\bigg(\sqrt{A}\int_{\mathcal{D}(0,\tau)}\mathbb{M}_{\text{deg}}(\phi^1) +\frac{1}{\sqrt{A}}\int_{\mathcal{D}(0,\tau)}
\widehat{\mathbb{M}}(r^{2-\delta}\phi^0)\bigg)\notag\\
&+\mu_0\sum_{j=0,1}
\int_{\mathcal{D}(0,\tau)}\mathbb{M}_{\text{deg}}(\phi_{+1}^j),\\
\label{eq:EnerMoraphi1V1}
\hspace{6ex}&\hspace{-6ex} {E}_{\tau}(\phi^1)
+{E}_{\mathcal{H}^+(0,\tau)}(\phi^1)
+\int_{\mathcal{D}(0,\tau)} {\mathbb{M}}_{\text{deg}}(\phi^1)\notag\\
\lesssim {}&C(\mu_1^{-1})\left({E}_{0}(r^{2-\delta}\phi_{+1}^0)+
E_{0}(\phi_{+1}^1)\right)
+\mu_1^{-1}\mathcal{E}(F_{+1}^1)
\notag\\
&
+\mu_1\sum_{j=0,1}\bigg(
\int_{\mathcal{D}(0,\tau)}\mathbb{M}_{\text{deg}}(\phi_{+1}^j)
+\mathcal{E}(F_{+1}^j)\bigg)\notag\\
\lesssim_{\veps_0}{}&\mu_1\bigg(\sqrt{A}\int_{\mathcal{D}(0,\tau)}\mathbb{M}_{\text{deg}}(\phi^1) +\frac{1}{\sqrt{A}}\int_{\mathcal{D}(0,\tau)}
\widehat{\mathbb{M}}(r^{2-\delta}\phi^0)\bigg)\notag\\
&+\mu_1\sum_{j=0,1}
\int_{\mathcal{D}(0,\tau)}\mathbb{M}_{\text{deg}}(\phi_{+1}^j).
\end{align}
\end{subequations}
Adding an $A$ multiple of \eqref{eq:EnerMoraphi1V1} to estimate \eqref{eq:EnerMoraphi0V1}, one deduces
\begin{align}\label{eq:EnerMoraphi01addedBulkTerm}
\hspace{6ex}&\hspace{-6ex} {E}_{\mathcal{H}^+(0,\tau)}(\phi^0,\phi^1)
+{E}_{\tau}(\phi^0,\phi^1)
+\int_{\mathcal{D}(0,\tau)}
\widetilde{\mathbb{M}}(\phi^0,\phi^1)\notag\\
\lesssim_{\veps_0} {}& (\mu_0^{-1/2}A^{-1/2}+\mu_0+\mu_1\sqrt{A}
+A\mu_1)\widehat{\mathbb{M}}(r^{2-\delta}\phi^0_{+1})\notag\\
&
+
(\mu_0^{-1}\sqrt{A}+\mu_0+A^{3/2}\mu_1+A\mu_1)\mathbb{M}_{\text{deg}}(\phi^1_{+1}).
\end{align}
Here, we have made use of the fact that
\begin{equation}\label{eq:equivalentoftwobulktermsinMorawesti}
\int_{\mathcal{D}(0,\tau)}\widetilde{\mathbb{M}}(\phi^0,\phi^1)\sim
\int_{\mathcal{D}(0,\tau)}\left(\widehat{\mathbb{M}}_{\text{deg}}
(r^{2-\delta}\phi^0)+
A\mathbb{M}_{\text{deg}}(\phi^1)\right).
\end{equation}
In the trapped region, $\widehat{\mathbb{M}}_{\text{deg}}(r^{2-\delta}\phi^0)+
A\mathbb{M}_{\text{deg}}(\phi^1)$ bounds over $A|Y\phi^0-\phi^0/\sqrt{\R}|^2$, $|\partial_{r^*}\phi^0|^2$ and $|\phi^0|^2$,
and then over $|\phi^0|^2$, $|Y(\phi^0)|^2$ and $|H(\phi^0)|^2$, $H=\partial_t+a/(r^2+a^2)\partial_{\phi}$ being a globally timelike vector field in the interior of $\mathcal{D}$ with
$-g(H,H)=\Delta\Sigma/(r^2+a^2)^2.$
Relation \eqref{eq:equivalentoftwobulktermsinMorawesti} then follows from elliptic estimates. By requiring
$\mu_0\ll 1$, $A\mu_0^{2}\gg 1$, and $\mu_1A\ll 1$, the RHS of \eqref{eq:EnerMoraphi01addedBulkTerm} can be absorbed by the LHS. Furthermore, once these constants are fixed, taking $\veps_0$ sufficiently small allows us to absorb the energy and Morawetz terms with $\tilde{C}\veps_0$ coefficients which are implicit in $\lesssim_{\veps_0}$, and this proves the $s=+1$ case of \eqref{eq:MoraEstiFinal(2)Kerru^0andphi1BothSpinCompNoLossDecay}.

\subsubsection{Spin $-1$ component}\label{sec:NegativeSpin1}
We first state a lemma controlling $|\nablaslash \phi^0_{-1}|$ by $|\nablaslash\phi^1_{-1}|$.
\begin{lemma}
\label{lem:estiphi0nega}
In a fixed subextremal Kerr spacetime $(\mathcal{M},g_{M,a})$ $(|a|<M)$ and given any $R\geq 5M$, the following estimate holds for spin $-1$ component:
\begin{align}\label{eq:estiphi0byphi1negafinal1}
\hspace{4ex}&\hspace{-4ex}\intRinfty |\nablaslash\phi^0|^2 +\int_{\Sigma_{\tau}\cap [R,\infty)}r|\nablaslash\phi^0|^2\notag\\
\lesssim {}&
\int_{\mathcal{D}(0,\tau)\cap [R-M,\infty)} {r}^{-1}{|\nablaslash\phi^1|^2}
+\intcut {r}^{-1}{|\nablaslash \phi^0|^2}\notag\\
&+\int_{\Sigma_0\cap [R-M, \infty)}r|\nablaslash \phi^0|^2.
\end{align}
\end{lemma}
\begin{proof}
We start with an identity that for the cutoff function $\chi_R(r)$,
any real value $\beta$ and $\nablaslash_i$ $(i=1,2,3)$ as defined in \eqref{SpinWeightedAngularDerivaBasisOnSphere}:
\begin{align}\label{eq:estiphi1negaiden1}
\hspace{6ex}&\hspace{-6ex}V\left(\chi_R r^{\beta}(\R)|r\nablaslash_i\phi^0|^2\right)
-\beta \chi_R r^{\beta-1}(\R)|r\nablaslash_i\phi^0|^2
\notag\\
&-\partial_r\chi_R r^{\beta}(\R)|r\nablaslash_i\phi^0|^2={}\chi_R r^{2+\beta}\Re\big(\nablaslash_i\phi^0
\overline{\nablaslash_i\phi^1}\big).
\end{align}
Integrating \eqref{eq:estiphi1negaiden1} over $\mathcal{D}(0,\tau)$ with the measure $d\check{V}=drdt^*\sin\theta d\theta d\phi^*$
for $\beta=-1$, and applying Cauchy--Schwarz to the last term, it is manifest that estimate \eqref{eq:estiphi0byphi1negafinal1} follows from summing over $i=1,2,3$.
\end{proof}

An application of the Cauchy--Schwarz inequality gives
\begin{align}\label{eq:ErrorTermsControlphi0PosiNega}
\mathcal{E}(F_{-1}^0)\lesssim_{\veps_0} {}& \sqrt{A}\int_{\mathcal{D}(0,\tau)}
\mathbb{M}_{\text{deg}}(\phi^1) +\frac{1}{\sqrt{A}}\int_{\mathcal{D}(0,\tau)}
{\mathbb{M}}(\phi^0),\\
\label{eq:ErrorTermsControlphi1PosiNega:v1}
\mathcal{E}(F_{-1}^1)\lesssim_{\veps_0}{}&
\bigg|\int_{\mathcal{D}(0,\tau)}{\Sigma}^{-1}\Re\big(F_{-1}^1 \partial_t\overline{\phi^1}\big)\bigg|.
\end{align}
For the last term in \eqref{eq:ErrorTermsControlphi1PosiNega:v1}, we have
\begin{align}\label{eq:remainingspin-1:phit}
\hspace{4ex}&\hspace{-4ex}
\bigg|\int_{\mathcal{D}(0,\tau)}\frac{1}{\Sigma}\Re\left(F_{-1}^1 \partial_t\overline{\phi^1}\right)\bigg|
\notag\\
\lesssim{}&
\bigg|\int_{\mathcal{D}(0,\tau)}
\Re\left(\frac{a^2(\PR)}{(\R)^2\Sigma}\phi^0\partial_{t^*}
\overline{\phi^1}\right)\bigg|\notag\\
& + \bigg|\int_{\mathcal{D}(0,\tau)}\frac{a(r^2-a^2)}{\Sigma(\R)}
\Re\left(\partial_{\phi}\phi^0\partial_{t}
\overline{\phi^1}\right)\bigg|.
\end{align}
By performing an integration by parts in $t^*$, the first integral term on the RHS is bounded by
\begin{align*}
\frac{Ca^2}{M^2}\bigg({E}_{0}(\phi^0,\phi^1)
+{E}_{\tau}(\phi^0,\phi^1)
+{E}_{\mathcal{H}^+(0,\tau)}(\phi^0,\phi^1)
+\int_{\mathcal{D}(0,\tau)}\widetilde{\mathbb{M}}(\phi^0,\phi^1)\bigg).
\end{align*}
As to the second integral term on the RHS of \eqref{eq:remainingspin-1:phit}, we split it into two sub-integrals with $r_+ < \check{r}_3 < r_{\text{trap}}^- $ to be chosen:
\begin{equation}
\bigg|\bigg(\int_{\mathcal{D}(0,\tau)\cap[r_+, \check{r}_3]}
+\int_{\mathcal{D}(0,\tau)\cap[\check{r}_3,\infty)}\bigg)
\frac{a(r^2-a^2)}{\Sigma(\R)}
\Re\left(\partial_{\phi}\phi^0\partial_{t}\overline{\phi^1}
\right)\bigg|,
\end{equation}
where the first sub-integral is clearly bounded by ${C|a|}M^{-1} \int_{\mathcal{D}(0,\tau)}\widetilde{\mathbb{M}}(\phi^0,\phi^1)$.
We substitute
$\partial_t \phi^1=
(\R)^{-1}\left(\Delta V\phi^1-a\partial_{\phi}\phi^1-(\R)\partial_{r^*}\phi^1\right)$
into the second sub-integral and find it is bounded by
\begin{align*}
&\bigg|\int_{\mathcal{D}(0,\tau)\cap[\check{r}_3,\infty]}
\frac{a\Delta(r^2-a^2)}{\Sigma(\R)^{5/2}}\Re\left(\partial_{\phi}
({\sqrt{\R}\phi^0})V\overline{\phi^1}\right)\bigg|\notag\\
&+\bigg|\int_{\mathcal{D}(0,\tau)\cap[\check{r}_3,\infty]}
\frac{a^2(r^2-a^2)}{\Sigma(\R)^2}V\Big(\big|
\partial_{\phi}\big(\sqrt{\R}\phi^0\big)\big|^2\Big)\bigg|
\notag\\
&+\bigg|\int_{\mathcal{D}(0,\tau)\cap[\check{r}_3,\infty]}
\frac{a(r^2-a^2)}{\Sigma(\R)}\Re\left(\partial_{\phi}
\overline{\phi^0}\partial_{r^*}\phi^1\right)\bigg|.
\end{align*}
Integrating by parts for the first two lines then shows
\begin{align}
\bigg|\int_{\mathcal{D}(0,\tau)\cap[\check{r}_3,\infty)}
2a\Sigma^{-1}
\Re\left(\partial_{\phi}\phi^0\partial_{t}
\overline{\phi^1}\right)\bigg|
\lesssim_{\veps_0} {}& 0.
\end{align}
Therefore,
\begin{align}
\label{eq:ErrorTermsControlphi1PosiNega}
\mathcal{E}(F_{-1}^1)\lesssim_{\veps_0}{}&0.
\end{align}

We use estimate \eqref{eq:MorawetEnergyEstimateforAlmostScalarWave} for $\psi=\phi^i_{-1}$ and $F=F_{-1}^i$ $(i=0,1)$, add a small multiple of estimate \eqref{eq:estiphi0byphi1negafinal1} with $R=R_0$ to estimate \eqref{eq:MorawetEnergyEstimateforAlmostScalarWave} with $(\psi,F)=(\phi^0_{-1},F_{-1}^0)$, and combine with estimates \eqref{eq:ErrorTermsControlphi0PosiNega} and \eqref{eq:ErrorTermsControlphi1PosiNega}, arriving at
\begin{subequations}
\begin{align}\label{eq:EnerMoraphi0V1Ne}
\hspace{6ex}&\hspace{-6ex} {E}_{\tau}(\phi_{-1}^0)
+{E}_{\mathcal{H}^+(0,\tau)}(\phi_{-1}^0)
+\int_{\Sigma_{\tau}}r\left|\nablaslash \phi_{-1}^0 \right|^2
+\int_{\mathcal{D}(0,\tau)} \left({\mathbb{M}}_{\text{deg}}(\phi_{-1}^0)+\left|\nablaslash \phi_{-1}^0\right|^2\right)\notag\\
\lesssim {}&
\mu_0^{-1}\left({E}_{0}(\phi_{-1}^0)+
E_{0}(\phi_{-1}^1)+
\mathcal{E}(F_{-1}^0)\right)
+\int_{\Sigma_{0}}r\left|\nablaslash \phi_{-1}^0 \right|^2
\notag\\
&
+\mu_0\sum_{j=0,1}\bigg(
\int_{\mathcal{D}(0,\tau)}\mathbb{M}_{\text{deg}}
(\phi_{-1}^j)
+\mathcal{E}(F_{-1}^j)\bigg)
+\int_{\mathcal{D}(0,\tau)}\mathbb{M}_{\text{deg}}(\phi_{-1}^1)\notag\\
\lesssim_{\veps_0}{}&\mu_0^{-1}
\bigg(\sqrt{A}\int_{\mathcal{D}(0,\tau)}
\mathbb{M}_{\text{deg}}(\phi_{-1}^1) +\frac{1}{\sqrt{A}}\int_{\mathcal{D}(0,\tau)}
{\mathbb{M}}(\phi_{-1}^0)\bigg)\notag\\
&+\mu_0\sum_{j=0,1}
\int_{\mathcal{D}(0,\tau)}\mathbb{M}_{\text{deg}}
(\phi_{-1}^j)
+\int_{\mathcal{D}(0,\tau)}
\mathbb{M}_{\text{deg}}(\phi_{-1}^1),\\
\label{eq:EnerMoraphi1V1Ne}
\hspace{6ex}&\hspace{-6ex} {E}_{\tau}(\phi_{-1}^1)
+{E}_{\mathcal{H}^+(0,\tau)}(\phi_{-1}^1)
+\int_{\mathcal{D}(0,\tau)} {\mathbb{M}}_{\text{deg}}(\phi_{-1}^1)\notag\\
\lesssim {}&\mu_1^{-1}\left({E}_{0}(\phi_{-1}^0)+
E_{0}(\phi_{-1}^1)+
\mathcal{E}(F_{-1}^1)\right)
\notag\\
&
+\mu_1\sum_{j=0,1}\bigg(
\int_{\mathcal{D}(0,\tau)}
\mathbb{M}_{\text{deg}}(\phi_{-1}^j)
+\mathcal{E}(F_{-1}^j)\bigg)\notag\\
\lesssim_{\veps_0}{}&
\mu_1\bigg(\sqrt{A}\int_{\mathcal{D}(0,\tau)}
\mathbb{M}_{\text{deg}}(\phi_{-1}^1) +\frac{1}{\sqrt{A}}\int_{\mathcal{D}(0,\tau)}
{\mathbb{M}}(\phi_{-1}^0)\bigg)\notag\\
&+\mu_1\sum_{j=0,1}
\int_{\mathcal{D}(0,\tau)}\mathbb{M}_{\text{deg}}
(\phi_{-1}^j).
\end{align}
\end{subequations}
The rest of the proof is the same as the argument in the last paragraph of Section \ref{sec:positivespin1}, and we omit it.

\subsection{Finishing the Proof of Theorem \ref{thm:EneAndMorEstiExtremeCompsNoLossDecayVersion2}}\label{sect:pfofmainthm}

Using \eqref{eq:expandformofLswaveop}, equation \eqref{eq:ReggeWheeler Phi^0KerrNega} of $\phi^0_{-1}$ can be rewritten as
\begin{align}\label{def:Qphi0phi1}
\tfrac{\Delta}{\R}Y\phi^1
={}&
\triangle_{\mathbb{S}^2}\phi^0
-2i\left(\tfrac{\cos\theta}{\sin^2\theta}\partial_{\phi}-a\cos\theta \partial_t\right)\phi^0
-\tfrac{1}{\sin^2\theta}\phi^0
+\tfrac{r\Delta}{(\R)^2}\phi^1\notag\\
&
+a^2\sin^2\theta\partial_{tt}^2\phi^0
+2a\partial_{t\phi}^2\phi^0
+\tfrac{a^2\Delta}{(\R)^2}\phi^0
-\tfrac{2ar}{\R}\partial_{\phi}\phi^0.
\end{align}
By multiplying $r^{-1}\overline{\phi^0}$ on both sides, taking the real part and integrating over $\Sigma_{\tau}\cap \{r\geq R_5\}$ $(\tau\geq 0)$ with large $R_5$ to be fixed, it follows
\begin{align}\label{eq:ControlOnemorerweightenergyStep1}
\int_{\Sigma_{\tau}}r\left|\nablaslash \phi^0\right|^2\lesssim {} E_{\tau}(\phi^0) +E_{\tau}(\phi^1) +a^2 \bigg|\int_{\Sigma_{\tau}\cap \{r\geq R_5\}}\frac{\sin^2\theta}{r}\Re(\partial_{tt}^2 \phi^0 \overline{\phi^0})\bigg|.
\end{align}
We substitute into the last integral the following identity
\begin{equation}
\partial_{tt}^2 =\left(\frac{\Delta}{\R}V-\frac{a\partial_{\phi}}{\R}
-\partial_{r^*}\right)\left(\frac{\Delta}{\R}V
-\frac{a\partial_{\phi}}{\R}
-\partial_{r^*}\right),
\end{equation}
use the replacement $V\phi^0= (\R)^{-1}\phi^1 -\frac{r}{\R}\phi^0$ and perform integration by parts, finally ending with
\begin{align}\label{eq:ControlOnemorerweightenergyStep2}
\bigg|\int_{\Sigma_{\tau}\cap \{r\geq R_5\}}\frac{\sin^2\theta}{r}\Re(\partial_{tt}^2 \phi^0 \overline{\phi^0})\bigg|\lesssim{} E_{\tau}(\phi^0) +E_{\tau}(\phi^1)+\int_{\Sigma_{\tau}\cap \{r=R_5\}} |\partial \phi^0|^2.
\end{align}
We can appropriately choose $R_5$ such that the last term is bounded by $CE_{\tau}(\phi^0)$.
From estimates \eqref{eq:ControlOnemorerweightenergyStep1}, \eqref{eq:ControlOnemorerweightenergyStep2} and \eqref{eq:MoraEstiFinal(2)Kerru^0andphi1BothSpinCompNoLossDecay} for spin $-1$ component, we conclude by dropping the terms $\inttau |\nablaslash \phi^0_{-1}|^2$ and $\int_{\Sigma_{\tau}}r|\nablaslash \phi_{-1}^0 |^2$ on LHS of \eqref{eq:MoraEstiFinal(2)Kerru^0andphi1BothSpinCompNoLossDecay} that
\begin{align}\label{eq:EnergyMoraNegaSpinLastSecond}
\hspace{4ex}&\hspace{-4ex}\sum_{i=1,2}\left({E}_{\tau}(\phi^i_{-1})+ {E}_{\mathcal{H}^+(0,\tau)}(\phi^i_{-1})\right)
+\int_{\mathcal{D}(0,\tau)} \left(\mathbb{M}(\phi^0_{-1}) + \mathbb{M}_{\text{deg}}(\phi^1_{-1})\right)\notag\\
\lesssim {}&{E}_{0}(\phi^0_{-1})+E_0(\phi^1_{-1}).
\end{align}

We have from \eqref{eq:DefOf Phi^0^1NegaSpin} that
\begin{align}
\phi_{-1}^1={}& \sqrt{\R}V\left((\R)^{-1/2}\Delta\psi_{[-1]}\right)\notag\\
={}&-\Delta Y\psi_{[-1]}+2[(r^2+a^2)\partial_t +a\partial_{\phi}]\psi_{[-1]}+\Big(r-\tfrac{2a^2M}{\R}\Big)\psi_{[-1]}.
\end{align}
Hence, from \eqref{eq:TME0orderNega}, the equation for $\psi_{[-1]}$ can be rewritten as
\begin{align}\label{eq:TMEregularNPnegacomp}
\hspace{4ex}&\hspace{-4ex}\left(\Sigma \Box_g-2i\left(\tfrac{\cos\theta}{\sin^2 \theta}\partial_{\phi}-a\cos \theta \partial_t\right)-\cot^2\theta\right)\psi_{[-1]} \notag\\
={}&\Big(\tfrac{3}{2}-\tfrac{5a^2M}{2r(\R)}\Big)\psi_{[-1]}+
\left(\tfrac{4(r-M)r-5\Delta}{2r}Y
+r\partial_t\right)\psi_{[-1]}\notag\\
&+\tfrac{5}{r}\left(a^2\partial_t+a\partial_{\phi}\right)\psi_{[-1]}
-\tfrac{5}{2r}\phi_{-1}^1.
\end{align}
For small enough $|a|/M$, the coefficient of $\psi_{[-1]}$ term on the RHS is positive and its derivative with respect to $r$ is nonnegative, and the term $((4(r-M)r-5\Delta)/(2r)Y
+r\partial_t)\psi_{[-1]}$ is close to a positive multiple of $N_{\chi_0}\psi_{[-1]}$ when $r$ is sufficiently close to $r_+$ from the choice of $N_{\chi_0}$ in Proposition \ref{prop:RedShiftEstiInhomoSWRWE}.
Therefore, arguing the same as in Section \ref{sect:Redshift}, there exists a radius constant $\check{r}_4\leq r_0$ close enough to $r_+$ such that the following red-shift estimate near $\mathcal{H}^+$ holds for $\psi_{[-1]}$ for all $\tau\geq0$ for $|a|/M$ sufficiently small:
\begin{align}\label{eq:RedShiftEstiInhomoSWRWERegularNegaSpinNPcomp}
\hspace{4ex}&\hspace{-4ex}E_{\mathcal{H}^{+}(0,\tau)}(\psi_{[-1]})
+\int_{\Sigma_{\tau}\cap\{r\leq \check{r}_4\}}\left|\partial\psi_{[-1]}\right|^2 +\int_{\mathcal{D}(0,\tau)\cap\{r\leq \check{r}_4\}}\left|\partial\psi_{[-1]}\right|^2
\notag\\
\lesssim {}&
\int_{\Sigma_{0}\cap\{r\leq r_1\}}\left|\partial\psi_{[-1]}\right|^2
+\int_{\mathcal{D}(0,\tau)\cap\{\check{r}_2\leq r\leq r_1\}}\left|\partial\psi_{[-1]}\right|^2\notag\\
&+\int_{\mathcal{D}(0,\tau)\cap\{r_+\leq r\leq r_1\}}
\left|\phi_{-1}^1\right|^2.
\end{align}
Adding a small amount of estimate \eqref{eq:RedShiftEstiInhomoSWRWERegularNegaSpinNPcomp} to \eqref{eq:EnergyMoraNegaSpinLastSecond} and taking \eqref{def:regularNPComps} into account, this finishes the proof of \eqref{eq:MoraEstiFinal(2)KerrRegularpsiBothSpinComp} for $(\varphi^0_2,\varphi^1_2)$ in \eqref{def:varphi01negative}. The proof of \eqref{eq:MoraEstiFinal(2)KerrRegularpsiBothSpinComp} for $(\varphi^0_0,\varphi^1_0)$ in \eqref{def:varphi01positive} follows easily from Theorem \ref{thm:EneAndMorEstiExtremeCompsNoLossDecay} for $\phi^0_{+1}=\psi_{[+1]}/(\R)$ and the definition of $\widetilde{\Phi_0}$ in \eqref{def:regularNPComps}.

As to the high-order estimates, one just needs to consider the case $n=1$ by induction. Commute the Killing vector field $T$ with \eqref{eq:spinweightedwaveeqAssumption}, $\chi_0Y$ with \eqref{eq:spinweightedwaveeqAssumption} for spin $+1$ component and \eqref{eq:TMEregularNPnegacomp} for spin $-1$ component, then it follows easily from the red-shift commutation property  \cite[Prop.5.4.1]{dafermos2010decay} and elliptic estimates that the estimate \eqref{eq:MoraEstiFinal(2)KerrRegularpsiBothSpinCompHighOrder} for $n=1$ holds true by replacing $\varphi^k_0$ $(k=0,1)$ by  $\hat{\varphi}^k_0$ $(k=0,1)$ defined by
\begin{align}
&\hat{\varphi}^0_0=(\R)^{-\delta/2}\psi_{[+1]}, & \hat{\varphi}^1_0&=\sqrt{\R}Y((\R)^{-1/2}\psi_{[+1]}),
\end{align}
or $\varphi^k_2$ $(k=0,1)$ by  $\hat{\varphi}^k_2$ $(k=0,1)$ defined by
\begin{align}
&\hat{\varphi}^0_2=\psi_{[-1]}, & \hat{\varphi}^1_2&=(\R)^{1/2}V((\R)^{-1/2}\Delta\psi_{[-1]}).
\end{align}
Due to the fact that the following equivalence relations hold true for $j=0, 2$, any $\tau\geq 0$ and any nonnegative integer $n$:
\begin{subequations}\label{eq:equirelabetweenenergys}
\begin{align}
\sum_{|i|\leq n}\sum_{k=0,1}{E}_{\tau}(\partial^i\hat{\varphi}^k_j)
\sim &
\sum_{|i|\leq n}\sum_{k=0,1}{E}_{\tau}(\partial^i\varphi^k_j),\\
\sum_{|i|\leq n}\sum_{k=0,1}{E}_{\mathcal{H}^+(0,\tau)}(\partial^i\hat{\varphi}^k_j)
\sim &
\sum_{|i|\leq n}\sum_{k=0,1}{E}_{\mathcal{H}^+(0,\tau)}
(\partial^i\varphi^k_j),
\end{align}
and
\begin{align}
&\sum_{|i|\leq n}\int_{\mathcal{D}(0,\tau)} \left(\mathbb{M}_{\text{deg}}(\partial^i\hat{\varphi}^1_0)
+\widehat{\mathbb{M}}(\partial^i\hat{\varphi}^0_0)\right)\notag\\
&\sim {}
\sum_{|i|\leq n}\int_{\mathcal{D}(0,\tau)} \left(\mathbb{M}_{\text{deg}}(\partial^i\varphi^1_0)
+\widehat{\mathbb{M}}(\partial^i\varphi^0_0)\right),\\
&\sum_{|i|\leq n}\int_{\mathcal{D}(0,\tau)} \left(\mathbb{M}_{\text{deg}}(\partial^i\hat{\varphi}^1_2)
+{\mathbb{M}}(\partial^i\hat{\varphi}^0_2)\right)\notag\\
&\sim {}
\sum_{|i|\leq n}\int_{\mathcal{D}(0,\tau)} \left(\mathbb{M}_{\text{deg}}(\partial^i\varphi^1_2)
+{\mathbb{M}}(\partial^i\varphi^0_2)\right),
\end{align}
\end{subequations}
the high-order estimates \eqref{eq:MoraEstiFinal(2)KerrRegularpsiBothSpinCompHighOrder} follow straightforwardly.

\appendix

\section{Commutator relations}\label{appx:commutatorrelation}

\begin{prop}
\label{prop:commutatorwaveandYV}
For any scalar $\psi$ with spin weight $s\in \{+1,-1\}$, we have the following commutators
\begin{align}
[\Ls, \curlV]\psi={}&
-\curlV\left((\R)\partial_r\left(\tfrac{\Delta}{(\R)^2}\right)\curlV\psi\right)
-\tfrac{4ar}{\R}\partial_{\phi}\curlV \psi \notag\\
&
-\tfrac{2a^3}{\R}\partial_{\phi}\psi-a^2(\R)\partial_r\left(\tfrac{\Delta}{(\R)^2}\right)\psi,
\label{eq:wavecommutatorwithrVr}\\
[\Ls, \curlY]\psi={}&
\curlY\left((\R)\partial_r\left(\tfrac{\Delta}
{(\R)^2}\right)\curlY\psi\right)
+\tfrac{4ar}{\R}\partial_{\phi}\curlV \psi \notag\\
&
-\tfrac{2a^3}{\R}\partial_{\phi}\psi-a^2(\R)\partial_r\left(\tfrac{\Delta}{(\R)^2}\right)\psi.
\label{eq:wavecommutatorwithrYr}
\end{align}
\end{prop}

\begin{proof}
Expand $\mathbf{L}_s \psi$ into the form of
\begin{align}
\mathbf{L}_s\psi
={}&-\sqrt{\R}Y\left(\tfrac{\Delta}{\R} V(\sqrt{\R}\psi)\right)
+\tfrac{2ar}{\R}\partial_{\phi}\psi
+\tfrac{a^2\Delta}{(\R)^2}\psi\notag\\
&+\left(\tfrac{1}{\sin{\theta}} \partial_{\theta}(\sin \theta \partial_{\theta})+\tfrac{\partial_{\phi\phi}^2}{\sin^2\theta}
+\tfrac{2is\cos\theta}{\sin^2 \theta}\partial_{\phi}-\tfrac{s^2}{\sin^2 \theta}\right)\psi\notag\\
&
+2a\partial_{t\phi}^2\psi+a^2 \sin^2 \theta\partial_{tt}^2\psi
-2ias\cos\theta \partial_t\psi.
\label{eq:expandformofLswaveop}
\end{align}
We prove the commutator relation \eqref{eq:wavecommutatorwithrVr}, and commutator \eqref{eq:wavecommutatorwithrYr} is manifest from \eqref{eq:wavecommutatorwithrVr} by letting $s\to -s$, $t\to -t$ and $\phi \to -\phi$ (hence $\partial_t\to -\partial_t$, $\partial_{\phi} \to -\partial_{\phi}$ and $V \to -Y$). We calculate the commutators between each term and $\curlV$. The last two lines of \eqref{eq:expandformofLswaveop} commute with $\curlV$, and hence, their commutators vanish. We calculate for the last two terms on the first line the commutators
\begin{align}
\hspace{4ex}&\hspace{-4ex}
[\tfrac{2ar}{\R}\partial_{\phi},\curlV]\psi
={}-\tfrac{2a^3}{\R}\partial_{\phi}\psi,
\label{eq:secondcommutatorrela}\\
\hspace{4ex}&\hspace{-4ex}
[\tfrac{a^2\Delta}{(\R)^2},\curlV]\psi
={}-a^2(\R)\partial_r\left(\tfrac{\Delta}{(\R)^2}\right)\psi.
\label{eq:thirdcommutatorrela}
\end{align}
The remaining commutator is
\begin{align}\label{eq:firstcommutatorrela}
\hspace{4ex}&\hspace{-4ex}[-\sqrt{\R}Y\left(\tfrac{\Delta}{\R} V(\sqrt{\R})\right), \curlV]\psi\notag\\
={}&-\sqrt{\R}Y\left(\tfrac{\Delta}{\R}V\left((\R)V(\sqrt{\R}\psi)\right)\right)\notag\\
&+\sqrt{\R}V\left((\R)Y\left(\tfrac{\Delta}{(\R)^2}(\R)V(\sqrt{\R}\psi)\right)\right)\notag\\
={}&-\tfrac{(\R)^{3/2}}{\Delta}\left[\tfrac{\Delta}{\R}Y,\tfrac{\Delta}{\R}V\right]\left((\R) V(\sqrt{\R}\psi)\right)\notag\\
&-\sqrt{\R}V\left((\R)^2\partial_r\left(\tfrac{\Delta}
{(\R)^2}\right)V(\sqrt{\R}\psi)\right).
\end{align}
Since $\big[\tfrac{\Delta}{\R}Y,\tfrac{\Delta}{\R}V\big]=
\tfrac{4ar\Delta}{(\R)^3}\partial_{\phi}$, relation \eqref{eq:wavecommutatorwithrVr} follows from the above commutator relations \eqref{eq:secondcommutatorrela}--\eqref{eq:firstcommutatorrela}.
\end{proof}

\section{Energy and Morawetz estimates for inhomogeneous SWFIE on Schwarzschild}\label{appe:MoraFS1FIeqSchw}

\subsection{Energy estimate}

Multiplying \eqref{eq:SWRWReducedSchw} by $T\overline{\psi_{m\ell}}=\partial_{t}\overline{\psi_{m\ell}}$ and taking the real part, we arrive at an identity:
\begin{align*}
-\Re\left(F_{m\ell}\partial_{t}\overline{\psi_{m\ell}}\right)&=
-\partial_r\left(\Re\left(\Delta\partial_r\psi_{m\ell}
\partial_t\overline{\psi_{m\ell}}\right)\right)\notag\\
+\half\partial_t&\left( \tfrac{r^4}{\Delta}|\partial_t \varphi|^2+\Delta|\partial_r\psi_{m\ell}|^2
+ \ell(\ell+1)|\psi_{m\ell}|^2-\tfrac{2M}{r}|\psi_{m\ell}|^2\right).
\end{align*}
Since $\ell\geq |s|=1$ and $\ell(\ell+1)\geq 2$, it follows
 \begin{equation}
\ell(\ell+1)-\tfrac{2M}{r}\geq \tfrac{1}{2} \ell(\ell+1).
\end{equation}
Summing over $m$ and $\ell$, applying the identity \eqref{eq:IdenOfEigenvaluesAndAnguDeriSchw} and finally integrating with respect to the measure $dt^*dr$ over $\{(t^*,r)|0\leq t^*\leq \tau, 2M\leq r<\infty\}$, we have the following energy estimate:
\begin{align}\label{eq:enerestipartialtSchw}
E_{\tau}^T(\psi)+E_{\mathcal{H}^+(0,\tau)}^T(\psi)\leq C\bigg(E_{0}^T(\psi)
+\int_{\mathcal{D}(0,\tau)}\left|r^{-2}
\Re\left(F\partial_t\overline{\psi}\right)\right|\bigg).
\end{align}
Here, for any $\tau\geq0$, in global Kerr coordinates,
\begin{subequations}\label{def:energyforTS1Schw}
\begin{align}
E_{\tau}^T(\psi)={}&
\int_{\Sigma_{\tau}}\left(|\partial_{t^*}\psi|^2+|\nablaslash \psi|^2+{\Delta}r^{-2}|\partial_r\psi|^2\right),
\end{align}
and in ingoing Kerr coordinates,
\begin{align}
E_{\mathcal{H}^+(0,\tau)}^T(\psi)
={}&\int_{\mathcal{H}^+(0,\tau)}
\left(|\partial_{v}\psi|^2+{\Delta}r^{-2}|\nablaslash \psi|^2\right).
\end{align}
\end{subequations}

\subsection{Morawetz estimate}\label{proof of integrated decay}

In this section,  we use the multipliers as in \cite{larsblue15hidden,Jinhua17LinGraSchw} and prove a Morawetz estimate for \eqref{eq:spinweightedwaveeqAssumptionSchw}.
We multiply \eqref{eq:SWRWReducedSchw} by a general radial multiplier
\begin{align}\label{eq:XchoiceSchw}
X(\overline{\psi_{m\ell}})=\hat{f}\partial_r \overline{\psi_{m\ell}}+\hat{q}\overline{\psi_{m\ell}},
\end{align}
take the real part and obtain
\begin{align}\label{eq:MoraEstidegSchw}
\hspace{4ex}&\hspace{-4ex}-\Re\left(X(\psi_{m\ell})\overline{F_{m\ell}}\right)\notag\\
={}&\frac{1}{2}\partial_r\left(\hat{f}\left[
\left(\ell(\ell+1)-\tfrac{2M}{r}\right)|\psi_{m\ell}|^2
-\tfrac{r^4}{\Delta}|\partial_t\psi_{m\ell}|^2-\Delta |\partial_r\psi_{m\ell}|^2\right]\right)\notag\\
&
+\tfrac{1}{2}\partial_r\left(
\left(\partial_r(\Delta\hat{q})-2\hat{q}(r-M)
-r^{-1}B^r(r)\right)
|\psi_{m\ell}|^2-\Re\left(2\Delta \hat{q}\overline{\psi_{m\ell}}\partial_r\psi_{m\ell}
\right)\right)
\notag\\
&+\partial_t\left(\Re\left(\tfrac{r^4}{\Delta}X(\psi_{m\ell})
\partial_t\overline{\psi_{m\ell}}\right)\right)
+B(\psi_{m\ell}).
\end{align}
The bulk term
\begin{align}
B(\psi_{m\ell})={}&B^t(r)|\partial_t\psi_{m\ell}|^2+B^0(r)|\psi_{m\ell}|^2\notag\\
&+r^{-2}B^r(r)|\partial_r(r\psi_{m\ell})|^2
+B^{\ell}(r)\left(\ell(\ell+1)|\psi_{m\ell}|^2\right),
\end{align}
with
\begin{align}
B^t(r)={}&\tfrac{1}{2}\partial_r
(r^4\Delta^{-1}\hat{f})
-\hat{q}r^4\Delta^{-1}\notag\\
B^r(r)={}&\tfrac{1}{2}\partial_r(\Delta\hat{f})-2\hat{f}(r-M)+\Delta \hat{q} \notag\\
B^{\ell}(r)={}&-\tfrac{1}{2}\partial_r(\hat{f})+\hat{q}\notag\\
B^0(r)={}&\partial_r(\hat{q}(r-M))-\tfrac{1}{2}\partial_{rr}^2(\Delta \hat{q})-{2M}r^{-1}\hat{q}+ \partial_r(Mr^{-1}\hat{f})\notag\\
&+r^2(\partial_r (r^{-3}B^r(r))+ r^{-4}B^r(r)).
\end{align}
By taking
$(\hat{f},\hat{q})={}(\tfrac{2(r-2M)(r-3M)}{r^2}, \tfrac{(2r-3M)\Delta}{r^4})$,
the above coefficients are
\begin{align*}
B^t(r)&=0, &B^0(r)&=-9 Mr^{-2}+60M^2r^{-3}-90M^3r^{-4},\notag\\
B^r(r)&=6M\Delta^2r^{-4},
&B^{\ell}(r)&=2(r-3M)^2 r^{-3}.
\end{align*}
Under this choice of $\hat{f}$ and $\hat{q}$,
\begin{align*}
(B^0(r)+B^{\ell}(r)\ell(\ell+1))|\psi_{m\ell}|^2\geq {}&(B^0(r)+2B^{\ell}(r))|\psi_{m\ell}|^2\notag\\
={}&\frac{4r^3-33Mr^2+96M^2r-90M^3}{r^4}|\psi_{m\ell}|^2.
\end{align*}
There exists exactly one real root for the third-order polynomial $4r^3-33Mr^2+96M^2r-90M^3$, and this root is less than $2M$; therefore, there exists a universal constant $c>0$ such that for any $r\geq 2M$,
\begin{align}\label{eq:Moraestibulktermcontrol1}
B(\psi_{m\ell})\geq c\left(\frac{\Delta^2}{r^4}|\partial_r\psi_{m\ell}|^2
+\frac{1}{r}|\psi_{m\ell}|^2
+\frac{(r-3M)^2}{r^3}\ell(\ell+1)|\psi_{m\ell}|^2\right).
\end{align}

Furthermore, if we take $(\hat{f},\hat{q})=(0, -\frac{\Delta(r-3M)^2}{r^7})$ in \eqref{eq:MoraEstidegSchw}, this allows us to control the bulk integral of $|\partial_t \psi_{m\ell}|^2$ part by the bulk integral of the RHS of \eqref{eq:Moraestibulktermcontrol1}.
We  add this estimate to the obtained \eqref{eq:Moraestibulktermcontrol1} to produce positive definite bulk integrals, sum over $m$ and $\ell$, apply identity \eqref{eq:IdenOfEigenvaluesAndAnguDeriSchw}, integrate with respect to the measure $dt^*dr$ over $\{(t^*,r)|0\leq t^*\leq \tau, 2M\leq r<\infty\}$, utilize estimate \eqref{eq:enerestipartialtSchw} to bound the boundary terms at horizon, eventually arriving at
the following Morawetz estimate in global Kerr coordinates
\begin{align}\label{Morawetz-0}
\hspace{4ex}&\hspace{-4ex}\int_{\mathcal{D}(0,\tau)}\left(\tfrac{\Delta^2}{r^6} |\partial_r\psi|^2+\tfrac{1}{r^4} |\psi|^2+\tfrac{(r-3M)^2}{r^2}\left(\tfrac{1}{r^3}|\partial_{t^*}\psi|^2 + \tfrac{1}{r}|\nablaslash\psi|^2\right)\right)\notag\\
 \lesssim {}&
E_{\tau}^T(\psi)+E_{0}^T(\psi)
+\int_{\mathcal{D}(0,\tau)}r^{-2}\left(\left|
\Re\left(X(\psi)\overline{F}\right)\right|
+\left|\Re\left(h\psi\overline{F}\right)\right|\right).
\end{align}

\subsection*{Acknowledgment}
The author is grateful to Steffen Aksteiner, Lars Andersson,  Pieter Blue,  Claudio Paganini and Jinhua Wang for many helpful discussions and comments. Special thanks to Claudio Paganini for double checking some equations by Mathematica. The author is also grateful to the anonymous referee for the helpful suggestions.

\newcommand{\arxivref}[1]{\href{http://www.arxiv.org/abs/#1}{{arXiv.org:#1}}}
\newcommand{\mnras}{Monthly Notices of the Royal Astronomical Society}
\newcommand{\prd}{Phys. Rev. D}
\newcommand{\apj}{Astrophysical J.}

\providecommand{\MR}{\relax\ifhmode\unskip\space\fi MR }
\providecommand{\MRhref}[2]{%
  \href{http://www.ams.org/mathscinet-getitem?mr=#1}{#2}
}
\providecommand{\href}[2]{#2}


\begin{thebibliography}{1}
\bibitem{andersson16decayMaxSchw}
Lars Andersson, Thomas B{\"a}ckdahl, and Pieter Blue, \emph{Decay of solutions
  to the {M}axwell equation on the {S}chwarzschild background}, Classical and
  Quantum Gravity \textbf{33} (2016), no.~8, 085010.

\bibitem{andersson2019stability}
Lars Andersson, Thomas B{\"a}ckdahl, Pieter Blue, and Siyuan Ma,
  \emph{Stability for linearized gravity on the {K}err spacetime}, \href{http://arxiv.org/abs/1903.03859} {arXiv preprint arXiv:1903.03859} (2019).
\bibitem{larsblue15hidden}
Lars Andersson and Pieter Blue, \emph{Hidden symmetries and decay for the wave
  equation on the {K}err spacetime}, Annals of Mathematics \textbf{182} (2015),
  no.~3, 787--853.

\bibitem{larsblue15Maxwellkerr}
Lars Andersson and Pieter Blue, \emph{Uniform energy bound and asymptotics for the {M}axwell field on
  a slowly rotating {K}err black hole exterior}, Journal of Hyperbolic
  Differential Equations \textbf{12} (2015), no.~04, 689--743.

\bibitem{Jinhua17LinGraSchw}
Lars Andersson, Pieter Blue, and Jinhua Wang, \emph{Morawetz estimate for linearized gravity in Schwarzschild}, \href{http://arxiv.org/abs/1708.06943} {arXiv preprint arXiv:1708.06943} (2017).

\bibitem{andersson2016mode}
Lars Andersson, Siyuan Ma, Claudio Paganini, and Bernard~F Whiting, \emph{Mode
  stability on the real axis}, Journal of Mathematical Physics \textbf{58}
  (2017), no.~7, 072501.

\bibitem{bar2007wave}
Christian B{\"a}r, Nicolas Ginoux, and Frank Pf{\"a}ffle, \emph{Wave equations
  on {L}orentzian manifolds and quantization}, European Mathematical Society,
  2007.

\bibitem{blue08decayMaxSchw}
Pieter Blue, \emph{Decay of the {M}axwell field on the {S}chwarzschild
  manifold}, Journal of Hyperbolic Differential Equations \textbf{5} (2008),
  no.~04, 807--856.

\bibitem{blue2003semilinear}
Pieter Blue and Avy Soffer, \emph{Semilinear wave equations on the
  {S}chwarzschild manifold {I}: Local decay estimates}, Advances in
  Differential Equations \textbf{8} (2003), no.~5, 595--614.

\bibitem{bluesoffer09phase}
Pieter Blue and Avy Soffer, \emph{Phase space analysis on some black hole manifolds}, Journal of
  Functional Analysis \textbf{256} (2009), no.~1, 1--90.

\bibitem{boyer:lindquist:1967}
Robert~H. Boyer and Richard~W. Lindquist, \emph{Maximal analytic extension of
  the {K}err metric}, J. Mathematical Phys. \textbf{8} (1967), 265--281.

\bibitem{Carter1968Separability}
Brandon Carter, \emph{Hamilton-{Jacobi} and {Schr{\"o}dinger} separable
  solutions of {Einstein}'s equations}, Communications in Mathematical Physics
  \textbf{10} (1968), no.~4, 280--310.

\bibitem{chandrasekhar1975linearstabSchw}
Subrahmanyan Chandrasekhar, \emph{On the equations governing the perturbations
  of the {S}chwarzschild black hole}, Proceedings of the Royal Society of
  London A: Mathematical, Physical and Engineering Sciences, vol. 343, The
  Royal Society, 1975, pp.~289--298.

\bibitem{DRG16linearstabSchw}
Mihalis Dafermos, Gustav Holzegel, and Igor Rodnianski, \emph{The linear
  stability of the {S}chwarzschild solution to gravitational perturbations},
  \href{http://arxiv.org/abs/1601.06467} {arXiv preprint arXiv:1601.06467} (2016).

\bibitem{Dafermos2019TMEKerr}
Mihalis Dafermos, Gustav Holzegel, and Igor Rodnianski, \emph{Boundedness and decay for the {T}eukolsky equation on {K}err spacetimes I: The case $|a|\ll M $}, Annals of PDE 5.1 (2019): 2.

\bibitem{dafrod09red}
Mihalis Dafermos and Igor Rodnianski, \emph{The red-shift effect and radiation
  decay on black hole spacetimes}, Communications on Pure and Applied
  Mathematics \textbf{62} (2009), no.~7, 859--919.

\bibitem{dafermos2010decay}
Mihalis Dafermos and Igor Rodnianski, \emph{Decay for solutions of the wave equation on {Kerr} exterior
  spacetimes {I}-{II}: The cases $|a|\ll m$ or axisymmetry},
  \href{http://arxiv.org/abs/1010.5132}{arXiv preprint arXiv:1010.5132} (2010).

\bibitem{dafermos2011bdedness}
Mihalis Dafermos and Igor Rodnianski, \emph{A proof of the uniform boundedness of solutions to the wave
  equation on slowly rotating {K}err backgrounds}, Inventiones mathematicae
  \textbf{185} (2011), no.~3, 467--559.

\bibitem{dafermos2016decay}
Mihalis Dafermos, Igor Rodnianski, and Yakov Shlapentokh-Rothman, \emph{Decay for solutions of the wave equation on K{}err exterior spacetimes III: {T}he full subextremal case $|a|< M$}, Annals of Mathematics (2016), 787-913.

\bibitem{fackerell:ipser:EM}
E.~D. {Fackerell} and J.~R. {Ipser}, \emph{{Weak Electromagnetic Fields Around
  a Rotating Black Hole}}, Phys. Rev. D. \textbf{5} (1972), 2455--2458.

\bibitem{finster2016linear}
Felix Finster and Joel Smoller, \emph{Linear stability of the non-extreme
  {K}err black hole}, \href{http://arxiv.org/abs/1606.08005}{arXiv preprint
  arXiv:1606.08005} (2016).

\bibitem{ghanem2014decayMaxSchw}
Sari Ghanem, \emph{On uniform decay of the {M}axwell fields on black hole
  space-times}, \href{http://arxiv.org/abs/1409.8040} {arXiv preprint arXiv:1409.8040} (2014).

\bibitem{Giorgi2019linearRNsmallcharge}
Elena Giorgi, \emph{The linear stability of {R}eissner-{N}ordstr\" om spacetime for small charge},
\href{http://arxiv.org/abs/1904.04926} {arXiv preprint arXiv:1904.04926} (2019).

\bibitem{Giorgi2019linearRNfullcharge}
Elena Giorgi, \emph{The linear stability of {R}eissner-{N}ordstr\" om spacetime: the full subextremal range}, \href{http://arxiv.org/abs/1910.05630} {arXiv preprint arXiv:1910.05630} (2019).

\bibitem{hafner2019linear}
Dietrich Häfner, Peter Hintz, and András Vasy. \emph{Linear stability of slowly rotating {K}err black holes},
\href{http://arxiv.org/abs/1906.00860} {arXiv preprint arXiv:1906.00860} (2019).

\bibitem{HHtetrad72}
S.~W. {Hawking} and J.~B. {Hartle}, \emph{Energy and angular momentum flow
  into a black hole}, Communications in Mathematical Physics \textbf{27}
  (1972), 283--290.

\bibitem{Hung2018linearwavegaugeodd}
Pei-Ken Hung, \emph{The linear stability of the {S}chwarzschild spacetime in the harmonic gauge: odd part},
\href{http://arxiv.org/abs/1803.03881} {arXiv preprint arXiv:1803.03881} (2018).

\bibitem{Hung2019linearwavegaugeeven}
Pei-Ken Hung, \emph{The linear stability of the {S}chwarzschild spacetime in the harmonic gauge: even part}, \href{http://arxiv.org/abs/1909.06733} {arXiv preprint arXiv:1909.06733} (2019).

\bibitem{Hung2017linearstabSchw}
Pei-Ken Hung, Jordan Keller, and Mu-Tao Wang, \emph{Linear stability of {S}chwarzschild spacetime: The {C}auchy problem of metric coefficients}, \href{http://arxiv.org/abs/1702.02843} {arXiv preprint arXiv:1702.02843} (2017).

\bibitem{kla2015globalstabwavemapKerr}
Alexandru~D Ionescu and Sergiu Klainerman, \emph{On the global stability of the
  wave-map equation in {K}err spaces with small angular momentum}, Annals of
  PDE \textbf{1} (2015), no.~1, 1.

\bibitem{Johnson2018linear}
Thomas~W Johnson, \emph{The linear stability of the {S}chwarzschild solution to gravitational perturbations in the generalised wave gauge},
\href{http://arxiv.org/abs/1810.01337} {arXiv preprint arXiv:1810.01337} (2018).

\bibitem{kaywald87Schw}
Bernard~S Kay and Robert~M Wald, \emph{Linear stability of {S}chwarzschild
  under perturbations which are non-vanishing on the bifurcation $2$-sphere},
  Classical and Quantum Gravity \textbf{4} (1987), no.~4, 893.

\bibitem{Kinnersley1969tetradForTypeD}
William Kinnersley, \emph{Type {D} vacuum metrics}, Journal of Mathematical
  Physics \textbf{10} (1969), no.~7, 1195--1203.

\bibitem{klainermanszeftel2017Schw}
Sergiu Klainerman and J\'{e}r\'{e}mie Szeftel, \emph{Global nonlinear stability of {S}chwarzschild spacetime under polarized perturbations},
\href{http://arxiv.org/abs/1711.07597} {arXiv preprint arXiv:1711.07597} (2017).

\bibitem{Ma17spin2Kerr}
Siyuan Ma, \emph{Uniform energy bound and {M}orawetz estimate for extreme
  components of spin fields in the exterior of a slowly rotating {K}err black
  hole {II}: linearized gravity}, \href{http://arxiv.org/abs/1708.07385} {arXiv preprint arXiv:1708.07385} (2017).

\bibitem{metcalfe2014PWdecayMaxBH}
Jason Metcalfe, Daniel Tataru, and Mihai Tohaneanu, \emph{Pointwise decay for
  the {M}axwell field on black hole space-times}, \href{http://arxiv.org/abs/1411.3693} {arXiv preprint
  arXiv:1411.3693} (2014).

\bibitem{moncrief74gravitational}
Vincent Moncrief, \emph{Gravitational perturbations of spherically symmetric
  systems. {I}. the exterior problem}, Annals of Physics \textbf{88} (1974),
  no.~2, 323--342.

\bibitem{morawetz1968time}
Cathleen~S Morawetz, \emph{Time decay for the nonlinear {K}lein-{G}ordon
  equation}, Proceedings of the Royal Society of London A: Mathematical,
  Physical and Engineering Sciences, vol. 306, The Royal Society, 1968,
  pp.~291--296.

\bibitem{newmanpenrose62}
Ezra Newman and Roger Penrose, \emph{An approach to gravitational radiation by
  a method of spin coefficients}, Journal of Mathematical Physics \textbf{3}
  (1962), no.~3, 566--578.

\bibitem{newmanpenrose63errata}
Ezra Newman and Roger Penrose, \emph{Errata: an approach to gravitational radiation by a method of
  spin coefficients}, Journal of Mathematical Physics \textbf{4} (1963), no.~7,
  998--998.

\bibitem{Fede2016MaxwellSchw}
Federico Pasqualotto, \emph{The spin $\pm 1$ {T}eukolsky equations and the
  {M}axwell system on {S}chwarzschild}, \href{http://arxiv.org/abs/1612.07244}{arXiv preprint arXiv:1612.07244} (2016).

\bibitem{ReggeWheeler1957}
Tullio Regge and John~A Wheeler, \emph{Stability of a {Schwarzschild}
  singularity}, Physical Review \textbf{108} (1957), no.~4, 1063.

\bibitem{shlapentokh-rothman2015}
Yakov Shlapentokh-Rothman, \emph{Quantitative mode stability for the wave equation on the {K}err spacetime},  Annales Henri Poincaré, vol. 16. No. 1. Springer Basel, 2015.

\bibitem{starobinsky1973amplification}
A.~A. {Starobinsky} and S.~M. {Churilov}, \emph{Amplification of
  electromagnetic and gravitational waves scattered by a rotating black hole},
  Zh. Eksp. Teor. Fiz \textbf{65} (1973), no.~3.

\bibitem{sterbenz2015decayMaxSphSym}
Jacob Sterbenz and Daniel Tataru, \emph{Local energy decay for {M}axwell fields
  part {I}: Spherically symmetric black-hole backgrounds}, International
  Mathematics Research Notices \textbf{2015} (2015), no.~11.

\bibitem{tataru2011localkerr}
Daniel Tataru and Mihai Tohaneanu, \emph{A local energy estimate on {K}err
  black hole backgrounds}, International Mathematics Research Notices
  \textbf{2011} (2011), no.~2, 248--292.

\bibitem{Teu1972PRLseparability}
S.~A. {Teukolsky}, \emph{Rotating black holes: Separable wave equations for
  gravitational and electromagnetic perturbations}, Physical Review Letters
  \textbf{29} (1972), no.~16, 1114.

\bibitem{Teukolsky1973I}
S.~A. {Teukolsky}, \emph{{Perturbations of a Rotating Black Hole. I.
  Fundamental Equations for Gravitational, Electromagnetic, and Neutrino-Field
  Perturbations}}, Astrophysical J. \textbf{185} (1973), 635--648.

\bibitem{TeuPress1974III}
S.~A. {Teukolsky} and W.~H. {Press}, \emph{{Perturbations of a rotating black
  hole. {III} - Interaction of the hole with gravitational and electromagnetic
  radiation}}, \apj \textbf{193} (1974), 443--461.

\bibitem{Vishveshwara1970stability}
C.~V. Vishveshwara, \emph{Stability of the {Schwarzschild} metric}, Physical
  Review D \textbf{1} (1970), no.~10, 2870.

\bibitem{whiting1989mode}
Bernard~F Whiting, \emph{Mode stability of the {K}err black hole}, Journal of
  Mathematical Physics \textbf{30} (1989), no.~6, 1301--1305.

\bibitem{Zerilli1970evenparity}
Frank~J Zerilli, \emph{Effective potential for even-parity {Regge-Wheeler}
  gravitational perturbation equations}, Physical Review Letters \textbf{24}
  (1970), no.~13, 737.
\end{thebibliography}
\end{document}